\documentclass[11pt,a4paper,reqno]{amsart}
\usepackage{amsthm,amsfonts,amsmath,amssymb,amsrefs,bookmark,appendix}
\usepackage{amsxtra, mathrsfs}
\usepackage{colordvi}
\usepackage[usenames,dvipsnames]{color}
\usepackage{dsfont}
\usepackage{graphics}
\usepackage{tikz}
\usetikzlibrary{cd,arrows}

\newtheorem{theorem}{Theorem}
\newtheorem{lemma}[theorem]{Lemma}
\newtheorem{corollary}[theorem]{Corollary}
\newtheorem{proposition}[theorem]{Proposition}
\newtheorem{rem}[theorem]{\bf Remark}
\newtheorem*{remnn}{\bf Remark}

\newtheorem*{notnn}{\bf Notation}

\DeclareMathOperator{\aux}{aux}

\DeclareMathOperator{\supp}{supp}
\DeclareMathOperator{\loc}{loc}
\DeclareMathOperator{\dist}{dist}
\DeclareMathOperator{\dom}{dom}

\DeclareMathOperator{\sing}{sing}

\DeclareMathOperator{\link}{Link}

\DeclareMathOperator{\Sf}{sf}

\DeclareMathOperator{\per}{per}

\DeclareMathOperator{\Wr}{Wr}
\DeclareMathOperator{\clos}{clos}
\DeclareMathOperator{\model}{mod}
\DeclareMathOperator{\full}{full}
\DeclareMathOperator{\spann}{span}
\DeclareMathOperator{\ans}{ans}
\renewcommand{\phi}{\varphi}
\newcommand{\eps}{\varepsilon}

\newcommand{\norm}[1]{\lVert #1 \rVert}

\newcommand{\cip}[2]{\langle #1, #2 \rangle}
\newcommand{\wt}[1]{\widetilde{#1}}
\newcommand{\nn}{\nonumber}

\newcommand{\D}{\mathbb{D}}
\newcommand{\R}{\mathbb{R}}
\newcommand{\C}{\mathbb{C}}
\newcommand{\Z}{\mathbb{Z}}
\newcommand{\T}{\mathbb{T}}
\newcommand{\N}{\mathbb{N}}
\renewcommand{\S}{\mathbb{S}}

\newcommand{\cA}{\mathcal{A}}
\newcommand{\cB}{\mathcal{B}}

\newcommand{\cD}{\mathcal{D}}

\newcommand{\cG}{\mathcal{G}}
\newcommand{\cH}{\mathcal{H}}
\newcommand{\cI}{\mathcal{I}}

\newcommand{\cN}{\mathcal{N}}

\newcommand{\cT}{\mathcal{T}}

\newcommand{\rT}{\mathrm{T}}
\renewcommand{\d}{\mathrm{d}}

\newcommand{\sD}{\mathscr{D}}

\newcommand{\sT}{\mathscr{T}}

\newcommand{\uS}{\underline{S}}

\newcommand{\Sd}{\mathcal{S}_{\mathrm{disc}}}

\newcommand{\bA}{\boldsymbol{A}}
\newcommand{\bB}{\boldsymbol{B}}

\newcommand{\bN}{\boldsymbol{N}}
\newcommand{\bS}{\boldsymbol{S}}
\newcommand{\bT}{\boldsymbol{T}}

\newcommand{\bx}{\boldsymbol{x}}

\newcommand{\bv}{\boldsymbol{v}}
\newcommand{\bp}{\boldsymbol{p}}
\newcommand{\bt}{\boldsymbol{t}}

\newcommand{\bg}{\gamma}

\newcommand{\Tf}{[0,1]_{\per}}
\newcommand{\bap}{a_{\parallel}}

\usepackage{xspace}

\newcommand{\dint}{\ensuremath{\displaystyle\int}}
\newcommand{\diint}{\ensuremath{\displaystyle\iint}}

\begin{document}

\title{Spectral flow of Dirac operators with magnetic cable knot}

\author[J. Sok]{J\'er\'emy Sok}
\address[J. Sok]{University of Basel, Switzerland} 
\email{jeremyvithya.sok@unibas.ch}

\author[J. P. Solovej]{Jan Philip Solovej}
\address[J. P. Solovej]{University of Copenhagen, Denmark} 
\email{solovej@math.ku.dk}



\date{}

\thanks{The authors acknowledge support from the ERC grant Nr.\ 321029 ``The mathematics of the structure of matter" and
VILLUM FONDEN through the QMATH Centre of Excellence grant. nr.\ 10059.
J.S. also acknowledges support from the Swiss National Science Foundation (SNF) through Grant No. 200021-169646. This work was completed at
the Mittag-Leffler Institute during the program \emph{Spectral Methods in Mathematical Physics}, J.S. is grateful to this institute for its kind hospitality.
}

\begin{abstract}
	We study the spectral flow of Dirac operators with magnetic links on $\S^3$. These are generalisations of Aharonov-Bohm solenoids
	where the magnetic fields contain finitely many field lines coinciding with the components of a link, the flux of each exhibiting
	the same $2\pi$-periodicity as A-B solenoids. We study the spectral flow of the loop obtained as tuning the flux from $0$
	to $2\pi$ in the case of only one field line: we relate the spectral flows obtained for one given knot and its cable knots, and obtain
	that torus knots have trivial spectral flow. The operators are studied in their Coulomb gauge in $\R^3$ (seen as a chart of $\S^3$
	through the stereographic projection), which is simply given by the Biot and Savart formula.
\end{abstract}

\maketitle

\tableofcontents

\section{Introduction and main results}

In a series of three papers \cite{dirac_s3_paper1,dirac_s3_paper2,dirac_s3_paper3}, we introduced Dirac operators
with singular magnetic fields supported on links (both in $\R^3$ and $\S^3$) and studied their spectral properties, in particular their
kernels. These magnetic fields only have finitely many field lines which form a link, that is a one-dimensional manifold 
diffeomorphic to finitely many copies of $\S^1$, and as such are called ``magnetic links". They are the generalizations of the celebrated Aharonov-Bohm (A-B)
solenoids which are magnetic fields supported on straight lines in $\R^3$ (note that we can conformally map straight lines to circles). 
One could discuss what the appropriate definition of singular Dirac operators is for these singular magnetic fields. Our definition is natural 
as the singular operators defined in \cite{dirac_s3_paper1} are limits (in the norm-resolvent sense) of Dirac operators with smooth magnetic fields, see \cite{dirac_s3_paper3}.

As for A-B solenoids, the flux carried by each field line of a magnetic link exhibits a $2\pi$-periodicity: if the flux is a multiple of $2\pi$,
the corresponding Dirac operator is -- depending on the chosen gauge -- equal or unitarily equivalent to the one with this field line removed. 
So by tuning a flux from $0$ to $2\pi$ while keeping fixed the geometry of the link, we obtain non-trivial loops of Dirac operators. 
The paper \cite{dirac_s3_paper2} was devoted to the study of the \emph{spectral flow} of such loops, that is the number of eigenvalues crossing $0$ along the loop, counted algebraically. 
On $\S^3$, the Dirac operators have been shown to be self-adjoint and to have discrete spectrum \cite{dirac_s3_paper1}. 
For a link, the spectral flow is generically non-zero, and its value varies with the linking numbers of its connected components and on their so-called writhe, see \eqref{eq:def_writhe}-\eqref{eq:link}. 
Also there exists choice of fluxes, for which the kernel of the Dirac operator is non-zero.

In this paper, we give a partial answer to a question left open in \cite{dirac_s3_paper2}. 
Consider a smooth closed curve, say $\gamma\subset \R^3$, and the corresponding loop of magnetic knots $(2\pi\alpha [\gamma])_{\alpha\in \Tf}$, 
where $\Tf$ denotes the segment $[0,1]$ with $0$ and $1$ identified. We know \cite{dirac_s3_paper2}  that provided the writhe $\Wr(\gamma)$ 
satisfies $\tfrac{1}{2}\Wr(\gamma)\notin\tfrac{1}{2}+\Z$, then the spectral flow  of the corresponding loop of Dirac operators is well-defined. 
It will be written $\Sf(\gamma)$ throughout this paper. If $\gamma$ is isotopic to a circle,
the spectral flow is equal to $\lfloor\tfrac{1}{2}(1-\Wr(\gamma)) \rfloor$, else its value is unknown. 
We only know how the spectral flow changes under the deformation of the curve: it jumps by $\pm1$ across the singular set $\tfrac{1}{2}(1-\Wr(\gamma))\in\Z$. 
By defining the spectral flow of an isotopy class of closed curves in $\R^3$ to be the one of its elements with null writhe, we trivially obtain a knot invariant. 
In this paper we investigate the spectral flow for other knots, in particular for torus knots. More precisely, we study the relation between the spectral flow of a given knot and those of its cable knots
(see Section~\ref{sec:cable_knot}). We obtain a simple formula in Theorem~\ref{thm:main}: the spectral flow of a cable knot of $\gamma$ is a multiple of
that of $\gamma$. This gives $0$ for all torus knots. We conjecture the spectral flow to be the trivial invariant.

\medskip
\paragraph{\textbf{Overview}}

The spectral flow was first introduced by Atiyah, Patodi and Singer in
\cites{ASP1,ASP3} to obtain an index Theorem for elliptic operators on
vector bundles over compact manifolds with boundary. 
The general formula for the spectral flow of smooth open paths $(D_t)_{t\in
  [0,1]}$ of self-adjoint elliptic operators involves the eta
invariant of the endpoints $\eta_{D_0}(0)$ and $\eta_{D_1}(0)$, and
the dimension of their respective kernels.  We refrain from giving an
overview of the (vast) literature; for more details on the index
theorems and the eta invariant we refer the reader to e.g.
\cites{AS68,ASP1,ASP2,ASP3,getzler,Melrose,Grubb05}.  We would like to
emphasize that the computation of the eta invariant is a
difficult problem: we refer the reader to the survey \cite{Goette12}.

Similarly, it is impossible to give a complete overview on the works
devoted to the spectral flow, and we only mention
results which are relevant for the present discussion.  In
\cite{Philips_spectral_flow}, the author gives an equivalent
definition of the spectral flow using the functional calculus. It is
shown that the spectral flow of a continuous path $(D_t)_{t\in [0,1]}$
of bounded, self-adjoint, Fredholm operators on a separable Hilbert
space corresponds (at least in a small interval $[t_1,t_2]$) to the
difference
$$
	\Sf\left[(D_t)_{t\in [t_1,t_2]}\right] = 
        \dim \mathds{1}_{[0,a]}\big(D_{t_2}\big)-\dim\mathds{1}_{[0,a]}\big(D_{t_1}\big),
$$ 
where $a>0$ is a given spectral level for which
$(\mathds{1}_{[-a,a]}(D_t))_{t\in [t_1,t_2]}$ is continuous. The
spectral flow of the whole path is obtained by subdividing $[0,1]$
into a finite family of intervals $[t_{i},t_{i+1}]$ for which we can
apply the above formula and adding all the contributions.  The concept
of spectral flow was then extended to unbounded, self-adjoint,
Fredholm operators in e.g. \cite{BBLP05}.  The restriction of the
spectral flow to loops constitutes a homotopy invariant in 
the gap topology, the topology of the norm-resolvent convergence.
As done in e.g. \cite{Wahl08}, the spectral flow 
can be extended to weaker topologies. 
Other topologies are considered in the litterature like the Riesz topology,
but these are stronger than the gap topology, we refer the reader to
\cites{BBLP05, Nicolaescu07} and references therein for more details.

The paths of Dirac operators with magnetic links that we studied in \cite{dirac_s3_paper2} are not 
continuous in the gap topology, due to the possible collapse of eigenfunctions
(by concentration of their mass on a set of Lebesgue measure $0$).  
As long as this only happens away from the spectral level zero it does
not affect Philipps' definition of the spectral flow \cite{Philips_spectral_flow}:
it is enough to control the eigenvalues in a spectral window around zero. This remark is
the starting point of \cite{Wahl08} to extend the spectral flow.
Here, the condition $\tfrac{1}{2}\Wr(\gamma)\notin \tfrac{1}{2}+\Z$ is precisely related to the collapse of eigenfunctions at the spectral level $0$.

In some cases, the spectral flow admits integral representations, see \emph{e.g.} \cite{GetzlerOdd,Sf_integral,GorLes}, but their evaluations 
can be difficult. The computation of the spectral flow for (two-dimensional) Dirac operators with local elliptic boundary condition is the purpose of \cite{Prokh}, and 
the case of several A-B solenoids is studied in more details in \cite{KatNaz}. One A-B solenoid with flux $\Phi$, say at $0\in\C$,
is modelled by a magnetic potential $\tfrac{\Phi}{2\pi}\nabla\theta$, which comes together with a forbidden region $R_0$ around $0$, inside which the spinors
are not defined and on whose boundary is prescribed a Berry-Mondragon-type boundary conditions \cite{BerryMondragon}*{Eq.~(33)}.

In this situation, the spectral flow has a topological description and its study enters the framework of the Atiyah-Singer index theorem for Dirac operators on compact manifolds with boundaries.
But it is different to the situation at hand: as explained above, the spectral flow is not purely topological (there is a dependence on the \emph{writhes} of the field lines).

\begin{remnn}

\noindent  -- Here, to simplify calculations, we will mainly work in $\R^3$ seen as a chart for $\S^3$ minus a pole through the stereographic projection.
Let us make a choice and consider the stereographic projection with inverse:
\begin{equation}\label{eq:conform}
\bx\in\R^3\mapsto \Big(\frac{2(x_1+ix_2)}{|\bx|^2+1},\frac{2x_3+i(|\bx|^2-1)}{|\bx|^2+1}\Big)\in\S^3\subset\C^2.
\end{equation}
The conformal factor which switches from the flat metric to the pullback of that of $\S^3$ is $\Omega(\bx):=\tfrac{2}{1+|\bx|^2}$.

\noindent  -- A smooth curve $\gamma$ in $\R^3$ is identified, up to fixing a basepoint $\bp_0\in\bg$, with its
arclength parametrisation $\gamma:(\R/(\ell\Z))\to \R^3$ in the flat metric of $\R^3$. The set $\R/(\ell\Z)$ is often denoted by $\T_\ell$.

\noindent  
-- We emphasize that it is possible to directly work in $\S^3$ 
(as it has been done in \cite{dirac_s3_paper1,dirac_s3_paper2}).
For instance, we will use the gauge given by the Biot and Savart law in $\R^3$ for magnetic knots, and
a Biot and Savart law has been established in $\S^3$ (see the very nice formulas in \cite{DeTGlu}).
\end{remnn}

\begin{notnn}
-- 
We will use the musical isomorphism $\flat$ (with inverse $\sharp$), which transforms a vector into a one-form
through the \emph{flat} metric of $\R^3$, mainly to keep track of the geometry. Moreover, for a vector field $\mathbf{V}$ on $\R^3$, we will
write $\sigma\cdot \mathbf{V}$ for $\sum_{k=1}^3V_k\sigma_k$ where the $\sigma_k$'s denote the Pauli matrices. 

-- We may use the words knot and link for specific realizations and not for the isotopy classes. To distinguish them, we overline the letter
 to designate the latter: $\gamma$ is a (smooth) curve while $\overline{\gamma}$ is its isotopy class in $\S^3$.
\end{notnn}

\subsection{Definition of the Dirac operators}\label{sec:def_dirac_op_intro}
Let us quickly define the singular Dirac operators in the Biot and Savart gauge\footnote{The gauge transformation switching to the singular gauge of \cite{dirac_s3_paper1}
is given in Section \ref{sec:gauge_choice}.}. The reader can find a motivation for the definition in \cite{dirac_s3_paper1,dirac_s3_paper3}.
We recall that for a curve $\bg$, the corresponding Biot and Savart law (for unit current or flux) is:
\begin{equation}\label{eq:BS}
	\bA_{\bg}(\bx):=\dfrac{1}{4\pi}\dint_{\gamma}\d\mathbf{r}\times \frac{\bx-\mathbf{r}}{|\bx-\mathbf{r}|^3}.
\end{equation}
 The one-form $\bA_{\bg}^\flat$ can be extended to the pole by $0$ in the $\S^3$-metric
 since we have $\lim_{\bx\to+\infty}|\bA_{\bg}(\bx)|(1+|\bx|^2)=0$. 
 We recall the asymptotic expansion:
 \begin{equation*}
  \bA_{\bg}(\bx)=\frac{1}{4\pi|\bx|^3}\int_{\gamma}\d \mathbf{r}\times \Big(\frac{3\cip{\bx}{\mathbf{r}}}{|\bx|}\frac{\bx}{|\bx|}-\mathbf{r}\Big)+\underset{\bx\to\infty}{\mathcal{O}}(|\bx|^{-4}).
 \end{equation*}

Consider a link $\gamma:=\dot{\cup}_{1\le k\le K}\gamma_k$, and a collection of fluxes
$0<2\pi\alpha_k<2\pi$. The B.S. magnetic potential is $\bA_{B.S.}:=\sum_{k=1}^K 2\pi\alpha_k \bA_{\gamma_k}$. 
Let $\bt:B_\eps[\gamma]\to \C^2$ be a smooth extension of the unit tangent vectors $\dot{\gamma}_k$ on a small tubular neighborhood of the link,
say by setting $\bt(\bx):=\bt(\gamma(s))\in\R^3$, where $\gamma(s)$ is the projection of $\bx$ onto $\gamma$.
Let $P_+:B_\eps[\gamma]\to \mathrm{End}(\C^2)$ be the pointwise projection onto the eigenspace $\ker(\sigma\cdot\bt-1)$, and $\chi:\R^3\to [0,1]$
be a smooth localization which equals $1$ on $B_{\eps/2}[\gamma]$, and $0$ outside $B_\eps[\gamma]$.

The Dirac operator in the flat metric is formally given by $\sigma\cdot(-i\nabla^{\R^3}+\bA_{B.S.})$ acting on $L^2(\R^3,\C)^2$. 
But we deal with operators on $\S^3$ through \eqref{eq:conform}, the formal Dirac operator is then $\cD_{\text{form}}:=\Omega^{-2}\sigma\cdot(-i\nabla^{\R^3}+\bA_{B.S.})\Omega$, 
acting on $L_{\Omega}^2:=L^2(\R^3,\Omega^3\d x)^2$ \cite{ErdSol01}*{Section~4}.
This Hilbert space is isometric to $L^2(\S^3)^2$: we will define the Dirac operator on this space. 

The minimal operator $\cD_{\bA_{B.S.}}^{(\min)}$ is defined as the graph norm closure in $L_{\Omega}^2$ of 
$C_0^\infty(\R^3\setminus \gamma)^2$ with respect to $\cD_{\text{form}}$. It is a symmetric operator.
The Dirac operator $\cD_{\bA_{B.S.}}$ is the self-adjoint extension of $\cD_{\bA_{B.S.}}^{(\min)}$ with domain:
\[
	\dom(\cD_{\bA_{B.S.}})=\big\{\psi\in\dom\,(\cD_{\bA_{B.S.}}^{(\min)})^*,\ P_+\chi \psi\in \dom(\cD_{\bA_{B.S.}}^{(\min)}) \big\}.
\]
A more precise description of its domain is given in \cite{dirac_s3_paper1} (see also Section~\ref{sec:desc_dom}).
The definition does not depend on the choice of $\chi$. There is no problem at infinity since in dimension $3$ the space $H^1$
of a ball and $H^1$ of this ball minus a point are the same.

\subsection{Cable knots}\label{sec:cable_knot}

\subsubsection{Definition}
Let $\gamma_0:\R/(\ell\Z)\to \R^3$ be a knot, identified with its arclength parametrisation. 
A cable knot $\gamma$ associated with $\gamma_0$ is an isotopy class of \emph{simple} closed curves which admits a 
realization $\widetilde{\gamma}$ of the form:
\begin{equation}\label{eq:def_gamma_eta}
	\widetilde{\gamma}(s)=\gamma_0([s])+\eta U(s),\ 0\le s\le N\ell,
\end{equation}
where
\begin{itemize}
	\item $N\ge 2$ is an integer and $\eta>0$ is sufficiently small so that the tubular neighborhood $\overline{B}_\eta[\gamma_0]$ is diffeomorphic to a solid torus,
	\item $[s]=s+\Z\ell\in \T_{\ell}:=\R/(\ell\Z)$ denotes the modulo class of $s$,
	\item for each $s$, $U(s)$ is a \emph{unit} vector orthogonal to $\dot{\gamma}_0([s])$.
\end{itemize}
The constraint on $\widetilde{\gamma}$ of being simple is given by the condition $U(s_1)\neq U(s_2)$ for each $s_1\neq s_2$ with $[s_1]=[s_2]$.

The cable knots of the unknot are called torus knots as they can be embedded in a torus (by definition).
For instance, the right-handed trefoil is the torus knot $(2,3)$, and corresponds to a periodic motion on the torus 
in which, while running along the unknot $N=2$ times, the vector $U$ makes $M=3$ turns around
the tangent vector $\dot{\gamma}_0$ in the positive direction (that is such that the linking number \eqref{eq:link} equals
 $\link(\gamma_0,\widetilde{\gamma})=3$), see Figure~\ref{fig:trefoil}.

In the literature, one also finds the following equivalent definition: given $\gamma_0$, a cable knot of $\gamma_0$ is the image of
a torus knot through a diffeomorphism mapping a solid torus $\overline{B}_\eta[\mathrm{circle}]$ onto a neighborhood of $\gamma_0$
which maps the circle onto $\gamma_0$.

As an isotopy class, a cable knot is characterized by its base $\gamma_0$ and two integers: the number of turns $N\ge 2$ run by $\widetilde{\gamma}$
along $\gamma_0$ and the linking number $\link(\gamma_0,\widetilde{\gamma})$. Since $\wt{\gamma}$ is a single curve, $N$ and $M$ are coprime.

\smallskip

The main idea of the paper is to consider the loop of Dirac operators corresponding to 
$\alpha\in [0,1]\mapsto 2\pi \alpha [\gamma_{\eta}]$ and to take the limit $\eta\to 0$.

\smallskip

For analytical purposes, very special realizations of cable knots will be considered (we will pick special curves $\gamma_0$ and vector maps $U:\T_{N\ell}\to \R^3$ in \eqref{eq:def_gamma_eta}).
To explain them, we need to introduce the notion of writhe explained in the next part. 

\begin{rem}\label{rem:on_torus_knots}
	We emphasize that the $(N,M)$-torus knot and the $(M,N)$-torus knot define the same isotopy class on $\S^3$. 
	To see it, it suffices to make a rotation in $\S^3$ in such a way that in $\R^3$ (through \eqref{eq:conform}),
	the vertical line passing through the origin is mapped onto the horizontal unit disk around the origin and vice-versa.
	Let $\mathbf{e}_2=(0,1,0)\in\R^3$, we can consider for instance the following conformal map:
	\[
		\mathbf{x}\mapsto \Big(\frac{2x_3}{|\mathbf{e}_2-\mathbf{x}|^2},\frac{|\mathbf{x}|^2-1}{|\mathbf{e}_2-\mathbf{x}|^2},\frac{2x_1}{|\mathbf{e}_2-\mathbf{x}|^2}\Big).
	\]
	It corresponds to the switch operator $(z_0,z_1)\mapsto (z_1,z_0)$ on $\S^3$ and the rotation of angle $\pi$ in the plane spanned by $(1+i,-1-i)$ and $(1-i,-1+i)$.
	Note that this only applies to torus knots \emph{a priori} (when the base $\gamma_0$ is the circle).
\end{rem}

\subparagraph{\emph{Terminology}}
Consider $(\widetilde{\gamma},\gamma_0)$ a realization of a cable knot of $\gamma_0$ as described in \eqref{eq:def_gamma_eta}, and let $M=\link(\widetilde{\gamma},\gamma_0)$, we will call the $(N,M)$-torus knot 
the pattern of $\widetilde{\gamma}$. 
We will also say that $\widetilde{\gamma}$ is the $(N,M)$-cable knot upon $\gamma$, and that $(\widetilde{\gamma},\gamma_0)$ is a cable representation of $\widetilde{\gamma}$.

\begin{rem}
	Throughout the paper we will often shorten $[s]$ to $s$: the class modulo $\ell$ will often be implicitly taken.
\end{rem}

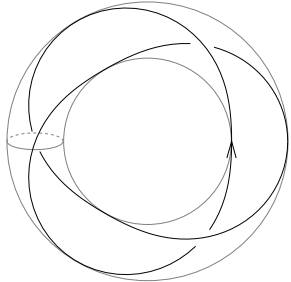
\begin{figure}[!ht]
	\resizebox{0.3\textwidth}{!}{
	\begin{tikzpicture}
		\draw[color=gray!100] (0,0) circle [radius=2.99];
		\draw[color=gray!100]  (0,0) circle [radius=5.01];
		\draw[color=gray!100,dashed] plot [domain=0:3.14] ({-4.01+cos(\x r)},{0.3*sin(\x r)});
		\draw[color=gray!100] plot [domain=3.14:6.28] ({-4.01+cos(\x r)},{0.3*sin(\x r)});
	\draw[thick] plot[samples=20,domain=0:0.48]  ({(4+cos(3*\x r))*cos(2*\x r)},{(4+cos(3*\x r))*sin(2*\x r)});
	\draw[thick] plot[samples=50,domain=0.58:2.57]  ({(4+cos(3*\x r))*cos(2*\x r)},{(4+cos(3*\x r))*sin(2*\x r)});
	\draw[thick] plot[samples=50,domain=2.66:4.67]  ({(4+cos(3*\x r))*cos(2*\x r)},{(4+cos(3*\x r))*sin(2*\x r)});
	\draw[thick] plot[samples=50,domain=4.76:2*pi]  ({(4+cos(3*\x r))*cos(2*\x r)},{(4+cos(3*\x r))*sin(2*\x r)});
	\draw[thick] (2.83,-0.6) -- (3,0) -- (3.16,-0.6);
	\end{tikzpicture}
	}
	\label{fig:trefoil} 
	\caption{The right-handed trefoil as the $(2,3)$-torus knot.}
\end{figure}

\subsubsection{Writhe, linking number and Seifert surface}

For two non-intersecting oriented knots $\gamma_1$ and $\gamma_2$,
their \underline{\emph{linking number}} $\link(\gamma_1,\gamma_2)\in\Z$ is:
\begin{equation}\label{eq:link}
	\link(\gamma_1,\gamma_2)=\frac{1}{4\pi}\int_{\gamma_1}\int_{\gamma_2}
		\left\langle\d\mathbf{r}_1\times \d \mathbf{r}_2,
			\frac{\mathbf{r}_1-\mathbf{r}_2}{|\mathbf{r}_1-\mathbf{r}_2|^3}\right\rangle.
\end{equation}
This knot invariant is the circulation along one knot of the Biot and Savart vector field with unit flux
relative to the other (seen as a magnetic potential). Eq.~\eqref{eq:link} is the Gauss linking number formula.
The writhe of $\gamma$ is given by taking $\gamma_1=\gamma_2=\gamma$ in the formula (there is in fact no singularity in the integrand of \eqref{eq:def_writhe}):
\begin{equation}\label{eq:def_writhe}
	\mathrm{Wr}(\gamma):=\frac{1}{4\pi}\diint_{\gamma\times\gamma}
		\left\langle\d \mathbf{r}_1\times\d \mathbf{r}_2,
			\frac{\mathbf{r}_1-\mathbf{r}_2}{|\mathbf{r}_1-\mathbf{r}_2|^3}\right\rangle.
\end{equation}
Alternatively the linking number \cite{Rolfsen}*{Part~D~Chapter~5} is the number of crossing of one knot through a \underline{\emph{Seifert surface}}
for the other, that is through a compact and oriented surface with the given knot as its boundary. Such a surface always exists \cite{FP30,Seifert35}. 
The two definitions coincide since the circulation is gauge invariant.
We refer the interested reader to \cite{On_Gauss_formula} and the references therein.

\subsection{Main theorem}

For a knot $\gamma$, we denote by $\Sf(\overline{\gamma})$ the spectral flow associated to a curve $\gamma'$ which is 
in the same isotopy class as $\gamma$, but with trivial writhe. Such a representative always exists \cite{dirac_s3_paper2}*{Appendix~A}. 
By \cite{dirac_s3_paper2}*{Theorem~21}, $\Sf(\overline{\gamma})$ is a well-defined number.
Let us now state the main theorem of the paper.

\begin{theorem}\label{thm:main}
 Let $\gamma_0,\gamma_1$ be knots in $\S^3$ such that $\gamma_1$ is a $(N,M)$-cable knot upon $\gamma_0$,
 where $(N,M)\in\mathbb{N}^2$ is a couple of coprime integers $M,N\ge 2$. Then we have:
 \[
  \Sf(\overline{\gamma_1})=N\Sf(\overline{\gamma_0}).
 \]
\end{theorem}

Theorem~\ref{thm:main} is proven in Section~\ref{sec:proof_main_thm}. Since $\Sf(\mathrm{unknot})=0$ \cite{dirac_s3_paper1}*{Theorem~29}, we have the following.

\begin{theorem}\label{cor:trefoil}
	The spectral flow of a torus knot is $0$. 
\end{theorem}

By simple transitivity we obtain the obvious corollary.
\begin{corollary}\label{cor:trans}
	Let $\overline{\gamma}$ be a knot (as an isotopy class). Assume that there exists a finite family of knots
	$(\overline{\gamma}_{n})_{0\le n\le N}$, $N\in\mathbb{N}$, such that $\overline{\gamma}_0$ is the unknot (isotopy class of the circle),
	$\overline{\gamma}_N=\overline{\gamma}$, and for all $1\le n\le N$, $\overline{\gamma}_n$ is a cable knot of $\overline{\gamma}_{n-1}$.
	Then the spectral flow of $\overline{\gamma}$ is $0$.
\end{corollary}

	\begin{rem}
	This result does not cover all knots, as we can see with the Alexander polynomial \cite{Lickorish}*{Thm~6.15}.
	Nevertheless, we conjecture the spectral flow to be the trivial invariant. 
	
	For instance the $4_1$ knot is not covered. It is not a torus knot, and its Alexander polynomial is
	$3-(t+t^{-1})$ (see \cite{Rolfsen}*{Appendix~C}). Then \cite{Lickorish}*{Thm~6.15} states that the Alexander polynomial $\Delta_{\gamma}(t)$ of a cable knot $\gamma$ of $\gamma_0$
	with pattern the $(N,M)$-torus knot $g$ is:
	\[
		\Delta_{\gamma}(t)=\Delta_{g}(t)\Delta_{\gamma_0}(t^N).
	\]
	\end{rem}

\subsection{Organization of the paper}\label{sec:org}

Section~\ref{sec:proof_main_thm} is devoted to the proof of Theorem~\ref{thm:main}. It is based on a \emph{homotopy} argument.
The homotopy in the space of singular Dirac operators is described by Eqs. \eqref{eq:intro_homotopy}--\eqref{eq:loop_odd}.
In Section~\ref{sec:topology}, we introduce appropriate topologies and state the continuity properties of this homotopy in two auxiliary lemmas, Lemmas
\ref{lem:compac} and \ref{lem:continuity_along_collapse}, which are proven later. The proof is given in Section~\ref{sec:displayed_proof}.
These results are the equivalent of \cite{dirac_s3_paper1}*{Theorems~22~\&~23} and \cite{dirac_s3_paper2}*{Theorems~14~\&~15}
for the homotopy at hand. 

\smallskip
The Biot and Savart (B.S.) formula for the cable knot at hand, seen as a Coulomb gauge,
is the gauge considered in the proof: Section~\ref{sec:BS_BS} is devoted to its study.
We explain the gauge transformation relating the B.S. gauge with the singular gauges considered in \cite{dirac_s3_paper1,dirac_s3_paper2,dirac_s3_paper3}.
In Theorem~\ref{thm:behavior_BS}, we express the behavior of the B.S. gauge in the vicinity of the curve.

This gauge is easier to handle than the singular one when studying the collapse of a cable knot onto its base.
This situation is studied in Section~\ref{BS_cable} for adapted realizations introduced in Section~\ref{sec:intro_adapted_cable_rep}.
Their nice geometrical properties are described in Section~\ref{sec:prop_adap_seif_fib}. Then we give in Lemma~\ref{lem:refinement_BS_behav}
 the behavior of the B.S. gauge for these adapted realizations, and its limit as the cable knot collapses onto its base.

\smallskip

Section~\ref{sec:descr_dirac_op} provides the analytical tools for the proof of the auxiliary lemmas \ref{lem:compac} and \ref{lem:continuity_along_collapse}. 
These tools are the equivalent of \cite{dirac_s3_paper1}*{Theorems~22~\&~23} and \cite{dirac_s3_paper2}*{Theorems~14~\&~15}. In Section~\ref{sub:2dim_op}, we study the $2$D Dirac operators with
$N$ A.-B. solenoids around the origin. In particular, we study the limit as the solenoids all converge to the origin, and provide in \eqref{eq:def_f_sing_2D} a useful
family for the singular part of the domain.

In Section~\ref{sec:desc_dom}, we study in details the domain for a $3$D Dirac operator of the homotopy. We also recall two tools:
an adapted partition of unity (Sec.~\ref{sec:loc}) and phase jump functions (Sec.~\ref{sec:remov_phase}). 
In Sec.~\ref{sec:ess_dom} we give in \eqref{eq:descr_dom_D_full} an adapted essential domain for the operator.

At last we relate in Section~\ref{sec:model_op} the $3$D operator with a model operator splitting the geometry into a transversal part and a longitudinal part.
The decomposition for the model operator is given in \eqref{eq:model_case_energy_equalities} and Lemma~\ref{lem:decomp_sing_subspace}.
The relation between the model operator and the operator on $\S^3$ is given in Lemma~\ref{lem:graph_norm_estimate}.

\smallskip

In Section~\ref{sec:proof_aux_lem} we give the proof of the two auxiliary lemmas \ref{lem:compac}~$\&$~\ref{lem:continuity_along_collapse}
using the same methods as the one used for proving \cite{dirac_s3_paper1}*{Theorems~22~\&~23} and \cite{dirac_s3_paper2}*{Theorems~14~\&~15}.

\section{Proof of the main theorem}\label{sec:proof_main_thm}

\subsection{Strategy}\label{sec:strategy}
We prove Theorem~\ref{thm:main} with a homotopy argument. We consider an adapted cable $(\gamma_{\eta},\gamma_0)$ representation 
for the $(M,N)$-cable knot of $\gamma_0$, see Section~\ref{sec:intro_adapted_cable_rep}: 
we consider $\gamma_\eta$ \eqref{eq:def_gamma_eta} with a peculiar base $\gamma_0$
and a peculiar vector map $U$ (their main feature is that for all $s\in \T_{N\ell}$, the tangent vector $\dot{\gamma}_\eta(s)$ is co-linear to $\dot{\gamma}_0([s])$).
 
We take the limit $\eta\to 0^+$ (to fix ideas we consider $0\le \eta\le \eta_0\ll 1$).
We consider a homotopy in the space $\Sd$ of self-adjoint operators on $L^2(\S^3)^2$ with discrete spectrum given by:
\begin{equation}\label{eq:intro_homotopy}
	\begin{array}{rcl}
		[0,1]\times [0,\eta_0)&\longrightarrow& \Sd,\\
		(t,\eta) &\mapsto& \cD_{\bA_{\mathrm{full}}(t,\eta)},
	\end{array}
\end{equation}
where $\bA_{\mathrm{full}}(t,\eta)$ is the full magnetic potential given by the sum of the Biot and Savart gauge 
$2\pi t\bA_{\gamma_{\eta}}$ for $2\pi t [\gamma_{\eta}]$ and an auxiliary potential $\bA_{\mathrm{aux}}$ 
with circulation $\phi_a(t)$ along $\gamma_0$.
Indeed (and unfortunately), an auxiliary gauge is needed to ensure the continuity of the homotopy in such a way
that the spectral flow 
	\begin{equation*}
			\Sf\big( (\cD_{\bA_{\mathrm{full}}(t,\eta)})_{0\le t\le t_1}\big)
	\end{equation*}
is well-defined and continuous, hence constant, as $\eta\to 0$. Then it suffices
to relate this number to $\Sf(\gamma_{\eta})$ and $\Sf(\gamma_0)$ to end the proof.

For simplicity, we will take $\bA_{\mathrm{aux}}$ of the following form. Consider a finite family of disks $(D_{\pm}(s_k))_{1\le k\le K}$ from the one-parameter family
\[
 D_{\pm}(s):=\{ \bx\in B_{\eps}[\gamma_0],\ s_{\gamma_0}(\bx)=s\quad \&\quad \rho_{\gamma_0}(\bx)\le \frac{\eps}{2}\},\ s\in\T_{\ell}.
\]
They are oriented along, respectively against $\gamma_0$:
$
\link(\gamma_0,\partial D_{\pm}(s))=\pm 1.
$

We consider the magnetic knots obtained from their boundaries carrying a flux $\frac{2\pi}{K}\alpha_a$, and define 
the gauge 
\begin{equation}\label{eq:aux_gauge}
\bA_{\mathrm{aux}}(\alpha_a)=\sum_{k=1}^K2\pi\frac{\alpha_a}{K}[D_-(s_k)]. 
\end{equation}
The number $K$ of disks will be chosen so that throughout the considered paths 
the renormalized flux $\tfrac{\alpha_a}{K}$ carried by each knot does not exceed $1$.

Let us anticipate: for\footnote{ If $M<0$, then we have to reverse all of them.} $M>0$ and $MN$ even,
we take $K=M+1$, $\phi_a(t)=-2\pi M t$, and define the path
\begin{equation}\label{eq:loop_even}
	t\in \big[0,1+\tfrac{1}{M}\big]\mapsto \cD_{\bA_{\full}(t,\eta)},\ \bA_{\full}(t):=2\pi \min(t,1)\bA_{\gamma_{\eta}}
		+\sum_{k=1}^{M+1}2\pi \tfrac{Mt }{M+1}[D_-(s_{k})].
\end{equation}
This gives rise to a homotopy $(\cD_{\bA_{\full}(t,\eta)})_{t\in \big[0,\frac{M+1}{M}\big],\  0\le \eta\le \eta_0}$.

If $MN$ is odd, we consider an additional fixed potential $\pi [D_+(s_{K+1})]$,  away from the other $K$ disks, and 
\begin{multline}\label{eq:loop_odd}
	t\in \big[0,\frac{M+1}{M}\big]\mapsto \cD_{\bA_{\full}(t,\eta)},\\
	\bA_{\full}(t):=2\pi \min(t,1)\bA_{\gamma_{\eta}}+\pi [D_+(s_{M+2})]+\sum_{k=1}^{M+1}2\pi\tfrac{Mt}{M+1}[D_-(s_{k})].
\end{multline}

\begin{rem}
		The choice of the gauge should not matter as the spectrum is gauge invariant. 
		By picking a singular gauge for $2\pi \min(t,1)[\gamma_{\eta}]$, one realizes that the
		considered paths \eqref{eq:loop_even}-\eqref{eq:loop_odd} 
		are unitary equivalent to loops of (singular) Dirac operators.
\end{rem}

To study the spectral flow, we consider topologies described in \cite{Wahl08}.
They are intermediate between the topology of the strong-resolvent convergence,
which is too weak to define the spectral flow 
and that of the norm-resolvent sense (also called the gap topology),
which is too strong for the paths we consider. 
In the next Section, we recall all these topologies and then state the auxiliary lemmas \ref{lem:compac} and \ref{lem:continuity_along_collapse}
which address the continuity of the map $\cD_{\bA_{\full}}$ with respect to these topologies. The former is the equivalent to  \cite{dirac_s3_paper1}*{Theorem~22} and the latter
that of \cite{dirac_s3_paper1}*{Theorem~23} and \cite{dirac_s3_paper2}*{Theorems~14~\&~15}.
The continuity of the homotopies follows easily Corollary~\ref{cor:homot}.

\subsection{Topologies and auxiliary lemmas}\label{sec:topology}

\subsubsection{The topology of the strong-resolvent continuity}
	This topology is the weakest topology (on the space $\mathcal{S}_{\cH}$ of self-adjoint operators on a Hilbert space $\cH$) so that 
	for any $\psi\in \cH$ the following maps are continuous
	\[
		D\in \mathcal{S}_{\cH}\mapsto (D\pm i)^{-1}\psi\in \cH.
	\]
	We have the following Lemma \cite{dirac_s3_paper1}.
	\begin{lemma}\label{lem:char_sres_conv}
		Let $(\cD_n)$ be a sequence of (unbounded) self-adjoint operators on a separable 
		Hilbert space $\cH$. 
		Then the following statements are equivalent.
		\begin{enumerate}
    			\item $\cD_n$ converges to $\cD$ in the strong resolvent sense.
    	
			\item For any $(f,\cD f)\in \cG_{\cD}$,
  			there exists a sequence $(f_n,\cD_n f_n)\in \cG_{\cD_n}$ converging to 
			$(f,\cD f)$ in $\cH \times \cH$.
    	
			\item The orthogonal projection $P_n$ onto $\cG_{\cD_n}$ 
			converges in the strong operator topology to $P$, the orthogonal projector onto $\cG_{\cD}$.
		\end{enumerate}
	\end{lemma}

\subsubsection{The topology of the norm-resolvent continuity}
	This is the topology given by the distance $\dist(D_1,D_0):=\norm{(D_1+i)^{-1}-(D_0+i)^{-1}}_{\cB(\cH)}$.

\subsubsection{Intermediate topologies: bump toplogies}
	We now introduce a scale of topologies calibrated to define the spectral flow and 
	refer the reader to \cite{Wahl08} for further details.
	
	Let $\phi\in\sD(\R;\R_+)$ be a bump function
	centered at $0$, in the sense that $\phi$ is even and on its support $[-a,a]\subset \R$, 
	$\phi'(x)>0$ for $x\in(-a,0)$ and $\phi'(x)<0$ for $x\in(0,a)$. We define $\phi_n(x):=\phi(nx)$. 

	The topology $\sT_{\phi}\subset 2^{\Sd}$ is defined 
	as the smallest topology such that for any $\psi\in \cH$, the following maps are continuous:
	$$
	\left\{
	\begin{array}{rcl}
		D\in \Sd&\mapsto &(D\pm i)^{-1}\psi\in \cH,\\
		D\in \Sd&\mapsto &\phi(D)\in\cB\big(\cH\big).
	\end{array}
	\right.
	$$
	Observe that by functional calculus we have the inclusion $\cT_{\phi'}\subset \cT_{\phi}$
	whenever the support of the bump function $\phi'$ is in the interior of that of $\phi$.

	For our purpose, we can only deal with the family $(\sT_{\phi_n})_{n\in\mathbb{N}}$. A map $c:U\to\Sd$
	from a topological set $U$ to $\Sd$ is said to be bump continuous if it is $\sT_{\phi_n}$-continuous
	for some $n$. In particular when $U$ is compact, local bump continuity\footnote{that is the existence, for each open set $O\subset U$, of a bump function $\phi_O$
	so that the restriction $c\restriction O$ is $\cT_{\phi_O}$-continuous.} coincides with bump continuity 
	thanks to the Heine-Borel property.
	
	\begin{rem}\nonumber
		More generally, we define the bump continuity at a spectral level $\lambda\in\R$ the same way,
		except that the bump functions $\Lambda$ are requested to be centered around $\lambda$ instead.
		If we do not precise the spectral level, then it is assumed to be $0$.
	\end{rem}

\subsubsection{Auxiliary lemmas for Theorem~\ref{thm:main}}

We now turn to auxiliary lemmas which address the continuity of $\cD_{\bA_{\full}}$.
The map refers to \eqref{eq:intro_homotopy} with the auxiliary flux \eqref{eq:aux_gauge}, and is
	seen as a function of 
	\[
	\big(t,\eta,\alpha_a\big)\in[0,1]_{\per}\times [0,\eta_0]\times [0,K).
	\]

\begin{lemma}\label{lem:compac}
	Let $(t_n,\eta_n,\alpha_a^{(n)})$ be a sequence with $t_n\to t\in[0^+,1^-]$, $\eta_n$ decreasing to $0$
	and $\alpha_a^{(n)}\to \alpha_a\in [0,K)$.
	Let $(\cD_n)$ be the corresponding sequence of operators $\cD_{\bA_{\full}}$, and $\psi_n\in\dom(\cD_n)$
	a sequence satisfying $\limsup_{n\to+\infty}\norm{\psi_n}_{\cD_n}<+\infty$. 
	
	Then, up to extraction, $(\psi_n,\cD_n \psi_n)$ converges to
	some $(\psi,\cD_{t,0}\psi)$, weakly in $L^2(\S^3)^2$. 
	Furthermore, $\psi_n$ converges in $L^2_{\loc}(\S^3\setminus \gamma_0)^2$.
\end{lemma}

\begin{lemma}\label{lem:continuity_along_collapse}
	The map $\cD_{\bA_{\full}}$ is continuous in the strong-resolvent sense
	and, given $\lambda\in\R$, satisfies the following.
	\begin{enumerate}
		\item It is $\lambda$-bump continuous
		in the range $[0,1)\times (0,\eta_0]\times [0,K)$.
	
		\item It is $\lambda$-bump continuous at the point $(t,0,\alpha_a)$ 
				if and only if $\lambda$ does not belong to
		\begin{equation}\label{eq:crit_eig}
	 \lambda\in\bigg\{\frac{1}{\wt{\ell}}\Big( -2\pi\frac{M}{N}(Nt-k-\tfrac{1}{2})+2\pi \alpha_a+\pi+2m\pi\bigg),
	 	\ m\in\Z,\ 0\le k\le  \lfloor Nt\rfloor -1\Big\},
	 \end{equation}
	 where $\wt{\ell}:=\int_0^\ell \Omega(\gamma_0(s))\d s$ is the $\S^3$-length of $\gamma_0$.
	\end{enumerate}
\end{lemma}

To prove them we proceed as in \cite{dirac_s3_paper2,dirac_s3_paper3} and study the convergences from the simplest layer to
the more difficult. Section~\ref{sec:descr_dirac_op} is devoted to a detailed description of the singular Dirac operators and their domains.
In Section~\ref{sec:proof_aux_lem} we study  the continuity in $L^2(\S^3)^2$
as in \cite{dirac_s3_paper2}.

\begin{corollary}\label{cor:homot}
	For $\eta_0>0$ sufficiently small, the homotopies \eqref{eq:loop_even}-\eqref{eq:loop_odd} are bump continuous.
\end{corollary}
\begin{proof}[Proof of Corollary~\ref{cor:homot}]

	We focus on the first part ($t\in[0,1]$), the second part is known to be bump continuons \cite{dirac_s3_paper2}*{Theorem~14-15}.
	We choose $\eta_0$ small enough such that $\Wr(\gamma_{\eta})$ is almost $N\Wr(\gamma_{0})$, and the homotopies are continuous
	in the limit $(t,\eta)\to (1^-,\eta)$ for all $0<\eta\le \eta_0$ \cite{dirac_s3_paper2}*{Theorem~15}.
	 Factorizing by $N^{-1}$ in Formula~\ref{eq:crit_eig} leads us to choose $\alpha_a(t)$ satisfying
	 \[
	 \begin{array}{rcl}
	  \pi(N+M)-2\pi MN t+2\pi N\alpha_a(t)&\not\equiv &0\mod (2\pi).
	 \end{array}
	  \]
	  As $M+N=(1-MN)\mod (2)$, we substitute in the formula above and 
	  divide everything by $2\pi N$ to obtain
	   \[
	 \begin{array}{rcl}
	  \alpha_a(t)&\not\equiv &M t+\tfrac{1}{2}\big(M-\tfrac{1}{N}\big)\mod \big(\tfrac{1}{N}\big).
	 \end{array}
	  \]
	  
	  If $MN$ is even, the choice $\alpha_a(t)=Mt$ ensures continuity.
	  If $MN$ is odd, the additional $\pi [D_+(s_{J+1})]$ to $\bA_{\aux}$ shifts the set \eqref{eq:crit_eig} of critical eigenvalues as $\eta\to 0^+$ 
	  which becomes (see the proof of Lemma~\ref{lem:continuity_along_collapse})
	  \begin{equation*}
	 \bigg\{\frac{1}{\wt{\ell}}\Big( -2\pi\frac{M}{N}(N t-k-\tfrac{1}{2})+2\pi \alpha_a+2m\pi\bigg),\ m\in\Z,\ 0\le k\le E(N t)-1\Big\}.
	  \end{equation*}
	  Then the same choice $\alpha_a(t)=Mt$ ensures continuity of the second homotopy.
	  Note that as $t\to 1^-$, the critical set converges to
	   \[
	 	\Big\{\frac{1}{N\wt{\ell}}\Big( M\pi +2\pi N\alpha_a+2\pi m\pi\Big),\ m\in\Z\Big\}.
	 \]
	 
	 Note that if $M$ is negative, it suffices to reverse the orientation of all the magnetic knots to keep the same circulation along the cable knot, hence the continuity
	 properties.
\end{proof}

\subsection{Proof of Theorem~\ref{thm:main}}\label{sec:displayed_proof}
	We now proceed to the proof as explained in Section~\ref{sec:strategy}. We recall that $M\wedge N=1$.

	As in \cite{dirac_s3_paper2}, we call critical points the couple of fluxes for which $\cD_{\bA_{\full}}$
	is not bump continuous. We refer the reader to \cite{dirac_s3_paper2}*{Theorems~14-15}.
	
	We have seen the bump continuity and strong resolvent continuity of the homotopies \eqref{eq:loop_even}-\eqref{eq:loop_odd} 
	in the previous section (Lemma~\ref{lem:continuity_along_collapse} and Corollary~\ref{cor:homot}).
	At $\eta\in [0,\eta_0]$ fixed, we write $L_\eta$ the corresponding path of Dirac operators 
	$(\cD_{\bA_{\full}(t,\eta)})_{0\le t\le \tfrac{M}{M+1}}$.
	The continuity gives $\Sf(L_{\eta})=\Sf(L_0)$ (we emphasize that the endpoints of each path 
	are unitary equivalent to the free Dirac operator). There remains to write $\Sf(L_{\eta})$ resp. $\Sf(L_0)$ 
	in terms of the spectral flow of $\gamma_\eta$ resp. $\gamma_0$.

	Along the proof, we will encounter two numbers: the number $\delta(N,M)\in (0,1)$ equal to $\tfrac{1}{2}(\tfrac{M}{N}-\mathds{1}_{(MN\in 2\N)})$ modulo $1$, 
	and $D(N,M)\in\mathbb{N}$ defined as follows:
	\begin{equation}\label{eq:def_D_M_N}
		D(N,M):=\sum_{k=1}^N\Big|\Big\{ \delta(N,M)+j,\ 0\le j< M\Big\}\cap \Big(0, \frac{Mk}{N}\Big)\Big|.
	\end{equation}
	Equivalently, we have $D(N,M)=\sum_{k=1}^N\sum_{j=0}^{M-1}\mathds{1}(j+\delta(N,M)<\tfrac{Mk}{N})$. 
	This combinatorial term can be easily computed in the case $N>M$.
	\begin{lemma}\label{lem:compute_D_N_M}
	Let $N>M>0$ be two coprime integers. 
	If $NM$ is even resp. odd, then $D(N,M)=\tfrac{MN}{2}$ resp. $D(N,M)=\tfrac{(M+1)N}{2}$.
	\end{lemma}

	Writing $\eps=1-\overline{\eps}=\mathds{1}(MN\in 2\mathbb{N})$, we first show the intermediate formula:
	\begin{equation}\label{eq:intermediate}
	\Sf(\overline{\gamma_1})= N\Sf(\overline{\gamma_0})+N\lfloor \tfrac{1}{2}(\eps-\tfrac{M}{N}) \rfloor+D(N,M)
  		-\frac{(M-\overline{\eps})N}{2}.
	\end{equation}
	We will end the proof with:
	\begin{lemma}\label{lem:comb}
	 Let $N,M>0$ be two coprime numbers and $\eps=1-\overline{\eps}$ be the number $\mathds{1}(MN\in 2\mathbb{N})$. 
	 Then we have
	 \[
	  N\lfloor \tfrac{1}{2}(\eps-\tfrac{M}{N}) \rfloor+D(N,M)
  		-\frac{(M-\overline{\eps})N}{2}=0.
	 \]
	\end{lemma}

    \begin{rem}
       We establish Lemma~\ref{lem:compute_D_N_M} with a combinatorial proof.
       Lemma~\ref{lem:comb} is proven through Lemma~\ref{eq:intermediate} and the fact that the $(N,M)$-torus knot and the $(M,N)$-torus knot define
       the same isotopy class. It would be interesting to have a purely combinatorial proof.
    \end{rem}

	\subsubsection{Proof of \eqref{eq:intermediate}}

	 \paragraph{\emph{Case $MN$ even}}
	 
	 As $t\to 1^-$ at fixed $\eta$, the critical values $\alpha_{c,+}^{(\eta)}>0$ for $\alpha_a$ are given by 
	 the equation
	 \[
	 	\pi(1-\Wr(\gamma_\eta))+2\pi N\alpha_{c,+}^{(\eta)}\equiv 0\!\!\!\mod (2\pi).
	 \]
	 As  $\eta\to 0^+$, we obtain:
	 \begin{equation}\label{eq:crit_eta}
	 	t=1,\ \alpha_{c,+}^{(\eta)}\equiv \frac{1}{2}\Big(M-\frac{1}{N}\Big)\!\!\!\mod \Big(\frac{1}{N}\Big).
	 \end{equation}
	 
	 At $\eta=0$ (for $\gamma_0$), the critical values $\alpha_{c,+}^{(0)}>0$ for $\alpha_a$ 
	 are given by
	 the equation $\pi(1-\Wr(\gamma_0))+2\pi \alpha_{c,+}^{(0)}\equiv 0\mod 2\pi$, whenever $Nt\to k^-$, $k\in\mathbb{N}$, that is:
	  \begin{equation}\label{eq:crit_zero}
	 	 t\in \tfrac{1}{N}\mathbb{N},\ \alpha_{c,+}^{(0)}\equiv \frac{1}{2}\Big(\frac{M}{N}-1\Big)\!\!\!\mod (1).
	  \end{equation}
	We draw the $L_\eta$'s on different copies of the torus of fluxes $\Tf^{M+2}$ corresponding to the different links. 
	Having coupled the fluxes of the auxiliary knots, the effective torus of fluxes 
	for $L_\eta$ is $\rT_{\mathrm{eff}}:=\Tf\times [0,\tfrac{M}{M+1}]_{\per}$
	with running point $(t,\alpha_a)$. It corresponds to the following subset in $\Tf^{M+2}$:
	\[
		\{(t,\tfrac{\alpha_a}{M+1},\cdots,\tfrac{\alpha_a}{M+1})\}\subset \Tf^{M+2}.
	\]
	We change $L_\eta$ into the loop $\Lambda_{\eta}$ defined by:
	$
	u\in [0,2+\tfrac{1}{M}]\mapsto 2\pi(\min(u,1)[\gamma_{\eta}] +\max(u-1,0)[\gamma_{\aux}]).
	$
	We use \cite{dirac_s3_paper2}*{Theorems~19~$\&$~21} (and their proof) to determine the change in the spectral flow
	at the crossing of a critical point $(t,\alpha_a)=(1,\alpha_c^{(\eta)})$ \eqref{eq:crit_eta}. 
	Indeed the spectral flow corresponding to the following small circle is $1$:
	\[
	 \Big(
	 2\pi\Big[(1+\eps \cos(\theta))[\gamma_{\eta}]+(\alpha_c^{(\eta)}+\eps \sin(\theta))[\gamma_{\aux}] 
	 \Big]\Big)_{\theta\in\R/(2\pi\Z)},\ 0<\eps\ll 1.
	\]
	There are $MN$ points in $\{1\}\times[0,M]$ satisfying \eqref{eq:crit_eta}. For $\eta\ll 1$, this gives
	\begin{equation}\label{eq:sf_eta}
	\begin{array}{rcl}
	\Sf(L_\eta)=\Sf(\gamma_{\eta})+MN&=&\Sf(\overline{\gamma_1})+\lfloor \tfrac{1}{2}(1-MN)\rfloor+MN,\\
					  &=&\Sf(\overline{\gamma_1})-\frac{MN}{2}.
	\end{array}
	\end{equation}
	We also draw $L_0$ on $\rT_{\mathrm{eff}}$: now the point $(t,\alpha_a)$ stands for the Dirac operator with magnetic link 
	$2\pi(Nt[\gamma_0]+\tfrac{\alpha_a}{M+1}\gamma_{\aux})$.
	Similarly we change $L_0$ into $\Lambda_0$ defined by 
	$u\in [0,2+\tfrac{1}{M}]\mapsto 2\pi(N\min(u,1)[\gamma_{0}] +\max(u-1,0)[\gamma_{\aux}]).$
	
	Recall: $\delta(N,M)\in(0,1)$ is the number equal to $\tfrac{1}{2}(\tfrac{M}{N}-1)$ modulo $1$.
	Now the number of critical points \eqref{eq:crit_zero} encountered is  $D(N,M)$ (defined in \eqref{eq:def_D_M_N}).
	See Figure~\ref{fig:homotopy}: keep in mind for Section~\ref{sec:proof_lemma_comput} that it corresponds to the number of critical points for
	the $\gamma_0$-loop (gray dots in the figure) \emph{above} the oblique line defined by the first part of $L_0$.
	We obtain:
	\begin{equation}\label{eq:sf_zero}
	\Sf(L_0)=N\Sf(\gamma_0)+D(N,M)=N(\Sf(\overline{\gamma_0})+\lfloor\tfrac{1}{2}(1-\tfrac{M}{N})\rfloor)+D(N,M).
	\end{equation}
	We emphasize that the paths $\Lambda_{\eta}$ and $\Lambda_0$ are made of two parts: $0\le u\le 1$, 
	which defines a first loop with spectral flow $\Sf(\gamma_{\eta})$ resp. $N\Sf(\gamma_0)$, and then $u\ge 1$,
	defining a second loop with trivial spectral flow \cite{dirac_s3_paper2}*{Theorems~14$\&$15} and \cite{dirac_s3_paper1}*{Theorem~29}.

	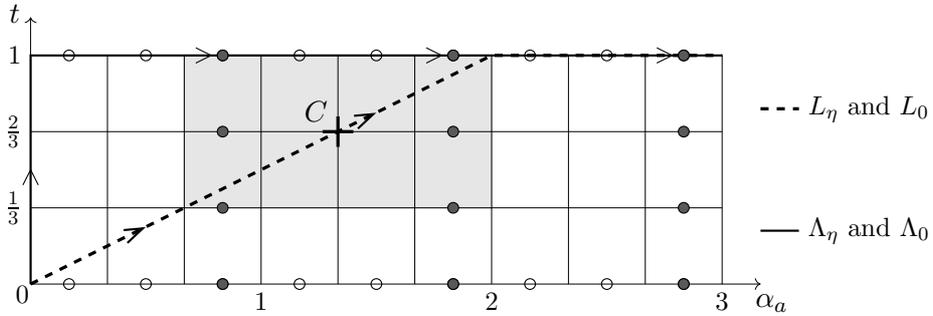
\begin{figure}[!!ht]
\resizebox{1\textwidth}{!}{
		\begin{tikzpicture}
			
			\draw[thick,fill,color=gray!20] (2,1)--(2,3)--(6,3)--(6,1)--cycle;	
			\draw [very thick] (3.8,2) -- (4.2,2);
			\draw [very thick] (4,1.8) -- (4,2.2);
			\node [above left] at (4,2) {$C$};
			
			\draw[very thin] (0,0) -- (9,0) -- (9,3) -- (0,3) -- cycle;
			\draw[very thin] (0,1) -- (9,1);
			\draw[very thin] (0,2) -- (9,2);
			\foreach \x in {1,2,3,4,5,6,7,8}
			{
				\draw[very thin] (\x,0) -- (\x,3);
			}
			\draw[very thick,dashed] (0,0) -- (6,3) -- (9,3);
			\draw [line width= 0.8pt] (0,0) -- (0,3) -- (9,3);

			\draw [very thin,->] (0,3) -- (0,3.5);
			\draw [very thin,->] (9,0) -- (9.5,0);
			\node[above left] at (0,3.3) {$t$};
			\node[below right] at (9.3,0) {$\alpha_a$};
			
			\node[left,scale=0.9] at (0,3) {$1$};
			\node[left,scale=0.9] at (0,2) {$\tfrac{2}{3}$};
			\node[left,scale=0.9] at (0,1) {$\tfrac{1}{3}$};
			
			\node[below,scale=0.9] at (3,0) {$1$};
			\node[below,scale=0.9] at (6,0) {$2$};
			\node[below,scale=0.9] at (9,0) {$3$};
			
			\node[below left,scale=0.9] at (0.1,0.1) {$0$};

			\foreach \x in {0,1,2,3,4,5,6,7,8}
			{
			\draw ({1/2+\x},3) circle [radius = 0.07];
			\draw ({1/2+\x},0) circle [radius = 0.07];
			}
			\foreach \x in {0,3,6}
			{
				\foreach \y in {0,1,2,3}
				{
				\draw [fill=gray!130] ({\x+5/2},{\y}) circle  [radius = 0.07];
				}
			}
			\draw [very thick,dashed] (9.5,2.3)--(10,2.3); 
			\draw [line width= 0.8pt] (9.5,0.7)--(10,0.7); 
			\node [right,scale=0.9] at (10,2.3) {$L_{\eta}$ and $L_0$};
			\node [right,scale=0.9] at (10,0.7) {$\Lambda_{\eta}$ and $\Lambda_0$};
			\draw (-0.1,1.3) -- (0,1.5) -- (0.1,1.3);
			\draw (2.15,3.1) -- (2.35,3) -- (2.15,2.9); 
			\draw (5.15,3.1) -- (5.35,3) -- (5.15,2.9); 
			\draw (8.15,3.1) -- (8.35,3) -- (8.15,2.9); 
			\draw[thick] (1.21,0.708) -- (1.46,0.73) -- (1.26,0.54);
			\draw[thick] (4.21,2.208) -- (4.46,2.23) -- (4.26,2.04);
		\end{tikzpicture}
		}
		\caption{The loops on the effective torus for $(N,M)=(3,2)$.}
		\label{fig:homotopy}
	\end{figure}
		In Figure~\ref{fig:homotopy}, we circled the critical points for the $\gamma_{\eta}$- and $\gamma_0$-loops 
		(and not all the critical points for the full homotopy).
		On $\alpha=0\sim 1$, there are $6$ critical points for the $(3,2)$-cable knot (displayed here for $\eta=0^+$), and 
		the critical points for the $\gamma_0$-loop are given by the grayed circles. To deform $L_\eta$ into $\Lambda_\eta$,
		or $L_0$ into $\Lambda_0$, we need to pass over $3\times2=6$ respectively $D(3,2)=3$ critical points.
		The gray area denotes the rectangular subset in which there is a central symmetry of the $\gamma_0$-critical points
		with respect to its center $C$.

	\medskip
	
	 \paragraph{\emph{Case $MN$ odd}}
	Keeping $\pi [\partial D_+(s_{M+2})]$ fixed, we do a similar study. 
	First, observe that for each $t$, the circulation of $\bA_{\aux}$ (with the additional term) along $\gamma_{\eta}$ is
	$N(\pi-2\pi\alpha_a)$, where $N$ is odd. We obtain the following critical values:
	\[
	\left\{
	\begin{array}{cl}
	\mathrm{for}\ \gamma_{\eta} & t=1,\ \alpha_{c,-}^{(\eta)}\equiv \frac{1}{2N}\Wr(\gamma_{\eta})\!\!\!\mod (\tfrac{1}{N}),\\
	\mathrm{for}\ \gamma_0 & t\in\tfrac{1}{N}\N,\ \alpha_{c,-}^{(0)}\equiv \frac{1}{2}\Wr(\gamma_{0})\!\!\!\mod (1).
	\end{array}
	\right.
	\]
	Let $\wt{L}_\eta$ and $\wt{L}_0$ denote the same loops but with the additional fixed term $\pi [\partial D_+(s_{M+2})]$.
	We do the same deformation as before. We obtain similar formulas to \eqref{eq:sf_eta}, \eqref{eq:sf_zero}
	except that $\Sf(\gamma_{\eta})$ and $N\Sf(\gamma_0)$ are replaced by the spectral flows corresponding to the following loops:
	\[
	\left\{
		\begin{array}{l}
			t\in[0,1]\mapsto 2\pi t[\gamma_{\eta}]+\pi [\partial D_+(s_{M+2})],\\
			t\in[0,1]\mapsto 2N\pi t[\gamma_{0}]+\pi [\partial D_+(s_{M+2})].
		\end{array}
	\right.
	\]
	Using \cite{dirac_s3_paper2}*{Theorem~21}, we get:
	\[
	\begin{array}{rcl}
		\Sf(\wt{L}_\eta)&=&\Sf(\overline{\gamma_{\eta}})+\lfloor \tfrac{1}{2}(1-N-MN)\rfloor+MN,\\
				&=&\Sf(\overline{\gamma_{\eta}})-\frac{(M-1)N}{2},\\
		\Sf(\wt{L}_0)&=&N(\Sf(\overline{\gamma}_0)+\lfloor -\tfrac{M}{2N}\rfloor)+D(N,M).
	\end{array}
	\]

	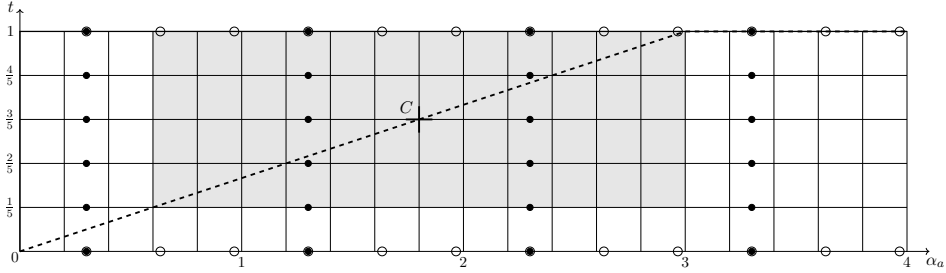
\begin{figure}[!!ht]
\resizebox{1\textwidth}{!}{
		\begin{tikzpicture}
			
			\draw[thick,fill,color=gray!20] (3,1)--(3,5)--(15,5)--(15,1)--cycle;		
			
			\draw[very thin] (0,0) -- (20,0) -- (20,5) -- (0,5) -- cycle;
			\foreach \x in {1,2,3,4}
			{
				\draw[very thin] (0,\x) -- (20,\x);
			}
			\foreach \x in {1,2,3,4,5,6,7,8,9,10,11,12,13,14,15,16,17,18,19}
			{
				\draw[very thin] (\x,0) -- (\x,5);
			}
			\draw[very thick,dashed] (0,0) -- (15,5) -- (20,5);
			\draw [line width= 0.8pt] (0,0) -- (0,5) -- (20,5);

			\draw [very thin,->] (0,5) -- (0,5.5);
			\draw [very thin,->] (20,0) -- (20.5,0);
			\node[above left] at (0,5.3) {$t$};
			\node[below right] at (20.3,0) {$\alpha_a$};
			
			\node[left,scale=0.9] at (0,5) {$1$};
			\node[left,scale=0.9] at (0,4) {$\tfrac{4}{5}$};
			\node[left,scale=0.9] at (0,3) {$\tfrac{3}{5}$};
			\node[left,scale=0.9] at (0,2) {$\tfrac{2}{5}$};
			\node[left,scale=0.9] at (0,1) {$\tfrac{1}{5}$};
			
			\node[below,scale=0.9] at (5,0) {$1$};
			\node[below,scale=0.9] at (10,0) {$2$};
			\node[below,scale=0.9] at (15,0) {$3$};
			\node[below,scale=0.9] at (20,0) {$4$};
			
			\node[below left,scale=0.9] at (0.1,0.1) {$0$};

			\foreach \x in {0,1,2,3,4,5,6,7,8,9,10,11}
			{
			\draw ({3/2+5/3*\x},5) circle [radius = 0.1];
			\draw ({3/2+5/3*\x},0) circle [radius = 0.1];
			}
			\foreach \x in {0,5,10,15}
			{
				\foreach \y in {0,1,2,3,4,5}
				{
				\draw [fill] ({\x+3/2},{\y}) circle  [radius = 0.07];
				}
			}
			
			\node[above left,scale=1] at (9,3) {$C$};
			\draw[very thick] (8.7,3)--(9.3,3);
			\draw[very thick] (9,2.7)--(9,3.3);
			\end{tikzpicture}
		}
		\caption{The loops on the effective torus for $(N,M)=(5,3)$.}
		\label{fig:homotopy_2}
	\end{figure}

\subsubsection{Proof of Lemma~\ref{lem:compute_D_N_M}}\label{sec:proof_lemma_comput}
	In the proof of the main theorem, we have seen that $D(N,M)$ was the number of critical points for
	the $\gamma_0$-loop (in gray in Figure~\ref{fig:homotopy}) \emph{above} the oblique line defined by the first part of $L_0$.
	Assume $N>M$.
	\medskip
	\paragraph{\textit{Case $MN$ even}}
	We have $\delta(N,M)=\tfrac{1}{2}(1+\tfrac{M}{N})$, and
	these critical points are the points with coordinates $(\alpha_a,\alpha)=(j+\tfrac{1}{2}(1+\tfrac{M}{N}),\tfrac{k}{N})$, $j,k\ge 0$
	in the effective torus. The key observations are that 
	\begin{itemize}
		\item all the critical points above the oblique line lie in
		the rectangular subset $[\tfrac{M}{N},M]\times[\tfrac{1}{N},1]$,
		\item in this rectangular subset of $(\alpha_a,t)$'s, there are in total $MN$ critical points 
		above and below the oblique line,
		\item there is a central symmetry of these critical points with respect to the center of the rectangle 
		$C=(\tfrac{M}{2}(1+\tfrac{1}{N}),\tfrac{1}{2}(1+\tfrac{1}{N}))$.
	\end{itemize}
	This gives immediately $D(N,M)=\tfrac{MN}{2}$.
	The last observation can be easily checked: shifting everything by $(\tfrac{M}{N},\tfrac{1}{N})$
	so that the rectangle becomes $[0,\tfrac{M-1}{N}]\times[0,\tfrac{1}{N}]$, the central symmetry with respect to the (shifted)
	center maps $(j+\tfrac{1}{2}(1-\tfrac{M}{N}),\tfrac{i}{N})$, $i=k-1$ into:
	\[
	\big(M(1-\tfrac{1}{N}),1-\tfrac{1}{N}\big)-\big(j+\tfrac{1}{2}(1-\tfrac{M}{N}),\tfrac{i}{N}\big)
		=\big((M-1)-j +\tfrac{1}{2}(1-\tfrac{M}{N}),1+\tfrac{1-i}{N} \big).
	\]
	\paragraph{\textit{Case $MN$ odd}}
	There holds $\delta(N,M)=\tfrac{M}{2N}$. We refer the reader to Figure~\ref{fig:homotopy_2}. 
	We can split the critical points for the $\gamma_0$-loop above the oblique line into two sets: first the $N$ points
	corresponding to $(\alpha_a,t)=(\tfrac{M}{2N},\tfrac{k}{N}),\ 1\le k\le N$, and the remainder. As for the case $MN$ even,
	the remainder set lies in the rectangular subset $[\tfrac{M}{N},M]\times[\tfrac{1}{N},1]$ (in gray in Figure~\ref{fig:homotopy_2}).
	In this rectangular subset there are in total $(M-1)N$ critical points with the same central symmetry with respect to the center 
	$(\alpha_a,t)=(\tfrac{M}{2}(1+\tfrac{1}{N}),\tfrac{1}{2}(1+\tfrac{1}{N}))$. We obtain:
	\[
	D(N,M)=N+\frac{(M-1)N}{2}=\frac{(M+1)N}{2}.
	\]
	
	\subsubsection{Proof of Lemma~\ref{lem:comb}}
	Using Lemma~\ref{lem:compute_D_N_M}, the result is obvious in the case $N>M$. Indeed in that case
	the following holds for $MN$ even and $MN$ odd respectively:
	\[
	\lfloor \tfrac{1}{2}(\eps(N,M)-\tfrac{M}{N}) \rfloor=0\quad\&\quad \lfloor \tfrac{1}{2}(\eps(N,M)-\tfrac{M}{N}) \rfloor=-1.
	\]
	
	Let us now assume $M>N$.
	Recall Remark~\ref{rem:on_torus_knots}: the $(M,N)$-torus knot and $(N,M)$-torus knot define the same isotopy class.
	Recall also $\Sf(\mathrm{unknot})=0$ \cite{dirac_s3_paper1}*{Theorem~29}. 
	Using \eqref{eq:intermediate} for $\gamma_0$ the circle and $\gamma_1$ the $(N,M)$-torus knot, which is the $(M,N)$-torus knot,
	we obtain:
	\[
	   \Sf(\overline{\gamma}_1)=0=N\lfloor \tfrac{1}{2}(\eps-\tfrac{M}{N}) \rfloor+D(N,M)
  		-\frac{(M-\overline{\eps})N}{2}.
	\]

\section{The Biot and Savart gauge}\label{sec:BS_BS}

Let $\bg:\T_\ell\to \R^3$ be a smooth simple curve in $\R^3$. 
We study the behavior of the vector fields -- or the one-form by duality -- given by the Biot and Savart law \eqref{eq:BS} (for unit current or flux).
For short, we write $\bA$ instead of $\bA_{\bg}$. First, we study this gauge choice and relate it to the singular gauges
considered in \cite{dirac_s3_paper1,dirac_s3_paper2,dirac_s3_paper3}. Then we give in Theorem~\ref{thm:behavior_BS}, the behavior of $\bA$ in the vicinity of $\bg$.

\subsection{Gauge choice}\label{sec:gauge_choice}
In \cite{dirac_s3_paper1,dirac_s3_paper2,dirac_s3_paper3}, we have studied the Dirac operator 
with a magnetic knot in a singular gauge, defined by a Seifert surface
of the associated knot and the flux $0\le 2\pi\alpha\le 2\pi$. This singular gauge imposes the phase jump $e^{-2i\pi\alpha}$ 
across $S$, and the Dirac operator acts like the free Dirac operator off the surface. 
In particular when the flux $2\pi\alpha\equiv 0\mod 2\pi$, then the Dirac operator coincides with the free Dirac operator.

Let us now consider the smooth gauge $\bA$ of \eqref{eq:BS}, or more precisely $2\pi\bA$.
The corresponding operator is unitarily equivalent to the free Dirac operator,
hence we necessarily have:
$
2\pi\bA=\nabla^{\R^3}\phi,
$
where $\phi$ is a phase function defined on $\R^3\setminus \bg$. Let us determine $\phi$.

With Gauss' formula for the linking number \eqref{eq:link} in mind, it is easy to figure out the class of such functions. 
For $\bx_0,\bx_1\in \R^3\setminus \bg$, let $c:[0,1]\to \R^3\setminus \bg$ be a differentiable path connecting $\bx_0$ and
$\bx_1$ (with $c(0)=\bx_0$ and $c(1)=\bx_1$). We define the phase $\phi(\bx_1,\bx_0)\in \R/(2\pi\Z)$ by the formula:
\begin{equation*}
	\phi(\bx_1,\bx_0):=2\pi\dint_{c}\cip{\bA}{\d \mathbf{s}}\!\!\mod 2\pi=2\pi(c;\bA^\flat)\!\!\mod 2\pi, 
\end{equation*}
where $(c;\bA^\flat)$ denotes the integration of $\bA^\flat\in\Omega^1(\R^3\setminus\bg)$ along $c$.
As the notation suggests $\phi(\bx_1,\bx_0)$ does not depend on the choice of $c$ connecting the two endpoints, 
and the following lemma holds.
\begin{lemma}
	Let $\bx_0,\bx_1,\bx_2\in \R^3\setminus \bg$. Then $\mathrm{exp}(i\phi(\bx_1,\bx_0))$ does not depend on $c$ 
	that is it is a homotopy invariant (with fixed endpoints). Furthermore, we have:
	\begin{equation}\label{eq:linear}
		\mathrm{exp}\big(i\phi(\bx_2,\bx_0)\big)=\mathrm{exp}\big(i\phi(\bx_2,\bx_1)+i\phi(\bx_1,\bx_0)\big).
	\end{equation}
\end{lemma}

\begin{proof}
	Let $c_1,c_2:[0,1]\to \R^3\setminus\bg$ be two differentiable paths connecting $\bx_0$ and $\bx_1$.
	Let $c_2^{-1}$ be the inverse path of $c_2$:
	\[
		c_2^{-1}:t\in[0,1]\mapsto c_2(1-t).
	\]
	The concatenation $c_1$ followed by $c_2^{-1}$ is a loop of $\R^3\setminus\bg$.  
	In general, we have to insert between $c_1$ and $c_2^{-1}$ a loop $c_3$ with basepoint $\bx_1$ and values 
	in a neighborhood $B_\eps[\bx_1]$ in such a way that $\dot{c}_1(1)\parallel \dot{c}_3(0)$ and $\dot{c}_2^{-1}(0)\parallel\dot{c}_3(1)$. 
	Similarly we add at the beginning a loop $c_4:[0,1]\mapsto B_\eps(\bx_0)$ with basepoint $\bx_0$
	such that $\dot{c}_4(0)\parallel \dot{c}_2^{-1}(1)$ and $\dot{c}_4(1)\parallel \dot{c}_1(0)$.
	The obtained concatenation $c:=[c_2^{-1}][c_3][c_1][c_4]$ is a differentiable loop
	(up to reparametrisation). As $\nabla\cdot 2\pi\bA=0$ on the contractile open sets $B_\eps[\bx_1]$
	and $B_\eps[\bx_0]$, we have:
	\[
	\big(c_3;2\pi\bA^\flat\big)=\big(c_4;2\pi\bA^\flat\big)=0,
	\]
	(as a closed form is locally exact). By Gauss' formula \eqref{eq:link}, we have:
	\[
	\big(c;2\pi\bA^\flat\big)=2\pi\link(c, \bg)\equiv 0\mod\,2\pi.
	\]
	We also have:
	\[
	\big(c;2\pi\bA^\flat\big)=\big(c_1;2\pi\bA^\flat\big)-\big(c_2;2\pi\bA^\flat\big).
	\]
	The formula~\eqref{eq:linear} is obvious, and follows by concatenation of paths.
\end{proof}

Note that the function $\phi(\bx_1,\bx_0)$ is differentiable in $\bx_1$ and $\bx_0$ on $\R^3\setminus\bg$, and that we have:
\[
2\pi\bA=\nabla^{\R^3}\phi(\cdot,\bx_0).
\]
The difference of $\phi(\cdot,\bx_0)$ and $\phi(\cdot,\bx_0')$ is then the constant $\phi(\bx_0,\bx_0')$. 

We obtain a partition of $\R^3\setminus \bg$ defined by the equivalence relation:
\[
\bx_1\sim\bx_0\,\iff\, \phi(\bx_1,\bx_0)\equiv 0\mod\,2\pi. 
\]
We claim that almost all equivalence classes are surfaces
$\phi(\cdot,\bx_0)=\mathrm{Cst}$, and define a foliation of $\R^3\setminus\bg$.
By Gauss' formula \eqref{eq:link}, the function $\phi(\cdot,\bx_0)$ takes value in the whole torus $\T$. Its formula and the identity \eqref{eq:linear},
imply that $\phi$ is a smooth function outside $\gamma$ with gradient $2\pi\bA$. 
By Sard's lemma it follows that almost every $\theta\in\T$ is a regular value
of $\phi(\cdot,\bx_0):\R^3\setminus\gamma\to \T$, and the corresponding level set defines a surface on $\R^3\setminus\gamma$. 
It is easy to see that the boundary of this surface is $\gamma$ and that the gradient $\nabla\phi(\cdot,\bx_0)$ defines an orientation. 
In other words almost every level set is an oriented surface with boundary $\gamma$. 
There remains to show that they are bounded to conclude that they are Seifert surfaces.
The element $\phi(\infty,\bx_0)$ is well-defined since $\R^3\setminus B(0,R)$ is simply connected and $\bA(\bx)$ decays like $|\bx|^{-3}$ 
at infinity. 
So there is only one level set connected to infinity.

If we pick another Seifert surface $S$ for $\bg$ and a flux $0<\alpha<1$, the gauge transformation connecting $2\pi\alpha\bA$ and $2\pi\alpha[S]$ is obtained
by going from $2\pi\alpha\bA$ to the singular gauge  $2\pi\alpha [\{\phi(\cdot,\bx_0)=\theta_0\}]$ with $\mathrm{exp}(2i\pi\alpha\phi(\cdot,\bx_0))$ and then from
$2\pi \alpha [\{\phi(\cdot,\bx_0)=0\}]$ to $2\pi\alpha[S]$ with the known gauge transformation \cite{dirac_s3_paper1} (using \cite{MR916076} for instance).

\subsection{Behavior of the Biot and Savart gauge for a general knot}
  Let $\gamma:\T_\ell\to\R^3$ be an oriented knot with Seifert surface $S$. 
  Let $(\bT,\bS,\bN_S)$ be the Seifert frame associated with $\gamma\subset S$. 
  We recall \cite{dirac_s3_paper1} that this frame
  defines local coordinates in a tubular neighborhood $B_\eps[\bg]\subset\R^3$ of $\gamma$ by:
\begin{equation}\label{eq:tub_coord}
  F:
  \begin{array}{ccc}
    \T_{\ell}\times [0,\eps)\times \R/(2\pi\Z)&\longrightarrow& B_\eps[\bg]\subset\R^3,\\
    (s,\rho,\theta)&\mapsto& \bg(s)+\rho\big[\cos(\theta)\bS(s)+\sin(\theta)\bN_S(s)\big],
  \end{array}
\end{equation}
   where $\rho$ is the distance to the curve and $\gamma(s)$ is the projection onto the curve.
   We emphasize that these tubular coordinates are defined w.r.t. the flat metric of $\R^3$ and not the metric of $\S^3$
   as it was done in \cite{dirac_s3_paper1}. These coordinates are studied in more details in Section~\ref{sec:estimates}.

  For $\bx\in\R^3$ in the vicinity of $\gamma$, 
  we denote $s(\bx)$ by $s_0$, and according to \eqref{eq:tub_coord} we write:
    $\bx=\gamma(s_0)+\rho V$, 
    where $V=V(\bx)\in (\mathrm{T}_{\gamma(s_0)}\gamma)^{\perp}$ is the vector of unit size
    corresponding to $\bx$. The vector $\bT(s_0)\times V(\bx)$ is denoted by $G=G(\bx)$.
  We also introduce
\begin{equation}\label{eq:def_v_w_2_w_3}
	\begin{array}{rclcrcl}
	v_2(\bx)&:=&\cip{\ddot{\bg}(s_0)}{V},&& w_2(\bx)&:=&\cip{\ddot{\bg}(s_0)}{G},\\
	v_3(\bx)&:=&\cip{\bg^{(3)}(s_0)}{V} ,&& w_3(\bx)&:=&\cip{\bg^{(3)}(s_0)}{G}.
	\end{array}
\end{equation}
In Section~\ref{sec:coord}, we show that these four functions do not depend on $\rho$ in the tubular coordinates. 
We define the angle $\Phi$ as the function 
\begin{equation}\label{eq:def_Phi}
 \Phi(\bx):=\frac{w_2}{2}(\rho\log(\rho)-\rho)+\frac{3}{8}v_2w_2\rho^2\log(\rho).
\end{equation}

Writing $\mathrm{Pf}$ for the finite part, $\wt{\bA}(s_0)$ ($s_0\in\T_{\ell}$) denotes the vector
\begin{equation}\label{eq:def_tilde_A}
	 4\pi\wt{\bA}(s_0):=(\log(2)-\tfrac{1}{2})\dot{\gamma}\times\ddot{\gamma}(s_0)+\mathrm{Pf}\,\int_{\T_{\ell}}\dot{\gamma}(s)\times\frac{\gamma(s_0)-\gamma(s)}{|\gamma(s_0)-\gamma(s)|^3}\d s.
\end{equation}

\begin{theorem}\label{thm:behavior_BS}
	Let $\bx=\bg(s_0)+\rho V\in \R^3$ be a point in a tubular neighborhood $B_\eps[\bg]$ of $\bg$.
	As $\rho\to 0$, we have:
	\begin{multline}\label{eq:formule_behavior}
		4\pi (\bA(\bx)-\wt{\bA}(s_0))=2(\tfrac{1}{\rho}G+\nabla\Phi)+\mathcal{O}(\rho)
		=\big(\tfrac{2}{\rho}+v_2\big)G-\log(\rho)\dot{\gamma}\times\ddot{\gamma}(s_0)\\
		-\rho\log(\rho)\Big[\gamma^{(3)}(s_0)\times V+\tfrac{|\ddot{\gamma}(s_0)|^2}{4}G
		+\tfrac{3}{2}v_2\dot{\gamma}\times\ddot{\gamma}(s_0)\Big]
		+\mathcal{O}(\rho).
	\end{multline}
	Furthermore for $\eps$ small enough, the restriction to $B_\eps[\bg]$ of the one-form $2\pi\bA^\flat-\d(\theta+\Phi)$
	can be continuously extended to $\gamma$ and the extension is exact. 
	That is, there exists a function $g:B_\eps[\bg]\to \R$ which satisfies
	\begin{equation}\label{eq:exact}
	 2\pi\bA^\flat-\d (\theta+\Phi)=\d g.
	\end{equation}
\end{theorem}
The proof of this theorem is split into two. The first part which consists in proving that 
$4\pi(\bA(\bx)-\wt{\bA}(s_0))$ is equal to the last part of \eqref{eq:formule_behavior}
is postponed to Appendix~\ref{sec:first_part_proof_thm_behavior} (this is essentially a calculation). 
We show in Section~\ref{sec:end_proof_behavior} that it also coincides with $2(\tfrac{1}{\rho}G+\nabla\Phi)$
up to an error $\mathcal{O}_{\rho\to 0}(\rho)$ and prove that $\omega:=2\pi\bA^\flat-\d (\theta+\Phi)$ is exact. 

We can define $g$ the following way:
given $\bx,\bx_0\in B_{\eps}[\gamma]$, and any differentiable path $c:[0,1]\to B_\eps[\bg]$ connecting $\bx_0$ to $\bx$, $g(\bx)=g(\bx,\bx_0)$
is given by:
\[
 g(\bx)=\int_c \omega.
\]

In Lemma~\ref{lem:refinement_BS_behav}, we extend this theorem to the smooth gauge given by cable knots surrounding $\gamma$.

\subsection{Proof of the second part of Theorem~\ref{thm:behavior_BS}}\label{sec:end_proof_behavior}
The main goal of this section is to show that the last part of \eqref{eq:formule_behavior} is equal to
$2(\tfrac{1}{\rho}G+\nabla\Phi)+\mathcal{O}_{\rho\to 0}(\rho)$ and to establish \eqref{eq:exact}.
Before the proof we introduce adapted coordinates (Section~\ref{sec:estimates}), then we write the
geometrical objects in these coordinates (section~\ref{sec:geom_local_coord}),
providing us with formulas used for the proof (Section~\ref{sec:proof_end_thm_behav})

\subsubsection{Coordinates}\label{sec:coord}\label{sec:estimates}
Consider the curve $\gamma:\T_\ell\to\R^3$ with Seifert surface $S$,
and Seifert frame $(\bT,\bS,\bN_S)$ (that is the Darboux frame associated to $\bg$ and $S$).
We write $(\bT,\bN_F,\bB_F)$ its Frenet frame, when it is well-defined and write $\kappa_{\bg}:=|\ddot{\bg}|$ its curvature and $\tau_{\bg}$ its torsion
(when it is defined). Its geodesic and normal curvatures and its relative torsion relative to the Seifert surface will be respectively written $\kappa_g,\kappa_n$ and $\tau$,
see \cite[Chapter~7]{Spivakvol4}.

We consider in $B_\eps[\bg]\subset\R^3$ the chart:
\[
F:(s,\rho,\theta)\mapsto \bg(s)+\rho\big[\cos(\theta)\bS(s)+\sin(\theta)\bN_S(s)\big].
\]
As in \cite{dirac_s3_paper1}*{Sec.~3.2.1}, writing:
\[
	h(s,\rho,\theta):=1-\rho\big[\kappa_g(s)\cos(\theta)+\kappa_n(s)\sin(\theta)\big],
\]
we have:
\begin{equation*}
F_*(\partial_s)(\bx)=h\bT+\tau G(\bx),\quad F_*(\partial_{\rho})(\bx)=V(\bx)\quad\&\quad F_*(\partial_{\theta})(\bx)=\rho G(\bx).
\end{equation*}
We will see in the next section that $h=1-\rho v_2$. This implies:
\begin{equation*}
	\d s=\frac{1}{h}\bT^{\flat},\quad\d\rho=V^{\flat}\quad\&
			\quad \d\theta=-\frac{\tau}{\rho h}\bT^{\flat}+\frac{1}{\rho}G^{\flat},
\end{equation*}
and
\begin{equation}\label{eq:nabla_local_coord}
 \nabla=\nabla^{\R^3}=\big(1+\frac{\rho v_2}{h}\big)\bT(\partial_s-\tau\partial_{\theta})+V\partial_{\rho}+G\frac{\partial_{\theta}}{\rho}.
\end{equation}

\subsubsection{Calculation of the geometrical objects in local coordinates}\label{sec:geom_local_coord}
In this section, we write the terms appearing in Theorem~\ref{thm:behavior_BS} in terms of $(s,\rho,\theta)$.
In particular we show that the functions defined in \eqref{eq:def_v_w_2_w_3} only depend on $s,\theta$.
Let $\bx$ in the vicinity of $\gamma$: we write $\bx:=\gamma(s_0)+\rho V(\bx)$.

First, decomposing $\ddot{\gamma}$ w.r.t. $\bS_S$ and $\bN_S$, we have:
\[
\kappa_{\bg}=\sqrt{\kappa_g^2+\kappa_n^2}.
\]
Furthermore, by differentiating $0=\cip{\ddot{\bg}}{V}$ with $V=\bS$ and $V=\bN_S$ we obtain:
\[
\cip{\ddot{\bg}}{\bS}=\kappa_g\quad \&\quad \cip{\ddot{\bg}}{\bN_S}=\kappa_n,
\]
and by further differentiating, we obtain:
\[
\cip{\bg^{(3)}}{\bS}=\dot{\kappa}_g-\tau\kappa_n  \quad \&\quad \cip{\bg^{(3)}}{\bN_S}=\dot{\kappa}_n+\tau\kappa_g.
\]
By differentiating $0=\cip{\ddot{\bg}}{\dot{\bg}}$, we have $\cip{\bg^{(3)}}{\dot{\bg}}=-|\ddot{\bg}|^2=-\kappa_{\bg}^2$.

Decomposing $\ddot{\gamma}$ w.r.t. $V$ and $G$ we have:
\begin{equation}\label{eq:formule_v_w_2}
	\begin{array}{rcl}
	v_2(\bx)&=&\kappa_g(s(\bx))\cos(\theta(\bx))+\kappa_n(s(\bx))\sin(\theta(\bx))\\
	w_2(\bx)&=&-\kappa_g(s(\bx))\sin(\theta(\bx))+\kappa_n(s(\bx))\cos(\theta(\bx)),
	\end{array}
\end{equation}
and
\begin{equation*}
	\dot{\gamma}\times\ddot{\gamma}(s_0)
		=[-\kappa_n\cos(\theta)+\kappa_g\sin(\theta)]V+[\kappa_n\sin(\theta)+\kappa_g\cos(\theta)]G,
\end{equation*}
where $\ddot{\bg}(s_0)$ (rather its parallel transportation along the $\gamma$-orthogonal geodesic
$(\gamma(s_0)+t\rho V(\bx))_{0\le t\le 1}$) is seen as an element of $\rT_{\bx}\R^3$. In particular:
\begin{equation}\label{eq:formule_beta}
	\dot{\gamma}\times\ddot{\gamma}(s_0)=-w_2 V(\bx)+v_2 G(\bx).
\end{equation}

Decomposing $\gamma^{(3)}(s_0)$ w.r.t. $V$ and $G$ we obtain
\begin{equation}\label{eq:formule_v_w_3}
  \begin{array}{rcl}
    v_3(\bx)&=&\cos(\theta(\bx))(\dot{\kappa}_g-\tau\kappa_n)(s(\bx))+\sin(\theta(\bx))(\dot{\kappa}_n+\tau\kappa_g)(s(\bx)),\\
		&=&\partial_s v_2(\bx)-\tau(s(\bx)) \partial_{\theta}v_2(\bx),\\
    w_3(\bx)&=&\cos(\theta(\bx))(\dot{\kappa}_n+\tau\kappa_g)(s(\bx))-\sin(\theta(\bx))(\dot{\kappa}_g-\tau\kappa_n)(s(\bx)),\\
		&=&\partial_s w_2(\bx)-\tau\partial_{\theta}w_2(\bx)=-\cip{\bT(s(\bx))}{\gamma^{(3)}(s(\bx))\times V(\bx)}.
  \end{array}
\end{equation}

The full decomposition of $\gamma^{(3)}(s_0)\times V$ is:
\begin{equation}\label{eq:formule_L}
\gamma^{(3)}(s_0)\times V=-w_3\bT-\kappa_{\bg}^2G=-(\partial_s w_2-\tau \partial_\theta w_2)\bT-\kappa_{\bg}^2G.
\end{equation}

\subsubsection{End of the proof of Theorem~\ref{thm:behavior_BS}}\label{sec:proof_end_thm_behav}

Using \eqref{eq:formule_v_w_2}, \eqref{eq:formule_beta}, \eqref{eq:formule_v_w_3} and \eqref{eq:formule_L},
and writing $\kappa_{\gamma}^2=v_2^2+w_2^2$, we obtain:
\begin{multline*}
W:=v_2G-\log(\rho)\dot{\gamma}\times\ddot{\gamma}(s_0)
		-\rho\log(\rho)\Big[\gamma^{(3)}(s_0)\times V+\tfrac{|\ddot{\gamma}_0|^2}{4}G
		+\tfrac{3}{2}v_2\dot{\gamma}\times\ddot{\gamma}(s_0)\Big]\\
=\rho\log(\rho)w_3\bT+(w_2\log(\rho)-\tfrac{3}{2}v_2w_2)V+(v_2-v_2\log(\rho)+\tfrac{3}{4}(v_2^2-w_2^2))G.
\end{multline*}
It follows from \eqref{eq:nabla_local_coord} that the vector $W$ is equal to:
\[
W=2\nabla \Phi(s,\rho,\theta)+\underset{\rho\to 0}{\mathcal{O}}(\rho).
\]

To end the proof we show that the one-form $\omega:=2\pi\bA^\flat-\d(\theta+\Phi)$ can be continuously extended to $\gamma$,
and then that it is exact on $B_{\eps}[\gamma]$.
This domain is topologically a torus and homotopic to $\gamma$. 
The first (real) de Rham co-homology group of $\gamma\simeq \S^1$
is $\R$ and the class is characterized by the integration of the form along $\gamma$. 
So a one-form $\omega_0$ on $B_\eps[\gamma]$ is exact if and only if $\omega_0$
is closed and satisfies $\int_{\gamma}\omega_0=0$.

First, from the formula of $\bA$ that we established, we have the expansion $\omega(\bx)=2\pi\widetilde{\bA}^{\flat}(s(\bx))+\mathcal{O}_{\rho\to 0}(\rho)$, and $\omega$
can be continuously extended to $\gamma$. 

We know from Section~\ref{sec:gauge_choice} that $2\pi\bA$ is a gradient on $\R^3\setminus\gamma$: there holds 
$2\pi\bA^\flat=\d \phi(\cdot,\bx_0)$ where $\bx_0\in \R^3\setminus\gamma$ (and we choose $\bx_0\in B_{\eps}[\gamma]$). 
So we have $\omega=\d \big(\phi(\cdot,\bx_0)-\theta-\Phi\big)$ on $B_{\eps}[\gamma]$. 
Since it can be continuously extended to $\gamma$, we get that for all contractile loops $c$ on 
$B_{\eps}[\gamma]$, we have $\int_c \omega=0$. In other words $\d\omega=0$.
Let us show $\int_{\gamma}\omega=0$. In local coordinates, we can write:
\begin{equation*}
	2\pi\bA^\flat(\bx)=\big(2\pi\wt{\bA}^\flat(\gamma(s_0))+\tfrac{\tau}{h}\bT^{\flat}\big)+\d(\theta+\Phi)+\underset{\rho\to 0}{\mathcal{O}}(\rho),
\end{equation*}
where 
\[
\d\theta=\frac{1}{\rho}G^{\flat}-\frac{\tau}{h}\bT^{\flat}=\frac{1}{\rho}G^{\flat}-\tau\bT^{\flat}+\underset{\rho\to 0}{\mathcal{O}}(\rho).
\]

Recall \eqref{eq:total_rel_torsion}:
\begin{equation*}
	2\pi\mathrm{Wr}(\gamma)=2\pi\int_{0}^{\ell}\bap(\gamma(s))\d s
					 =-\int_0^{\ell}\tau(s)\d s,
\end{equation*}
where $\bap(\gamma(s))$ denotes $\cip{\bT(\gamma(s))}{\wt{\bA}(\gamma(s))}$. Hence we have $\int_{\gamma}\omega=0$,
which ends the proof.

\subsection{The Biot and Savart gauge for cable knots}\label{BS_cable}

	\subsubsection{Collapse of a cable knot}
	Let $\overline{\gamma_1}$ be a cable knot.
	Let us consider its adapted cable representations $(\gamma_\eta,\gamma_0)$ defined in Section~\ref{sec:intro_adapted_cable_rep},
	with vector map $U:\T_{N\ell}\to\R^3$. We have $\gamma_{\eta}(s)=\gamma_0(s)+\eta U(s)$ for $0<\eta<\eta_0$ and $s\in \T_{N\ell}$.
	For each $\eta$, we write $c_{\eta}$ the arclength parametrisation of $\gamma_{\eta}$, 
	and $u=u_{\eta}$ denotes its arclength parameter.

	The knot $\gamma_{\eta}$ weakly converges to $\gamma_0$ in the following sense. For each interval $I\subset \T_{N\ell}$ with length 
	$0<|I|<\ell$, the local branch $(\gamma_{\eta})\restriction{I}$ converges strongly to the associated branch $(\gamma_0)\restriction{J}$
	where $J$ is the image of $I$ through the congruent map $\R/(N\ell\Z)\to \R/(\ell\Z)$ (in the metric \cite{dirac_s3_paper1}*{Section~4.1} or as a $C^\infty(I,\R^3)$-map).

	For analytical purposes, we choose special cable representations $(\gamma_\eta,\gamma_0)$ that we now describe. With this choice,
	$\Wr(\gamma_\eta)$ has a simple expansion as $\eta\to 0^+$ (which is not the case for other cable representations).

	\subsubsection{Adapted realization of cable knots}\label{sec:intro_adapted_cable_rep}
	We \emph{choose} a (smooth) realization $\gamma_0:\T_{\ell}\to \R^3$ of $\overline{\gamma_0}$ which satisfies
	\begin{equation}\label{eq:cond_gamma_0}
	\Wr(\gamma_0)=\frac{M}{N}.
	\end{equation}
	This choice is always possible, see \cite{dirac_s3_paper2}*{Appendix~A}. Let $S$ be a smooth Seifert surface for $\gamma_0$,
	 $(\bT:=\dot{\gamma}_0,\bS,\bN)$ the associated Darboux frame, and $\tau:\T_\ell\to \R$ be the relative torsion of $\gamma_0$ w.r.t. $S$ 
	 (in the flat metric), see \cite{Spivakvol4}*{Chapter~7}. 
	As shown in \cite{dirac_s3_paper2}*{Appendix~A}, the following two facts hold.
	
	\smallskip
	
	\noindent 1. An application of the C\u{a}lug\u{a}reanu-White-Fuller Theorem  \cite{Calu61,White69}
	gives
	\begin{equation}\label{eq:total_rel_torsion}
	I_\tau(\gamma_0):=\int_0^{\ell}\tau(s)\d s=-2\pi\Wr(\gamma_0).
	\end{equation}
	
	\noindent 2. Let $U$ be the parallel transportation of $\bS(0)$ along $\gamma_0$
	w.r.t. the canonical connection on the normal bundle $\cN\gamma_0$: for $s_1\in (0,\ell)$, the differential of the measured angle $\measuredangle \big(\bS(s),U(s)\big)$
	in  $(\T_{s_1}\gamma_0)^\perp$ (oriented by $\dot{\gamma}_0(s_1)$) is $-\tau(s)\d s$.

	\smallskip
	
	After one turn, $U$ has rotated from its initial position: we have 
	$U(s+\ell)=R_{s,2\pi M/N}U(s)$, where $R_{s,\theta}$ is the rotation of positive axis $\dot{\gamma}_0(s)$
	and angle $\theta$. It comes back at its initial position after $N$ turns.
	
	For $0\le s\le N\ell$, we denote by $V(s)$ the vector $\dot{\gamma}_0(s)\times U(s)$, and by
	$I_0(s)$ the integrated torsion of the curve $\gamma_0$ w.r.t. the Seifert surface $S$:
	\begin{equation}\label{eq:def_I_0}
	 I_0(s):=\int_0^{s}\tau(s')\d s'.
	 \end{equation}
	 We can write $U$ and $V$ in terms of $I_0$:
	\begin{equation}\label{eq:def_U_V}
		\begin{array}{rcl}
			U(s)&=&\cos(I(s))\bS(s)-\sin(I(s))\bN(s),\\
			V(s)&=&\sin(I(s))\bN(s)+\cos(I(s))\bN(s).
		\end{array}
	\end{equation}
	
	We \emph{choose} this $U$ in $\gamma_{\eta}$ \eqref{eq:def_gamma_eta}:
	it is a cable representation of $\gamma_1$. In other words, we will consider a realization:
	\[
		\gamma_{\eta}(s):=\gamma_0([s])+\eta U(s),\ 0\le s\le N\ell,
	\]
	with \eqref{eq:cond_gamma_0} and $U$ given as above (and base-point $\gamma_0(0)+\eta\bS(0)$). 
	We will call such a realization $(\gamma_\eta,\gamma_0)$ an \emph{adapted cable representation} of $\overline{\gamma_1}$. 
	Its nice geometrical properties (studied in Section~\ref{sec:prop_adap_seif_fib}) follow from the fact that $\dot{\gamma}_\eta(s)$ and $\dot{\gamma}_0(s)$ are co-linear.

	\subsubsection{Properties of the adapted realizations}\label{sec:prop_adap_seif_fib}

	Recall that $I_0$ is defined in \eqref{eq:def_I_0}.
	
	\smallskip
	
	\paragraph{\emph{Parallel tangent vectors}}
	Let $m=m(\eta,s),\ 0\le s\le N\ell$ be the function defined by: 
	\begin{equation}\label{eq.def_m}
		m(\eta,s):=h(s,\eta,-I_0(s))=1-\eta(\kappa_g(s)\cos(-I_0(s))+\kappa_n(s)\sin(-I_0(s))).
	\end{equation}
	Then we have:
	\begin{equation*}
		\dot{\gamma}_{\eta}(s)=m(\eta,s)\dot{\gamma}_0(s)\in\rT_{\gamma_0(s)}\gamma_0.
	\end{equation*}
	Thus, up to normalization, the tangent vectors $\dot{\gamma}_{\eta}(s)$ 
	define a vector field in $B_{\eps}[\gamma_0]$ which is precisely
	the extension of $\dot{\gamma}_0$ by parallel transportation (in the flat metric) 
	along $\gamma_0$-orthogonal geodesics.
	
	\smallskip
	
	\paragraph{\emph{Length}}
	From  \eqref{eq:hey_root_guy} and \eqref{eq.def_m} we obtain
	that the $\gamma_{\eta}$'s all have length $N\ell$. 
	
	\smallskip
	\paragraph{\emph{Convergence of tubular coordinates}}
	Moreover, there holds convergence of the tubular coordinates \eqref{eq:tub_coord} of the cable knot $\gamma_{\eta}$
	to those of the base-curve $\gamma_0$ in the following sense. 
	
	In the vicinity of a point $\gamma_0(s)\in\gamma_0$, the cross-section
	$B_\eps[\gamma_0]\cap (\gamma_0(s)+(\dot{\gamma}_0(s))^\perp)$ intersects $N$ times the curve $\gamma_{\eta}$ corresponding
	to $N$ different local branches of the cable knot. 
	
	Each of this branch together with
	the (parallel transported) Seifert frame $(\bT,\bS,\bN)$ defines tubular coordinates \eqref{eq:tub_coord}.
	For any $\eps_1\in (0,\eps)$, the tubular coordinates of any branch converge to that of $\gamma_0$ in
	the space $C^\infty(B_\eps[\gamma_0(s)]\cap (B_{\eps_1}[\gamma_0])^c)$. We have fixed a basepoint for $\gamma_0$
	fixing that of $\gamma_{\eta}$ (as in \cite{dirac_s3_paper1}*{Proposition~21}). 
	Then we used the congruent map $\R/(N\ell\Z)\to \R/(\ell\Z)$ to compare the arclength parameters of $\gamma_{\eta}$ and $\gamma_0$.
	The proof is similar to that of  \cite{dirac_s3_paper1}*{Proposition~21} and left to the reader.
	
	\medskip
	
	\paragraph{\emph{Writhe}}
	We can easily estimate the writhe of $\gamma_{\eta}$.
	\begin{lemma}\label{lem:writh_cabled_knot}
	 Let $(\gamma_{\eta},\gamma_0)$ be an adapted cable representation for $\overline{\gamma_1}$.
	 Then as $\eta\to 0$ we have:
	 \[
	  \Wr(\gamma_{\eta})=N^2\Wr(\gamma_0)+\underset{\eta\to 0}{\mathcal{O}}(\eta).
	 \]
	\end{lemma}
	The proof is a corollary of Lemma~\ref{lem:refinement_BS_behav} and given in Section~\ref{sec:behavior_cable}.

\subsubsection{Study of the curves $\gamma_{\eta}$}\label{sec:study_cable}

Fix $\eta\in (0,\eta_0)$. Let $u(s)$ be the arclength and $c_{\eta}$ be the arclength parametrisation of $\gamma_{\eta}$:
\begin{equation*}
	u(s):=\dint_0^{s}m(\eta,s')\d s'.
\end{equation*}
For $u\in\T_{N\ell}$, we consider $W(u)\subset \T_{N\ell}$ the set of its ``neighbors": for all $v\in W(u)$, there exists
$j\in \Z$ such that
\begin{equation*}
s(v)=s(u)+j\ell,
\end{equation*}
where $s(v),s(u)\in \T_{N\ell}$ are the parameters such that
\[
\gamma_{\eta}(s(u))=c_{\eta}(u)\quad\&\quad \gamma_{\eta}(s(v))=c_{\eta}(v).
\]

 Denoting the differentiation w.r.t. $s$ and $u$ by dots and primes respectively, we have:
\begin{multline}\label{eq:calcul_c}
	c_{\eta}'(u)=\bT(s(u)),\ c_{\eta}''(u)=\frac{1}{m(\eta,s(u))}\ddot{\gamma}_0(s(u))\\
	 \&\ c_{\eta}'''(u)=\frac{1}{m(\eta,s(u))^2}\bigg(\gamma_0^{(3)}(s(u))-\frac{\dot{m}(\eta,s)}{m(\eta,s)}\ddot{\gamma}_0(s(u))\bigg).
\end{multline}

\subsubsection{The complex polynomial $P_{\eta}$}\label{sec:complex_poly}
In this section we introduced functions that will be thoroughly used in later sections.
They all depend on the integers $N$, but to simplify notations we will not emphasize this dependence.

Let $\zeta$ be the $N$-th root
	of the unity:
	\begin{equation*}
	 \zeta:=e^{2i\pi\tfrac{M}{N}}.
	\end{equation*}
	We emphasize the equality
	\begin{equation}\label{eq:hey_root_guy}
	 \sum_{k=0}^{N-1}\zeta^k=\frac{1-\zeta^N}{1-\zeta}=0.
	\end{equation}

For $\eta>0$, we write $Z_\eta$ the
multipoint defined as the collection of all the $N$-th roots of $\eta^N$, and we call $P_\eta\in\C[z]$
the complex polynomial 
\begin{equation*}
	P_\eta(z):=z^N-\eta^N=\prod_{k=0}^{N-1}(z-\eta\zeta^k)=\prod_{\xi\in Z_{\eta}}(z-\xi).
\end{equation*}
 
Let $Q_\eta$ and $\Theta_\eta\in\R/(2\pi\Z)$  
 denote the polar coordinates of $P_\eta$:
 \begin{equation*}
 	Q_\eta(z):=|P_\eta(z)|\quad\&\quad e^{i\Theta_\eta(z)}:=\frac{P_\eta(z)}{|P_\eta(z)|}.
 \end{equation*}
 The level sets of $Q_\eta$ are called \emph{Cassinians} and those of $\Theta_\eta$ \emph{stelloids}.

 The Seifert frame $(\bT,\bS,\bN)$ gives a trivialization of $\cN\gamma_0$ 
 (realized by the tubular coordinates \eqref{eq:tub_coord}). Up to identifying $\R^2$ with $\C$,
 we can identify the normal planes of $\gamma_0$ with $\C$. The intersection of $\gamma_{\eta}$
 with the cross-section $\{s(\bx)=0\}$ corresponds to $Z_\eta\subset \C$, and as $s$ runs over $\T_{\ell}$, there is a rotation
 of this intersection on $\C$. On the cross section $\{s(\bx)=s_0\}$ the intersection corresponds to
 $e^{-iI_0(s_0)}Z_\eta$ ($I_0$ is defined in \eqref{eq:def_I_0}).

 We introduce the notations on $B_\eps[\gamma_0]$
 \begin{equation*}
 	z(\bx):=\rho(\bx) e^{i\theta(\bx)}\quad\&\quad z_k(\bx):=\rho(\bx) e^{i\theta(\bx)}-\eta \zeta^k e^{-iI_0(s(\bx))}.
 \end{equation*}

 In the formula above, $s(\bx)\in \T_{\ell}$. 
Considering the tubular coordinates with $s=s(\bx)\in [0,\ell)$, 
we define $p_\eta,q_\eta,\theta_\eta$ on $B_\eps[\gamma_0]$ by:
 \begin{equation}\label{eq:def_fun_eta}
 	\begin{array}{ccl}
 	p_\eta(\bx)&:=&\prod_{k=0}^{N-1}z_k(\bx)=e^{-iNI_0(s)}P_\eta(z(\bx)e^{iI_0(s)}),\\
	q_\eta(\bx)&:=& |p_\eta(\bx)|=Q_\eta(z(\bx)e^{iI_0(s)}),\\
	\mathrm{exp}(i\theta_\eta(\bx))&:=&\frac{p_\eta(\bx)}{|p_\eta(\bx)|}=e^{-iNI_0(s)}\mathrm{exp}(i\Theta_\eta(z(\bx)e^{iI_0(s)})).
	\end{array}
 \end{equation}
 In particular we have $q_\eta(\bx)^2=\rho^{2N}+\eta^{2N}-2\rho^N\cos(N(\theta+I_0(s)))$.

 \subsubsection{Behavior of the B.S. gauge for cable knots}\label{sec:behavior_cable}

 In this section, we improve Theorem~\ref{thm:behavior_BS} for the behavior of the Biot and Savart gauge
 on $B_\eps[\gamma_0]$ corresponding to $2\pi\alpha [\gamma_{\eta}]$ (Lemma~\ref{lem:refinement_BS_behav}). 
 As above, let $\bA_{\gamma_{\eta}}$ defined by
 \[ 
 \bA_{\gamma_{\eta}}(\bx)=\frac{1}{4\pi}\int_{\gamma_{\eta}}\d \mathrm{r}\times \frac{\bx- \mathrm{r}}{|\bx- \mathrm{r}|^3}.
 \]
 Then $2\pi\alpha \bA_{\gamma_{\eta}}$ is a magnetic potential for $2\pi\alpha [\gamma_{\eta}]$. 
 
 Let $\bx\in B_\eps[\gamma_0]$. In the tubular coordinates, we have $\bx=\gamma_0(s)+\rho V(\bx)$. The curve $\gamma_{\eta}$
 intersects the cross-section $s=s(\bx)$ at the $N$ points $\gamma_{s,\eta}(s+k\ell)$, $0\le k\le N-1$, and we write $u_k(s(\bx))$ the corresponding
 points in the arclength parametrisation $c_{\eta}$ of $\gamma_{\eta}$. For $0\le k\le N-1$, we define $V_k(\bx)$ as the unit vector such that
 \[
 	\bx-\gamma_{\eta}(s(\bx)+k\ell)=\rho_k V_k(\bx),\ \rho_k=|\bx-\gamma_{\eta}(s(\bx)+k\ell)|,
 \]
 and $G_k(\bx)$ denotes $\bT(s(\bx))\times V_k(\bx)$. We recall that $G(\bx)=\bT(s)\times V(\bx)$. 
 Using \eqref{eq:def_v_w_2_w_3} and \eqref{eq:calcul_c}, we define:
\begin{equation*}
	\begin{array}{rclcrcl}
	v_{2,k}(\bx)&:=&\cip{c_{\eta}''(u_k(s(\bx)))}{V_k(\bx)},&& w_2(\bx)&:=&\cip{c_{\eta}''(u_k(s(\bx)))}{G_k(\bx)},\\
	v_{3,k}(\bx)&:=&\cip{c_{\eta}'''(u_k(s(\bx)))}{V_k(\bx)} ,&& w_3(\bx)&:=&\cip{c_{\eta}'''(u_k(s(\bx)))}{G_k(\bx)},
	\end{array}
\end{equation*}
and on $B_\eps[\gamma_0]$, we denote by $\Phi_\eta$ the \emph{differentiable} function
\begin{equation*}
	\Phi_\eta(\bx):=\sum_{k=0}^{N-1} \Big(\frac{w_{2,k}}{2}(\rho_k\log(\rho_k)-\rho_k)+\frac{3}{8}v_{2,k}w_{2,k}\rho_k^2\log(\rho_k)\Big).
\end{equation*}
	
	Recall the formula of the gradient in tubular coordinates \eqref{eq:nabla_local_coord}.
 Using the chain rule in \eqref{eq:def_fun_eta} and
 the identity $(\partial_s-\tau\partial_{\theta})e^{i(\theta+I_0(s))}=0$, we obtain:
 \begin{equation*}
 	\begin{array}{rcl}
 	\nabla q_\eta&=& \dfrac{N}{|(\rho e^{i\theta})^N-\eta^N|}\Big\{ \big(\rho^{2N-1}-\rho^{N-1}\cos[N(\theta+I_0(s))]\big)V\\
				&&\quad+\rho^N\sin[N(\theta+I_0(s))]G\Big\},\\
	\nabla \theta_\eta&=& -N\frac{\tau}{h}\bT+\sum_{k=0}^{N-1}\frac{1}{\rho_k}G_k.
	\end{array}
 \end{equation*}

\begin{lemma}\label{lem:refinement_BS_behav}
	Let $(\gamma_{\eta},\gamma_0)$ be an adapted cable representation as explained in Section~\ref{sec:intro_adapted_cable_rep},
	and let $\bA_{\gamma_{\eta}}$ denotes the Biot and Savart gauge w.r.t. $\gamma_{\eta}$. 
	Let $g_0$ and $\Phi_0$ be the function of Theorem~\ref{thm:behavior_BS} associated with $\gamma_0$.
	Then on the tubular neighborhood
	$B_\eps[\gamma_0]$, the $1$-form $\omega_{\eta}\in \Omega^1(B_\eps[\gamma_0]\setminus \gamma_0)$:
	\[
		\omega_{\eta}:=2\pi \bA_{\gamma_{\eta}}^\flat-\d (\theta_\eta+\Phi_{\eta}),
	\]
	can be continuously extended to $B_\eps[\gamma_0]$ as an \emph{exact} $1$-form $\d g_\eta$ on the tubular neighborhood. 
	
	Furthermore $g_\eta$ varies continuously with $\eta>0$, and converges almost everywhere and uniformly on any set $\{0<\eps_1\le \rho <\eps\}$
	to $Ng_0$. The functions $g_\eta$ and $g_0$ are set to be equal to $0$ on $\gamma_0(0)$.
	The same holds for the convergences $\theta_\eta\to N\theta$ and $\Phi_{\eta}\to N\Phi_0$.
\end{lemma}

 \begin{rem}
	This Lemma~\ref{lem:refinement_BS_behav} can be easily extended to other cable representations $(\gamma_1,\gamma_0)$. 
\end{rem}

\begin{proof}
	We follow the same proof as that of Theorem~\ref{thm:behavior_BS}. Let $\bx=\gamma_0(s)+\rho V(\bx)$
	and $0<\eta_0<\tfrac{\ell}{2}$. Let $\{u_k(s),\,0\le k\le N-1\}$ be the corresponding points in the arclength parametrisation of $\gamma_{\eta}$: 
	\[
		\int_0^{s+k\ell}|\dot{\gamma}_{\eta}(s')|\d s'=u_k(s).
	\]
	If we read the first part in Appendix \ref{sec:first_part_proof_thm_behavior},
	one sees that the integration over $[-\eps_0+u_k,\eps_0+u_k]$ in the formula of $\bA_{\gamma_{\eta}}$ gives rise
	to the same expansion as in \eqref{eq:formule_behavior}, with $s_0,\rho,V,G$ replaced by $u_k,\rho_k,V_k,G_k$
	respectively. This shows that $\omega_\eta$ can be continuously extended on $B_\eps[\gamma_0]$. 
	As in Section~\ref{sec:end_proof_behavior}, we obtain that $\d \omega_\eta=0$. It is exact if and only if $N\int_{\gamma_0}\omega_\eta=0$.
	By homotopy, this is equivalent to $\int_{\gamma_{\eta}}\omega_\eta=0$, and indeed:
	\[
		\int_{\gamma_{\eta}}\omega_\eta=2\pi \Wr(\gamma_{\eta})+N\int_0^{\ell} \frac{\tau(s)}{h(s,\eta,-I_0(s))}h(s,\eta,-I_0(s))\d s=0.
	\]
	We have used an approximation argument to obtain
	\[
		-\lim_{r\to \eta^+}\int_{\gamma_{r}}\d (\theta_{\eta}+\Phi_{\eta})=\int_{\gamma_{\eta}} N\frac{\tau}{h}\bT^{\flat},
	\]
	using the fact that $\dot{\gamma}_{r}$ is orthogonal to the $G_k$'s, and $\int_{\gamma_{r}}\d \Phi_{\eta}\equiv 0$.
	Alternatively by the Seifert surface crossing characterization of the linking number \eqref{eq:link},
	in the limit $r\to \eta^{-}$ we have
	\[
	\frac{1}{2\pi}\int_{\dot{\gamma}_{r}}\d \theta_{\eta}= \link(\gamma_{\eta},\gamma_{0})= 2\pi\int_{\gamma_{r}}\bA_{\gamma_{\eta}}^\flat.
	\]

	The convergence of $g_\eta,\theta_\eta$ and $\Phi_{\eta}$ follows from the convergence of the tubular coordinates $(\gamma_{\eta}(u_k),\rho_k,\theta_k)$ to 
	$(\gamma_0(s),\rho,\theta)$, $0\le k\le N-1$.
\end{proof}

 \begin{proof}[Proof of Lemma~\ref{lem:writh_cabled_knot}]

	For a given $\gamma_0(s)$, there corresponds $N$ points on $\gamma_{\eta}$,
	with parameters $u_k$, $0\le k\le N-1$. We use the integral formula of the writhe for $\gamma_{\eta}$, and consider 
	$\cip{\bA_{\gamma_{\eta}}^{\flat}(u)}{c_{\eta}'(u)}$ (which is well-defined as a limit). 
	We focus on the integral formula for $\bA_{\gamma_{\eta}}(\bx)$ and $\bA_{\gamma_0}(\bx)$: 
	\[
	 \bA_{\gamma_{\eta}}(\bx)=\int_{\T_{N\ell}}f_{\eta}(\bx,u)\d u\quad\&\quad \bA_{\gamma_0}(\bx)=\int_{\T_{\ell}}f_{0}(\bx,s)\d s.
	\]
	We write $\cI_\eta\{\bx,J_1\}$ resp. $\cI_0\{\bx,J_2\}$ the partial integrals over the subsets $J_1$ and $J_2$.
	At the point $u_k$, $\cI_\eta\big\{c_{\eta}(u_k),\T_{N\ell}\setminus (\cup_{k'}[u_{k'}-\eps,u_{k'}+\eps])\big\}$ differs from the quantity 
	$N\cI_0\big\{\gamma_0(s),\T_{\ell}\setminus [s-\eps,s+\eps]\big\}$ by a  $\mathcal{O}_{\eta\to 0}(\eta)$. As shown in the proof of Lemma~\ref{lem:refinement_BS_behav}, for all $k'\neq k$ we have:
	\[
	\big\langle c_{\eta}'(u_k),\, \cI_\eta\big\{c_{\eta}(u_k),[u_{k'}-\eps,u_{k'}+\eps]\big\} \big\rangle=-w_{3,k'}(c_{\eta}(u_k))\rho_{kk'}\log(\rho_{kk'})+\mathcal{O}_{\eta\to 0}(\eta),
	\]
	where $\rho_{kk'}$ denotes $|c_{\eta}(u_k)-c_{\eta}(u_{k'})|$. Using Formula \eqref{eq:formule_v_w_3}, but for the tubular coordinates of $c_{\eta}$,
	and the fact that $c_{\eta}^{(3)}(u_k)=\gamma_0^{(3)}(s)+\mathcal{O}_{\eta\to 0}(\eta)$, we get cancellation and the following quantity is $\mathcal{O}_{\eta\to 0}(\eta)$:
	\[
	\big\langle c_{\eta}'(u_k),\,  \cI_\eta\big\{c_{\eta}(u_k),[u_{k'}-\eps,u_{k'}+\eps]\big\}\big\rangle +   \big\langle c_{\eta}'(u_{k'}),\, \cI_\eta\big\{c_{\eta}(u_{k'}),[u_{k}-\eps,u_{k}+\eps]\big\}\big\rangle.
	\]
	Summing all the terms yields the announced convergence.
\end{proof}

\section{Dirac operators with magnetic cable knots}\label{sec:descr_dirac_op}

Let $(\gamma_\eta,\gamma_0)$ be an adapted cable representation with running cable knot $\gamma_\eta$ collapsing to $\gamma_0$,
and consider the corresponding magnetic field with a given flux $0<2\pi\alpha<2\pi$. 
We add to it an auxiliary magnetic field as explained in Section \eqref{sec:strategy}, with total gauge $\bA_{\full}$.
In the vicinity of $\gamma_0$, the action of the Dirac operator $\cD_{\bA_{\full}}$ is locally decomposed 
into the derivative colinear to the Seifert fibration and the derivatives transversal to it.

We start with the transversal part and study the 2D-Dirac operator with $N$ Aharonov-Bohm solenoids placed on 
the $N$-th roots of $\eta^N$. In particular we deal with the convergent issues as $\eta\to 0$.

Then we describe more thoroughly the domain of $\cD_{\bA_{\full}}$ (Dirac operator in the $\S^3$-metric) and relate it to a model operator.
The important relation is given in Lemma~\ref{lem:graph_norm_estimate}. We use this lemma to decompose an element of the 3D Dirac operator
into a regular and a singular part (according to the decomposition of the model operator, Lemma~\ref{lem:decomp_sing_subspace}).

\subsection{$2$D operator with $N$ Aharonov-Bohm solenoids of same flux}\label{sub:2dim_op}
As the curve $\gamma_{\eta}$ collapses onto $\gamma_0$, it intersects each cross-section at $N$ points.
This leads us to study -- for $\eta\in(0,\eps)$ and $\ 0<\alpha<1$-- the two-dimensional Dirac operators with magnetic field:
\[
2\pi\alpha \sum_{k=0}^{N-1}\delta_{\eta\zeta^k},
\]
where  $\zeta=e^{2i\pi M/N}$.  
Its scalar gauge is:
$
\alpha \sum_{k=0}^{N-1}\nabla\theta(\cdot-\eta\zeta^k).
$
\subsubsection{Description of the domains and convergence result}
Let us write $\cD_{\eta,\alpha}=\cD_{\eta,\alpha}^{(-)}$ the associated Dirac operator \cite{Persson_dirac_2d} 
(see also \cite{dirac_s3_paper3}). We recall that the minimal domain $\dom(\cD_{\eta,\alpha}^{(\min)})$
is the closure in the graph norm of $C^{\infty}_0(\C\setminus\{\eta\zeta^k,\,0\le k\le N-1\})^2$, and the minimal domain has co-dimension $N$
in the full domain $\dom(\cD_{\eta,\alpha}^{(-)})$.

Let $\chi\in C^{\infty}_0(\C,[0,1])$ be a radial function which equals one on the unit disk $\mathbb{D}$ and vanishes outside $2 \mathbb{D}$.
For $0\le k\le N-1$, let $\chi_{k,\eta}$ be the function 
\[
\chi_{k,\eta}(z):=\chi\big(\tfrac{N}{\eta}(z-\eta \zeta^k)\big).
\]
The function $\chi_{k_0,\eta}$ equals $1$ around $\eta\zeta^{k_0}$ and vanishes around the other $\eta\zeta^k$'s.
Using \cite{Persson_dirac_2d}, the domain of $\cD_{\eta,\alpha}$ can be given as the sum:
\[
	\dom(\cD_{\eta,\alpha}^{(\min)})+
		\mathrm{Span}\bigg\{ \chi_{k,\eta}(z)\begin{pmatrix} 0\\ |z|^{-\alpha}\end{pmatrix},\ 0\le k\le N-1\bigg\},
\]
Using the fact that for all $0\le k_0\le N-1$, the polynomial $\frac{z^N-\eta^N}{z-\eta \zeta^{k_0}}$ has degree $N-1$
and vanishes on the $\eta\zeta^k$'s but $\eta \zeta^{k_0}$, we can rewrite this domain as
\begin{equation*}
	\dom(\cD_{\eta,\alpha}^{(-)})=
		\dom(\cD_{\eta,\alpha}^{(\min)})+\mathrm{Span}\bigg\{ \chi(z)\overline{z}^k\begin{pmatrix} 0\\ Q_{\eta}(z)^{-\alpha}\end{pmatrix},\ 0\le k\le N-1\bigg\},
\end{equation*}
where $Q_\eta(z)=|z^N-\eta^N|$. Let us write $g_{\eta,k}$ the $N$ vectors defining the additional span in the expression above:
\begin{equation}\label{eq:def_psi_k_eta}
	g_{\eta,k}(z)=g_k(z):=\chi(z)\overline{z}^k\begin{pmatrix} 0\\ Q_{\eta}(z)^{-\alpha}\end{pmatrix}\in L^2(\C)^2,\ 0\le k\le N-1.
\end{equation}

A computation shows the following:
\begin{equation}\label{eq:prop_g}
	\cD_{\eta,\alpha}g_0=\sigma_{\R^2}\cdot(-i\nabla^{\R^2}\chi)\begin{pmatrix} 0\\ Q_{\eta}(z)^{-\alpha}\end{pmatrix}
		\ \& \ \cD_{\eta,\alpha}g_{k}=\overline{z}^k\cD_{\eta,\alpha}g_0.
\end{equation}

Similarly for $m\in\Z$, we write $\cD_{0,\alpha}^{(m)}$ the Dirac operator with magnetic field $2\pi\alpha\delta_0$ with the gauge $(\alpha+m)\nabla\theta$.
For $m=0$, we obtain the Dirac operator defined in \cite{Persson_dirac_2d}, and $e^{-im\theta}$ is the gauge transformation mapping $\dom(\cD_{0,\alpha}^{(0)})$
onto $\dom(\cD_{0,\alpha}^{(m)})$.

\begin{lemma}\label{lem:strg_res_cont_2D}
	For $0< \alpha < 1$, we decompose $N\alpha$ into its integer and fractional parts $E(N\alpha):=\lfloor N\alpha \rfloor$ and  
	$e(\alpha)=N\alpha-E(N\alpha)$.
	
	As $\eta\to 0$, $\cD_{\eta,\alpha}$ converges to $\cD_{0,e(N\alpha)}^{(E(N\alpha))}$ in the strong-resolvent sense.
\end{lemma}

Note that the convergence is qualitatively different for $0<\alpha<\tfrac{1}{N}$ and for $\alpha\ge \tfrac{1}{N}$.
More generally as $\alpha$ increases, the convergence changes  at the $N$ values $\alpha_k:=\tfrac{k+1}{N}$, $0\le k\le N-1$.
For $\alpha\ge \alpha_k$, the complex line $\C (g_{\eta,k},\cD_{\eta,\alpha}g_{\eta,k})$ collapses onto $0$ (that is the corresponding unit vector
concentrates onto $0$ and converges $L^2$-weakly to $0$). Hence, for $\alpha\ge \frac{1}{N}$, there cannot be norm-resolvent convergence.

\begin{proof}

First, $\alpha\nabla \vartheta_{\eta}$ converges to $N\alpha \nabla\theta$ 
in $L^2_{\loc}(\C\setminus \{0\})^2$, and uniformly on each $\C\setminus D_\C(0,r)$, $r>0$.
For $0\le k\le N-1$, the same holds for
\begin{equation*}
	\begin{array}{rcc}
		g_{\eta,k}&\underset{\eta\to 0}{\longrightarrow}& \chi(z)\frac{\overline{z}^k}{|z|^{N\alpha}}\begin{pmatrix} 0& 1\end{pmatrix}^T,\\
		\cD_{\eta,\alpha}g_{\eta,k}&\underset{\eta\to 0}{\longrightarrow}& 
			\sigma_{\R^2}\cdot(-i\nabla^{\R^2}\chi)(z)\frac{\overline{z}^k}{|z|^{N\alpha}}\begin{pmatrix} 0& 1\end{pmatrix}^T.
	\end{array}
\end{equation*}
For $N\alpha-k<1$, this shows that $(g_{\eta,k},\cD_{\eta,\alpha}g_{\eta,k})$ converges in $L^2(\C)^2\times L^2(\C)^2$ to
$(g_{0,k},\cD_{0,e(N\alpha)}^{(E(N\alpha))}g_{0,k})$. 
Using the second characterisation of Lemma~\ref{lem:char_sres_conv}, strong-resolvent convergence is easy to establish:
the details are left to the reader.
\end{proof}

\subsubsection{Adapted decomposition of the domain}\label{sec:decomp_2d_dom}
  It is known that $\dom(\cD_{\eta,\alpha})$ can be decomposed into $\dom(\cD_{\eta,\alpha}^{(\min)})$
  and its graph norm-orthogonal complement $D_{\sing}$ (the latter subspace has dimension $N$).
  Following this decomposition, an element of the domain can be split into a regular part and a singular part. 
  Here, we introduce an adapted basis of $D_{\sing}$ which is semi-explicit.

  Let $P_{\min}$ be the graph-norm orthogonal projection onto $\dom(\cD_{\eta,\alpha}^{(\min)})$.
  Recall $\zeta:=\mathrm{exp}(2i\pi\tfrac{M}{N})$. Consider the family $\underline{g}=(g_k)_{0\le k\le N-1}$ defined in \eqref{eq:def_psi_k_eta}. 
  This family satisfies the property $g_k(\zeta z)=\overline{\zeta}^kg_k(z)$ and \eqref{eq:prop_g},
  and so it is already orthogonal with respect to the graph-norm
  of $\cD_{\eta,\alpha}$. We construct the family $\underline{f_{\sing}}=(f_{\sing,k})_{0\le k\le N-1}$ by removing the regular part from $\underline{g}$:
  \begin{equation}\label{eq:def_f_sing_2D}
   f_{\sing,k}:= \frac{(1-P_{\min})g_k}{\norm{(1-P_{\min})g_k}_{\cD_{\eta,\alpha}}}.
  \end{equation}
  The family  $\underline{f_{\sing}}$ is an orthonormal basis of $D_{\sing}$, and it
  turns out that it satisfies the same properties as $\underline{g}$.
  \begin{lemma}
   For $0\le k\le N-1$, $f_{\sing,k}$ has no spin up component and satisfies $f_{\sing,k}(\zeta z)=\overline{\zeta}^kf_{\sing,k}(z)$
   for almost all $z\in\C$.
  \end{lemma}
  \begin{proof}
  Let $\Omega_\eta$ be the open set $\C\setminus \{\eta\zeta^k,\,0\le k\le N-1\}$.
   Observe that the $g_k$'s only have a spin down component, hence are already $\norm{\cdot}_{\cD_{\eta,\alpha}}$-orthogonal to 
   $H^1_0(\Omega_\eta)\otimes\begin{pmatrix} 1& 0\end{pmatrix}^T$, the spin-up part of the minimal domain.
   
   The graph-inner product restricted to 
   $
   H^1_0(\Omega_\eta)\otimes\begin{pmatrix} 0& 1\end{pmatrix}^T\oplus D_{\sing}
   $
   corresponds to the scalar inner product defined by the graph-norm of $-2i\partial_z$ acting on the spin down component.
   Consider the \emph{orthogonal} decomposition $L^2(\C)=\oplus_{0\le k\le N-1}L_{\overline{\zeta}^k}^2(\C)$ where:
   \[
   L_{\overline{\zeta}^k}^2(\C):=\{ f\in L^2(\C),\ \mathrm{for\ almost\ all\ }z,\ f(\zeta z)=\overline{\zeta}^kf(z)\}.
   \]
   The orthogonality follows from the polar decomposition of $L^2(\C)$ into the tensor $L^2(\R_{\ge 0},r\d r)\otimes L^2(\S^1,\d\theta)$ and the $N$-folding
   of $L^2(\S^1,\d\theta)$ (identified with $L^2(\R/(2\pi\Z),\d\theta)$, see section \eqref{sec:N_fold}).
   This decomposition naturally extends to $H^1_0(\Omega_\eta)$ and to $D_{\sing}$, and is also $\norm{\cdot}_{-2i\partial_z}$-orthogonal.
   Thus we can decompose $H^1_0(\Omega_\eta) \oplus D_{\sing}$ in a similar way. As 
   $\cip{\begin{pmatrix} 0 & 1\end{pmatrix}^T}{g_k}_{\C^2}$ is in $(H^1_0(\Omega_\eta) \oplus D_{\sing})\cap L_{\overline{\zeta}^k}^2(\C)$, we obtain
   $f_{\sing,k}\in L_{\overline{\zeta}^k}^2(\C)^2$.
   \end{proof}

Fix $\alpha$. As $\eta\to 0$, the element $(g_{\eta,k},\cD_{\eta,\alpha}g_{\eta,k})$ behaves differently depending on whether $N\alpha-k<1$ or $N\alpha-k\ge 1$ as we have seen in the proof of Lemma~\ref{lem:strg_res_cont_2D}. In the first case it converges in $L^2$ to an element of the graph of the limit operator, which is also in the graph of the corresponding \emph{minimal} operator.
In the latter case, either it converges in the singular part of the limit operator for $k=E(N\alpha)$, or its $L^2$-norm explodes to $+\infty$ while
its energy remains bounded. The different behaviors have the following impact on the basis $\underline{f_{\sing}}$. 
Below, $K_a$, for $0<a<1$, denotes the modified Bessel function of the second kind and $C_a$ the constant
	\[
	C_a:= 2\pi\int_{r=0}^{+\infty}r(K_a(r)^2+K_{1-a}(r)^2)\d r.
	\]
\begin{lemma}\label{lem:dichotomy_sing_basis}
	Let $\underline{f_{\sing}}$ be the $\norm{\cdot}_{\cD_{\eta,\alpha}}$-orthonormal family \eqref{eq:def_f_sing_2D}.
	As $\eta\to 0$ the following holds.
	\begin{enumerate}
		\item If $k>E(N\alpha)$, or $k=E(N\alpha)=N\alpha$ then
		$\norm{\cD_{\eta,\alpha}f_{\sing,k}}_{L^2}\underset{\eta\to 0}{\to}1$, that is 
		\[
		\lim_{\eta\to 0}\frac{\norm{\cD_{\eta,\alpha}f_{\sing,k}}_{L^2}}{\norm{f_{\sing,k}}_{L^2}}=+\infty.
		\]
		\item If $k=E(N\alpha)$, $e(N\alpha)>0$, then $f_{\sing,k}$ converges in $L^2(\C)^2$ to 
		\[
			C_{e(N\alpha)}^{-1/2}e^{-iE(N\alpha)\theta}K_{e(N\alpha)}(r)\begin{pmatrix} 0 & 1\end{pmatrix}^T.
		\]
		\item If $k<E(N\alpha)$, then $f_{\sing,k}\rightharpoonup 0$ in $L^2$, collapses onto $0$ and satisfies:
		\[
			\lim_{\eta\to 0}\norm{\cD_{\eta,\alpha} f_{\sing,k}}_{L^2}=0.
		\]
	\end{enumerate}
\end{lemma}
\begin{proof}
In all cases, we check the convergence along any decreasing sequence $\eta_n\to 0$.
Remember that Lichnerowicz formula applies for elements $\phi$ in $\dom(\cD_{\eta,\alpha}^{(\min)})$ \cite{dirac_s3_paper1}*{Proposition~5}:
\begin{equation}\label{eq:Lichnero}
\norm{\cD_{\eta,\alpha}e^{-i\alpha\Theta_{\eta}}(e^{i\alpha\Theta_{\eta}}\phi)}_{L^2}=\norm{\nabla^{\R_2}(e^{i\alpha\Theta_{\eta}}\phi)\restriction\{\Theta_{\eta}\neq 0\}}_{L^2}.
\end{equation}

We deal with the two first cases to begin with.
Up to the extraction of a subsequence, we assume that
$(f_{\sing,k}^{(n)},\cD_{\eta_n,\alpha}f_{\sing,k}^{(n)})\rightharpoonup (f_s,b_s)$ in $[L^2(\C)^2]^2$, and $f_{\sing,k}^{(n)}$ strongly in $L^2_{\loc}(\C\setminus \{0\})^2$. Let $\phi\in C^1_c(\C\setminus\{0\})^2$, for $n$ big enough, we have:
\[
\cip{\cD_{\eta_n,\alpha}f_{\sing,k}^{(n)}}{\phi}_{L^2}=\cip{f_{\sing,k}^{(n)}}{\cD_{\eta_n,\alpha}\phi}_{L^2}
		\underset{n\to+\infty}{\to} \cip{f_s}{\cD_{0,e(N\alpha)}^{(E(N\alpha))}\phi}_{L^2}.
\] 
Thus either $f_s=0$, or $f_s$ is in the maximal domain of $\sigma\cdot(-i\nabla^{\R^2}+N\alpha\nabla\theta)$. The latter domain is known explicitely (see \cite{Persson_dirac_2d} for instance). Since $f_{\sing,k}^{(n)}$ is known to be spin down, we obtain $f_s=0$ or 
$f_s=e^{-iE(N\alpha)\theta} K_{e(N\alpha)}(r)\begin{pmatrix}  0& 1\end{pmatrix}^T$, in which case:
\[
\sigma\cdot(-i\nabla^{\R^2}+N\alpha\nabla\theta)f_s=e^{-i[E(N\alpha)+1]\theta} K_{1-e(N\alpha)}(r)\begin{pmatrix} 1& 0\end{pmatrix}^T.
\]
It is obtained for $k=E(N\alpha)$, $e(N\alpha)>0$. By orthogonality if $k>E(N\alpha)$, then $f_s=0$ and $\lim_{n}\norm{\cD_{\eta_n,\alpha}f_{\sing,k}^{(n)}}_{L^2}=1$. Similarly, if $k=E(N\alpha)=N\alpha$, then the limit operator is a gauge transform of the free Dirac operator, and by orthogonality with the minimal domain, we get the same results as for $k>E(N\alpha)$.
Let us check norm convergence in case (2). 
Consider the smooth partition of unity $\chi+(1-\chi)$, and the power series representation of $K_{e(N\alpha)}(r)$.
We define $K_{e(N\alpha),\eta}$ as the function given by the same power series except that the term 
$\tfrac{\Gamma(e(N\alpha))}{2^{1-e(N\alpha)}}r^{-e(N\alpha)}$ is replaced by:
\[
\frac{\Gamma(e(N\alpha))}{2^{1-e(N\alpha)}}\Big[\frac{\overline{z}^{E(N\alpha)}}{(\overline{z}^N-\eta^N)^\alpha}\chi+(1-\chi)r^{-e(N\alpha)}\Big],
\]
and form $f_{\eta,E(N\alpha)}:=e^{-iE(N\alpha)\theta} K_{e(N\alpha),\eta}\begin{pmatrix}  0& 1\end{pmatrix}^T\in \dom(\cD_{\eta,\alpha})$. 

Setting $C=\tfrac{\Gamma(e(N\alpha))}{2^{1-e(N\alpha)}}$, we have $f_{\eta,E(N\alpha)}-Cg_{\eta,E(N\alpha)}\in\dom(\cD_{\eta,\alpha}^{(\min)})$, thus $f_{\eta,E(N\alpha)}$ and $g_{\eta,E(N\alpha)}$ have the same singular part.
Furthermore $f_{\eta,E(N\alpha)}$ converges strongly, and up to the extraction of a subsequence, the regular part of $f_{\eta_n,E(N\alpha)}$ also converges strongly. This proves the strong convergence of $f_{\sing,E(N\alpha)}$ in $L^2$.

\medskip

We now turn to the last case.
Observe that $\norm{g_{\eta,k}}_{L^2}$ tends to $+\infty$ and $\norm{\cD_{\eta,\alpha}g_{\eta,k}}_{L^2}$ remains finite
as $\eta\to 0$. Consider $G_{\eta,k}:=\tfrac{g_{\eta,k}}{\norm{g_{\eta,k}}_{\cD_{\eta,\alpha}}}$: it collapses to $0$, and its energy tends to $0$.
We write $h_{\eta_n,k}$ resp. $\wt{f}_{\eta_n,k}$ its regular resp. singular part.
Let us study the sequence $h_{\eta_n,k}:=P_{\min}G_{\eta_n,k}\in\dom(\cD_{\eta,\alpha}^{(\min)})$. 
It is graph-norm bounded (by Pythagorean theorem) and Lichnerowicz formula~\eqref{eq:Lichnero} applies.
So up to extracting a subsequence it converges strongly in $L^2(\C)^2$. But all the mass of $G_{\eta_n,k}$ concentrates on $0$:
\[
1=\lim_{\eps\to 0}\liminf_{n\to+\infty}\int_{D_\eps[0]}\big(|G_{\eta_n,k}|^2+|\cD_{\eta_n,\alpha}G_{\eta_n,k}|^2\big)=\lim_{\eps\to 0}\liminf_{n\to+\infty}\int_{D_\eps[0]}|G_{\eta_n,k}|^2.
\]
Once again, Pythagorean theorem gives
\begin{equation*}
1=\lim_n\norm{G_{\eta_n,k}}_{L^2}^2=\norm{h_{\eta_n,k}}_{L^2}^2+\norm{\cD_{\eta_n,\alpha}h_{\eta_n,k}}_{L^2}^2+\norm{\wt{f}_{\eta_n,k}}_{L^2}^2+\norm{\cD_{\eta_n,\alpha}\wt{f}_{\eta_n,k}}_{L^2}^2,
\end{equation*}
Since $G_{\eta_n,k}=h_{\eta_n,k}+\wt{f}_{\eta_n,k}$ collapses to $0$, and that $h_{\eta_n,k}$ converges in $L^2$, necessarily its limit is $0$.
Taking the limsup in the expression above, we obtain the limit $\lim_n(\norm{\cD_{\eta_n,\alpha}h_{\eta_n,k}}_{L^2}+\norm{\cD_{\eta_n,\alpha}\wt{f}_{\eta_n,k}}_{L^2})=0$.
\end{proof}

Let us anticipate and introduce the $N$-folding of functions in $L^2(\T_{N\ell})$, which will enable us to link the above $2$D-operators
with $\cD_{\bA_{\full}}$.

\subsubsection{The $N$-folding of $L^2(\T_{N\ell})$}\label{sec:N_fold}

Recall $\zeta:=e^{2i\pi\tfrac{M}{N}}$. We decompose $L^2(\T_{N\ell})$ according to the cable representation
$\gamma_{\eta}$ the following way. Using \eqref{eq:hey_root_guy}, one realizes that any $f\in L^2(\T_{N\ell})$ can be decomposed as follows:
\[
f(s)=\sum_{k=0}^{N-1}\ \frac{1}{N}\sum_{j=0}^{N-1}\zeta^k f(s+j\ell)=:\sum_{k=0}^{N-1} f_k(s),
\]
where the $f_k$'s are orthogonal to each other. In other words, writing
\[
L_{\zeta^k}^2(\T_{N\ell}):=\big\{f\in L^2(\T_{N\ell}),\ f(s+\ell)=\zeta^k f(s)\big\},
\]
we have the $L^2$-orthogonal decomposition:
\[
	L^2(\T_{N\ell})=\underset{0\le k\le N-1}{\overset{\perp}{\bigoplus}}L_{\zeta^k}^2(\T_{N\ell}),
\]
and $H^1_{\zeta^k}(\T_{N\ell})$ denotes $H^1(\T_{N\ell})\cap L_{\zeta^k}^2(\T_{N\ell})$. 
The orthogonality is obvious and follows from \eqref{eq:hey_root_guy}.

We call this procedure the $N$-folding of $L^2(\T_{N\ell})$ and $H^1(\T_{N\ell})$ respectively. 
The key idea is that the $N$-folding of $H^1(\gamma_{\eta}\simeq_{\mathrm{top}}\T_{N\ell})$ is naturally realized with the help of complex polynomial
in the parameter $\overline{z}=\rho e^{-i\theta}$ of the cross-section. As for $\cD_{\eta,\alpha}$, this enables us to describe
the domain of $\sigma(-i\nabla +2\pi\alpha \bA_{\gamma_{\eta}})$ in terms of functions in $H^1(\T_{\ell})$ and not in $H^1(\T_{N\ell})$,
and ultimately to study the continuity issue (compare \eqref{eq:ansatz_tubulaire_u}-\eqref{eq:ess_dom_tub_u} and \eqref{eq:descr_dom_D_full}-\eqref{eq:ansatz}).

\subsection{Description of the domain of $\cD_{\bA_{\full}}$}\label{sec:desc_dom}
We recall that we deal with Dirac operators in the $\S^3$-metric. 
The operators act on $L^2_\Omega:=L^2(\R^3,\Omega^3\d \bx)^2$ where $\Omega(\bx)=\tfrac{2}{1+|\bx|^2}$ is the conformal factor.
Also $H^1_{\Omega}$ denotes the Hilbert space $H^1(\R^3,\Omega^3\d\bx)^2\widehat{=} H^1(\S^3)^2$, 
defined w.r.t. the connection \cite{ErdSol01}*{Section~4}
\begin{equation}\label{eq:connexion_Omega}
\nabla_X^{(\Omega)}:=\nabla^{\R^3}_X+\frac{1}{4}\Omega^{-1}[\sigma\cdot X,\sigma\cdot \nabla^{\R^3}\Omega].
\end{equation}
For a given magnetic potential $\bA$, the Dirac operator $\cD_{\bA}=\cD_{\bA}^{(\Omega)}$ in the $\S^3$-metric is obtained from that in the flat metric
by \cite{ErdSol01}*{Section~4}
\[
	\cD_{\bA}^{\Omega}= \Omega^{-2}\cD_{\bA}^{\R^3}\Omega.
\]

We only deal with the case $MN$ even, the case $MN$ odd is similar to derive and left to the reader.
Here $\bA_{\full}=\bA_{\full}(\eta,t)$ is the sum of the Biot and Savart gauge $2\pi t\bA_{\gamma_{\eta}}$ and an auxiliary singular gauge 
$\bA_{\mathrm{aux}}$ equal to $2\pi\tfrac{\alpha_a}{K}\sum_{k=1}^K[D_-(s_k)]$. 
We write $\gamma_{\aux}$ the link $\cup_{1\le j\le K}\partial D_-(s_k)$, $S_{\aux}$ its gauge $\cup_{1\le k\le K}D_-(s_k)$. 
The link $\gamma_{\aux}\cup \gamma_{\eta}$ is denoted by $\gamma_{\full}$ and $\gamma_{\aux}\cup\gamma_0$ by $\gamma_{\lim}$.

In the next section we define the Dirac operator. Then in Sections \ref{sec:loc} and \ref{sec:remov_phase}, 
we recall two tools developped in \cite{dirac_s3_paper1} for its analysis, namely the localization functions and the phase jump functions.
We give in Section \ref{sec:ess_dom} an essential domain for $\cD_{\bA_{\full}}$ adapted to the cable representation.

\subsubsection{Definition of $\cD_{\bA_{\full}}$}\label{sec:def_dirac}

Let us recall the definition of the associated Dirac operator. In a tubular neighborhood $B_{\eps_1}[\gamma_{\full}]$ 
of $\gamma_{\eta}$ and $\gamma_{\aux}$, we can extend the tangent vector $\bT$ by parallel transportation along $\gamma_{\full}$-orthogonal geodesics. 
We can indifferently choose the flat metric of $\R^3$ or the metric of $\S^3$ as we will see below.
Of course we choose the former as we did so throughout the paper.

Let us consider a Seifert surface $S_{\eta}$ for $\gamma_{\eta}$. The tubular coordinates of $\gamma_{\full}$, 
defined on $B_{\eps_1}[\gamma_{\full}]$ are written $(u,r,\vartheta)$, while those of $\gamma_{\lim}$
are written $(s,\rho,\theta)$ and defined on $B_\eps[\gamma_{\lim}]$ with $\eps>0$ fixed.

Let $\Psi$ be the (trivial) spinor bundle of $\S^3$, with chart $\R^3\times\C^2$ (without the north pole). Locally around the link we can decompose $\Psi$ into two line bundles
corresponding to the eigen-spaces of $\sigma\cdot \bT$. We consider $\xi_+$ and $\xi_-$ two unit sections of these line bundles:
\[
	\sigma\cdot \xi_{\pm}=\pm\xi_{\pm},
\]
and call $P_+$ resp. $P_-$ the pointwise projections onto them. Thus $P_+$ and $P_-$ are local sections of endomorphisms on $\Psi$. We 
fix the relative phase of the $\xi_{\pm}$'s according to the Seifert surfaces $D_-(s_k)$ and $S$ (for $\gamma_0$). That is, writing $(\bT,\bS,\bN)$
the corresponding Seifert frame, they are chosen in such a way that for every $1$-form $\omega\in\Omega^1(B_{\eps}[\gamma_{\lim}])$ there holds pointwise
\begin{equation}\label{eq:relative_phase}
 \omega(\bS)+i\omega(\bN)=\cip{\xi_-}{\sigma(\omega)\xi_+}_{\C^2}.
\end{equation}

We are ready to define the domain. The singular part $2\pi\tfrac{\alpha_a}{K}\sum_{k=1}^K[D_-(s_k)]$ imposes a phase jump $\mathrm{exp}(-2i\pi\alpha_a/K)$ 
across each surface\footnote{due to the orientation of these disks, from $s=s_k^-$ to $s=s_k^+$, there is a phase jump of $\mathrm{exp}(2i\pi\alpha_a/K)$.} 
$D_-(s_k)$, and, writing $S_{\aux}:=\cup_{k=1}^K D_-(s_k)$, we consider
\begin{multline*}
	\cA_0:=\Big\{\psi\in H_{\Omega}^1,\ (\supp\,\psi)\cap \gamma_{\full}=\emptyset\\ 
			\&\quad \psi\restriction{(D_-(s_k))_+}=\mathrm{exp}(-2i\pi\alpha_a/K)\psi\restriction{(D_-(s_k))_-}\in L^2_{\loc}(S_{\aux}\setminus \gamma_{\aux})^2\Big\}.
\end{multline*}
The operator $\cD_{\bA_{\full}}$ acts on $\cA_0$ as follows:
\[
	\cD_{\bA_{\full}}\psi:= \big(\Omega^{-2}\sigma\cdot(-i\nabla)(\Omega\psi)\big)\restriction{\R^3\setminus S_{\aux}}\in L_\Omega^2.
\]
The minimal operator $\cD_{\bA_{\full}}^{(\min)}$ is the closure of $[\cD_{\bA_{\full}}]\restriction{\cA_0}$ in $L_\Omega^2$, and defines 
a symmetric operator on $\S^3$. (The north pole has co-dimension $3$ in $\S^3$).

Let $\wt{\chi}\in \sD(\R^3,[0,1])$ be a cut-off function whose support is included in $B_{\eps_1}[\gamma_{\full}]$, 
which only depends on the distance $r$ to the link $\gamma_{\full}$ and which equals $1$ in a small neighborhood around it. 
We define the self-adjoint extension $\cD_{\bA_{\full}}=\cD_{\bA_{\full}}^{(-)}$
of $\cD_{\bA_{\full}}^{(\min)}$ by (see \cite{dirac_s3_paper1}):
\[
	\dom(\cD_{\bA_{\full}})=\big\{\psi\in \dom\,(\cD_{\bA_{\full}}^{(\min)})^*\subset L_\Omega^2,\ P_+(\wt{\chi}\psi)\in \dom(\cD_{\bA_{\full}}^{(\min)})\big\}.
\]
Observe that for $\psi\in \dom\,(\cD_{\bA_{\full}}^{(\min)})^*$, we have $\dist(\cdot,\gamma_{\full})\wt{\chi} \psi\in \dom(\cD_{\bA_{\full}}^{(\min)})$
\cite{dirac_s3_paper1}*{Lemma~12}. That is why this definition does not depend
on $\wt{\chi}$ and that we could extend $\bT$ in the flat metric or in the metric of $\S^3$.

\subsubsection{Localization}\label{sec:loc}
Now we describe the useful partition of unity \eqref{def:part_unity}.
Let $\eps>0$ such that for all $0\le k \neq k'\le K$ the intersection $B_{\eps}[\gamma_k] \cap B_{\eps}[\gamma_{k'}]$ is empty.
This implies that the tubular coordinates $(s,\rho,\theta)$ (with respect to the introduced Seifert frames \eqref{eq:tub_coord}) are well-defined on $B_{\eps}[\gamma_k]$.
Note that by the continuity of the coordinates, the condition is still satisfied for the link $\gamma_{\full}$ for $\eta>0$ small enough.

Let $\chi: \R \to \R_+$ be a smooth function with $\supp \chi \in [-1,1]$ and $\chi(x) = 1$ for $x \in [-2^{-1},2^{-1}]$. 
For $0\le k\le K$, we define the localization function at level $0<\delta <\eps$ by
\begin{equation*}
	\chi_{\delta,\gamma}: 
	\begin{array}{ccl}
		B_{\delta}[\gamma_k]\subset\R^3 &\longrightarrow& \R_+\\
		\bx = \gamma_k(s)+\rho_{\gamma_k} \bv_0(s)
		&\mapsto& \chi\left(\rho_{\gamma}\delta^{-1}\right),
	\end{array}
\end{equation*}
where $\bv_0(s)\in\rT_{\gamma_k}\gamma_k(s)$. We then pick $\delta>0$ such 
that
$$
	0<\delta< \eps \cdot \min\left\{1, \left(\norm{\kappa}_{L^{\infty}} + \sqrt{\eps + \norm{\kappa}_{L^{\infty}}^2}\right)^{-1}\right\},
$$
Here, writing $(s,\rho,\theta)$ the tubular coordinates of the link $\gamma_{\lim}$, $\norm{\kappa}_{L^{\infty}}$ denotes:
\[
	\sup_{0\le k\le K}\sup_{\theta,s}\big| \cos(\theta)\kappa_{g,\gamma_k}(s)+\sin(\theta)\kappa_{n,\gamma_k}(s)\big|.
\]

Furthermore, we set
$$
	\chi_{\delta,S_{k}}(\bx) 
	:= \chi\left(4\tfrac{\dist_{\R^3}(\bx,S_{k})}{\delta}\right)\left(1-\chi_{\delta,\gamma_{k}}(\bx)\right),
$$
which has support close to the Seifert surface $S_{k}$, but away from the knot $\gamma_{k}$.
(Here, by convention $\chi_{\delta,S_0}=0$ as we have picked the B.S. gauge for the knot $\gamma_0$).
The remainder is then defined as
$$
	\chi_{\delta,R_k}(\bx) := 1-\chi_{\delta,S_k}(\bp) - \chi_{\delta,\gamma_k}(\bp).
$$
We get a partition of unity subordinate to $\gamma_k,S_k$.
The partition of unity for the entire link $\gamma_{\lim}$ is then given by
\begin{align}\label{def:part_unity}
	1 &= \prod_{k=0}^K\left(\chi_{\delta,\gamma_k}(\bx) + \chi_{\delta,S_k}(\bx)
	+\chi_{\delta,R_k}(\bp) \right)\nn\\
	&= \sum_{\underline{a} \in \{1,2,3\}^{K+1}}\chi_{\delta,\underline{a}}(\bx)
	=\sum_{k=0}^K\chi_{\delta,\gamma_k}(\bx)+\sum_{\underline{a} \in \{2,3\}^{K+1}}\chi_{\delta,\underline{a}}(\bx),
\end{align}
where $\chi_{\delta,\underline{a}}$ is the product $\prod_{k=0}^J\chi_{\delta,X_k}$, with 
$$
	X_k= \left\{\begin{array}{lr}
	\gamma_k, &a_k = 1,\\
	S_k, &a_k = 2,\\
	R_k, &a_k = 3.
	\end{array}\right.
$$

\subsubsection{Phase jump functions}\label{sec:remov_phase}
\paragraph{\emph{Around $\gamma_0$}}
Let $1\le k\le K$: $\gamma_0$ intersects $S_{k}$ non trivially.
Up to taking $\eps>0$ small enough we can assume that 
the two subsets $B_{\eps}[\gamma_0]\cap S_{k}$ and $\gamma_0\cap S_{k}$ have the same
number of connected components.

We define a  map $E_{0,k}$ describing the phase jumps on $B_{\eps}[\gamma_0]$ due to $S_{k}$.
\medskip

The curve $\gamma_0$ intersects $S_{k}$ at the point
$0\le s_1<\ell_{0}$. Call $C_1$ the corresponding connected
components of the intersection $B_{\eps}[\gamma_0]\cap S_{k}$.
The surface $S_{k}$ induces a phase jump $e^{\pm 2i\pi\alpha_k}$ across the cut $C_1$.

		The cut neighborhood
		$B_\eps[\gamma_0]\setminus S_{k}=:B(\eps,0,k)$ is contractible. 
		Thus we can lift the coordinate map $s_{\gamma_0}(\cdot)$ on this subset which gives a smooth function
		\[
		    s_{0,k}:B_\eps[\gamma_0]\setminus S_{k}\mapsto \mathbb{R},
		\]
		satisfying for all $\bx\in B_\eps[\gamma_0]\setminus S_{k}$:
		\[
		    \mathrm{exp}\Big(\frac{2i\pi}{\ell_0}s_{0,k}(\bx)\Big)=\mathrm{exp}\Big(\frac{2i\pi}{\ell_0}s_{\gamma_0}(\bx)\Big).
		\]
		We then define for all $\bx\in B_\eps[\gamma_0]\setminus S_{k}\supset B_\eps[\gamma_0]\cap\Omega_{\uS}$
		\begin{equation*}
			E_{0,k}(\bx):=\exp\Big(\mp 2i \pi\alpha_k\frac{s_{0,k}(\bx)}{\ell_0}\Big).
		\end{equation*}

Accompanying these decompositions, for $\bp \in B_{\eps}[\gamma_k] \cap \Omega_{\uS}$ we define
the function
\begin{equation*}
	E_{\delta,0}(\bp) :=
	\prod_{1\le k\le K}E_{0,k}(\bp),
\end{equation*}
and introduce the slope
\begin{equation*}
	c_a:=\sum_{1\le  k\le K}\frac{2\pi\alpha_{a}}{K}\link(\gamma_{0},\gamma_{k})=-2\pi\alpha_a.
\end{equation*}
Note that the function $E_{\delta,0}$ has a bounded derivative in $B_{\eps}[\gamma_0] \cap \Omega_{\uS}$, and that for any vector field $X$,
$\overline{E_{\delta,0}}X(E_{\delta,0})=ic_a X(s_{\gamma_0}(\cdot))$ can be extended to an element in $C^1(B_{\eps}[\gamma_0])$.

The $E_{0,k}$'s are $\S^1$-valued functions, locally depending only on $s_{\gamma_0}$,
with a fixed slope and the correct phase jump across the cut $C_1$: 
they are a unique up to a constant phase. Note also the important remark.
\begin{rem}
Due to our choice of the auxiliary potentials, the phase jump function $E_{\delta,0}$ is also adapted to the cable knot $\gamma_{\eta}$.
Henceforth the following proposition holds.
\end{rem}

\begin{proposition}\label{prop:sa_link}
	Let $(\,\cdot\,) = \{(\max),(-),(\min)\}$, then for $\eta$ small enough
	the map $\psi \to \overline{E_{\delta,0}}\psi$ maps
	the set $\{\psi \in \dom(\cD_{\bA_{\full}}^{(\,\cdot\,)}): \supp\psi \in B_{\eps}[\gamma_0] \cap \Omega_{\uS}\}$
	onto the set $\{\psi \in \dom(\cD_{\bA_{\gamma_{\eta}}}^{(\,\cdot\,)}): \supp\psi \in B_{\eps}[\gamma_0] \cap \Omega_{\uS}\}$.
	And for $\psi \in \dom(\cD_{\bA_{\full}}^{(\,\cdot\,)})$ localized around $\gamma_0$, we have:
	\[
	 \cD_{\bA_{\full}}^{(\,\cdot\,)}\psi=E_{\delta,0}\cD_{2\pi\alpha \bA_{\gamma_{\eta}}}^{(\,\cdot\,)}(\overline{E}_{\delta,0}\psi)+\frac{c_a}{h}\sigma(\bT^{\flat})\psi.
	\]
\end{proposition}

Observe that $E_{\delta,0}$ is a function of $s$ only (the arclength parameter of $\gamma_0$),
and $\omega_0:=-i\overline{E}_{\delta,0}\d E_{\delta,0}$ can be continuously extended to the $1$-form $-2\pi\alpha_a\frac{\d s}{\ell}$ on $B_\eps[\gamma_0]$.
In particular we have
\begin{equation*}\label{eq:integ_phase_jump}
	\int_{\gamma_0}\omega_0=-2\pi\alpha_a.
\end{equation*}

\paragraph{\emph{Around the auxiliary surface $S_{\aux}$}}
For $1\le k\le K$, consider the localization function $\chi_{\delta,S_k}$ of the previous section, whose support
is in $B_{4\delta}[S_k]\setminus B_{\delta}[\gamma_k]$. So, the surface $S_k$ cuts this support into two regions
$O_{k,\pm}$ above and below $S_k$ respectively. Remember that the elements in $\dom(\D_{\bA_{\full}})$ exhibits
a phase jump $e^{-2i\pi\tfrac{\alpha_a}{K}}$ across the surface $S_k$.
Consider now the localization function $\chi_{\delta,\underline{a}}$ with $\underline{a}\in\{2,3\}^{K+1}$.
The following phase jump function contains all the phase jumps due to the surfaces $S_k$'s in its support
\begin{equation}\label{eq:phase_jump_fun_surface}
E_{\delta,\underline{a}}:=\prod_{k:a_k=2}e^{2i\pi\tfrac{\alpha_a}{K}\mathds{1}_{O_k,-}}.
\end{equation}

\subsubsection{Essential domain for $\cD_{\bA_{\full}}$}\label{sec:ess_dom}
We now give an essential domain for $\cD_{\bA_{\full}}$. Consider a cut-off function $\wt{\chi}=\wt{\chi}_{\eta,0}+\sum_{k=1}^K\wt{\chi}_k$ which localizes around the link
$\gamma_{\full}$ and which is a function of $r$, the distance to the link only. Using Lemma~\ref{lem:refinement_BS_behav}, we introduce the sets
\begin{equation}\label{eq:ansatz_tubulaire_u}
	\begin{array}{l}
	\cA_-(\gamma_{\eta}):=\spann \bigg\{ \Omega^{-1}e^{-i\alpha(\Phi_{\eta}+g_{\eta})}\wt{\chi}_{\eta,0}(r)E_{\delta,0}(s) \dfrac{f(u)}{r^{\alpha}}\xi_-(u),f\in H^1(\T_{N\ell})\bigg\},\\
	\wt{\cA}_-(\gamma_k):=\spann \bigg\{ \Omega^{-1}\wt{\chi}_k(r)\dfrac{h(u)}{r^{\alpha_a/K}e^{i\alpha_a/K\vartheta}}\xi_-(u),\ h\in H^1(\T_{|\gamma_k|})\bigg\},
	\end{array}
\end{equation}
where $e^{i\alpha_a\vartheta}$ is cut along $S_{\aux}$. Thanks to \cite{dirac_s3_paper1}*{Lem.~10}
and the proof of \cite{dirac_s3_paper1}*{Theorem~7}, we know that the following set is an essential domain for $\cD_{\bA_{\full}}$: 
\begin{equation}\label{eq:ess_dom_tub_u}
\cA_-(\gamma_{\eta})+\dom(\cD_{\bA_{\full}}^{(\min)})+\sum_{k=1}^K\wt{\cA}_-(\gamma_k).
\end{equation}

Using Sections~\ref{sec:N_fold} and \ref{sec:complex_poly}, we can define another essential domain. We proceed as follows. 
Consider the local sections $\xi_{\pm}$ satisfying \eqref{eq:relative_phase}.
Consider the localization function $\sum_{k=0}^K\chi_{\delta,\gamma_k}\in \sD(\R^3,[0,1])$ localizing around $\gamma_{\lim}$. 
Let
\begin{equation}\label{eq:descr_dom_D_full}
	\begin{array}{rcl}
	\cA_-(\gamma_0)&:=&\spann \bigg\{ \Omega^{-1}e^{-i\alpha(\Phi_{\eta}+g_{\eta})}\chi_{\delta,\gamma_0}(\rho)E_{\delta,0}(s) \dfrac{f_{\gamma_0}(s)(s)\overline{z}^k}{q_\eta(\bx)^{\alpha}}\xi_-(s),\\
		&&\hfill 0\le k\le N-1,\ f_{\gamma_0}\in H^1(\T_{\ell})\bigg\},\\
	\cA_-(\gamma_k)&:=&\spann \bigg\{ \Omega^{-1}\chi_{\delta,\gamma_k}(\rho)\dfrac{f_{\gamma_k}(s)}{\overline{z}^{\alpha_a/K}}\xi_-(s),\ f_{\gamma_k}\in H^1(\T_{|\gamma_k|})\bigg\},
	\end{array}
\end{equation}
We recall that $z=\rho e^{i\theta}$ and that $q_\eta(\bx)=|(ze^{iI_0(s)})^N-\eta^N|$ ($I_0(s)$ denotes $\int_0^s\tau(s')\d s'$). 
Consider the span $\cA_-:=\cA_-(\gamma_0)+\sum_{j=1}^J\cA_-(\gamma_k)$.
Then the following set is an essential domain of $\cD_{\bA_{\full}}$
 \begin{equation}\label{eq:ansatz}
  \cA:=\cA_0+\cA_-:=\cA_0+\cA_-(\gamma_0)+\sum_{k=1}^K\cA_-(\gamma_k).
 \end{equation}

\subsection{Model operator}\label{sec:model_op}
As in \cite{dirac_s3_paper1}, we infer a model operator for the behavior of the operator $\cD_{\bA_{\full}}$ 
on functions $\Omega^{-1}\psi$ supported in the vicinity of $\gamma_0$. We aim to derive the very useful Lemma~\ref{lem:graph_norm_estimate} 
which will enable us to split such a localized element into a regular part (in the minimal domain) and a singular part.
We recall that $\tau$ denotes the relative torsion associated with the Seifert surface $S$ of $\gamma_0$ and that $I_0(s)=\int_0^s \tau(s')\d s'$.

\begin{rem}\label{rem:wlog}
Up to making the local gauge transformation $e^{i\alpha(\Phi_{\eta}+g_{\eta})}\overline{E}_{\delta,0}(s)$, we can assume w.l.o.g. that 
$\bA_{\full}=\alpha \d \theta_{\eta}-2\pi\alpha_a \tfrac{\d s}{\ell}$ when we deal with an element $\psi$ which is localized in the vicinity of $\gamma_0$.
\end{rem}

As before we consider the local sections $\xi_+(s),\xi_-(s)$ corresponding to the Seifert frame of $\gamma_0\subset S$
(satisfying \eqref{eq:relative_phase}).
We will use the notations:
\[
 \zeta:=e^{2i\pi\tfrac{M}{N}},\ \zeta^{1/2}:=e^{i\pi\tfrac{M}{N}}\ \&\ \zeta^{-1/2}=\overline{\zeta}^{1/2}:=e^{-i\pi\tfrac{M}{N}}.
\]

\medskip

\paragraph{\textit{Cable coordinates}}
We now use the \emph{cable coordinates}, that is, the tubular coordinates $(s,\rho,\widetilde{\theta})$ relative to 
the frame $(\bT,U(s),V(s))$ (defined in \eqref{eq:def_U_V})
\begin{equation}\label{eq:new_coord}
  \begin{array}{rcl}
  [0,\ell)\times [0,\eps)\times \R/(2\pi\Z)&\longrightarrow & B_{\eps}[\gamma_0],\\
 (s,\rho,\wt{\theta})&\mapsto &\gamma_0(s)+\rho(\cos(\wt{\theta})U(s)+\sin(\wt{\theta})V(s)),
 \end{array}
\end{equation}
with $\theta=\wt{\theta}-I_0(s)$. These coordinates are orthogonal, but they are not defined up to $s\mapsto \ell^-$ 
since $U(\ell)$ has rotated by $2\pi\Wr(\gamma_0)=2\pi\frac{M}{N}$ after one turn.
We can enforce continuity in the chart by setting (with $w=\rho e^{i\wt{\theta}}$)
\begin{equation}\label{eq:def_base_space}
 \widetilde{B}_{\model}:= [0,\ell]\times \C/\Big((\ell,w)\sim (0, \zeta w)\Big).
\end{equation}
Similarly we change the basis of sections $(\xi_+,\xi_-)$ and pick the one relative to $(\bT,U(s),V(s))$:
\[
 \xi_u(s):= e^{iI_0(s)/2}\xi_+(s)\quad\&\quad \xi_d(s):= e^{-iI_0(s)/2}\xi_-(s),\ 0\le s<\ell. 
\]
This basis is not continuous across $s=0$ for the same reason as for $U,V$. For a section $\psi$ of
$B_\eps[\gamma_0]\times \C^2$, we write $f_e=\cip{\xi_e}{\psi}$, $e\in\{u,d\}$ its components in the new basis.
The element $f:=\begin{pmatrix}f_u& f_d \end{pmatrix}^T$ satisfies:
\[
 f_u(\ell,w)=\zeta^{1/2}\cip{\xi_+(0,w)}{\psi(0,\zeta w)}_{\C^2},\   f_d(\ell,w)=\zeta^{-1/2}\cip{\xi_-(0,w)}{\psi(0,\zeta w)}_{\C^2},
\]
and can be seen as a section of
\begin{equation}\label{eq:def_bundle_twist}
 \widetilde{\Psi}_{\model}:= [0,\ell]\times \C\times \C^2/\Big[
 \big(\ell,w,\begin{pmatrix}f_u& f_d \end{pmatrix}^T\big)\sim \big(0,\zeta w,\begin{pmatrix} \zeta^{1/2}f_u& \zeta^{-1/2}f_d \end{pmatrix}^T\big)\Big].
\end{equation}

This is a complicated bundle. The description is more clear when we untwist the base space and work with the $2N$-folded covering space of $B_\eps[\gamma_0]\times \R^2$,
where the cable coordinates \emph{and} the sections $(\xi_u,\xi_d)$ are globally defined (the factor $2$ in $2N$ is due to the spinorial nature of the problem). 

\medskip

\paragraph{\textit{Description of the model case}}

Consider the auxiliary base space $B_{\aux}:=\T_{2N\ell}\times \R^2$ and spin$^c$ bundle $\Psi_{\aux}:=B_{\aux}\times\C^2$.
The base space is endowed with the flat metric and coordinates map $(s,r e^{i\phi})=(s,w_1+iw_2)=(s,w)$, and the spin$^c$ bundle 
is endowed with the Clifford map:
\[
 \sigma_{\aux}(\d s)=\sigma_3,\ \sigma_{\aux}(\d w_1)=\sigma_1,\ \sigma_{\aux}(\d w_2)=\sigma_2.
\]
We endow the base space with its flat connection and the bundle with the corresponding connection (see \cite{ErdSol01}).
We consider the magnetic field $2\pi\alpha\sum_{k=0}^{N-1}[(s,\eta \zeta^k)_{s\in\T_{N\ell}}]$, whose Coulomb gauge is
$
\alpha\sum_{k=0}^{N-1}\mathrm{d} \phi(w-\eta \zeta^k).
$
For simplicity we consider its singular gauge $\boldsymbol{\alpha}_{\aux}:=\alpha[\mathrm{arg}(P_\eta (w))=0]$ where \
$P_\eta(w)=w^N-\eta^N$. 

We consider the corresponding Dirac operator $\cD_{\aux}:=\cD_{\boldsymbol{\alpha}_{\aux}}^{(-)}$ on $\Psi_{\aux}$.
Having \eqref{eq:def_bundle_twist} in mind, we define the Hilbert subspace of $L^2(B_{\aux})^2$:
\begin{multline}\label{eq:mod_hilb}
 \mathcal{H}_{\zeta}:=\Big\{f=\begin{pmatrix} f_u&f_d\end{pmatrix}^T\in L^2(\T_{2N\ell}\times\C)^2,\ \forall(s,w)\in\T_{N\ell}\times\C,\\
 f(s+\ell,w)= \begin{pmatrix}\zeta^{1/2}& 0\\ 0& \overline{\zeta}^{1/2} \end{pmatrix}f(s,\zeta w)\Big\},
\end{multline}
and realize that ``the" model operator $\cD_{\model}$ is the restriction of $\cD_{\aux}$ to $\dom(\cD_{\model}):= \dom(\cD_{\aux})\cap \mathcal{H}_{\zeta}$.
As $\cD_{\model}$ acts like the free Dirac operator away from the phase jump surface $\{\mathrm{arg}(P_\eta (w))=0\}$, 
the operator satisfies $\cD_{\model}:\dom(\cD_{\model})\to \mathcal{H}_{\zeta}$
since there holds:
\[
\begin{array}{rcl}
 \partial_w \big(\overline{\zeta}^{1/2}f_d(s,\zeta w)\big)&=&\zeta^{1/2}(\partial_w f_d)(s,\zeta w),\\
 \partial_{\overline{w}} \big(\zeta^{1/2}f_u(s,\zeta w)\big)&=&\overline{\zeta}^{1/2}(\partial_{\overline{w}} f_u)(s,\zeta w).
 \end{array}
\]
Thus $\cD_{\model}$ is a self-adjoint operator on $\mathcal{H}_{\zeta}$. We can define $\cD_{\model}^{(\max)},\cD_{\model}^{(\min)}$
in a similar way.
By using Stokes' formula on $\{(s,w),\ q_\eta(w)\ge r,\ \mathrm{arg}(w)\neq 0\}$ and taking the limit $r\to 0^+$, we obtain the following.
\begin{equation}\label{eq:model_case_energy_equalities}
 \begin{array}{lcrl}
  \mathrm{If} \ f\in \dom(\cD_{\model}) & \mathrm{then} & \int|\cD_{\model}f|^2=&\int_{\Theta_\eta(w)\neq 0} \big(|\partial_s f|^2+|\sigma_\perp\cdot\nabla_wf|^2\big),\\
  \mathrm{If} \ f\in \dom(\cD_{\model}^{(\min)}) & \mathrm{then} & \int|\cD_{\model}f|^2=&\int_{\Theta_\eta(w)\neq 0} |\nabla f|^2.
 \end{array}
\end{equation}

\paragraph{\textit{Decomposition of $\dom(\cD_{\model})$}}

As in \cite{dirac_s3_paper3}*{Lemma~32}, we see $\cD_{\aux}$ on $\Psi_{\aux}$ as
$-i\tfrac{\d}{\d s}\otimes \sigma_3+\mathrm{id}\otimes \cD_{\eta,\alpha}$ according to 
the decomposition $L^2(\T_{2N\ell}\otimes\C)^2=L^2(\T_{2N\ell})\otimes L^2(\C)^2$. We recall that 
the operator $\cD_{\eta,\alpha}$ is defined in Section \ref{sub:2dim_op}.
This decomposition is valid thanks to formulas~\eqref{eq:model_case_energy_equalities}, which also enable us to 
decompose $\psi\in \dom(\cD_{\model})$
into Fourier modes with values in $\dom(\cD_{\eta,\alpha})$.

We consider the graph-norm of $\cT_{\eta,\alpha}:=\mathrm{id}\otimes \cD_{\eta,\alpha}$ on $\dom(\cD_{\aux})$, which is intermediate between
the $L^2$-norm and the graph norm $\norm{\cdot}_{\cD_{\aux}}$. We obtain a $\norm{\cdot}_{\cT_{\eta,\alpha}}$-orthogonal
decomposition of $\dom(\cD_{\aux})$ from orthogonal decompositions of $L^2(\T_{2N\ell})$ and of $\dom(\cD_{\eta,\alpha})$ by tensoring.

By an appropriate restriction to $\mathcal{H}_{\zeta}$ we get an adapted decomposition of $\dom(\cD_{\model})$.
Consider the $\norm{\cdot}_{\cT_{\eta,\alpha}}$-projection $P_{\cT}^{(\min)}:=\mathrm{id}\otimes P_{\min}$: we have
$P_{\cT}^{(\min)}\,\dom(\cD_{\aux})=H^1(\T_{2N\ell})\otimes \dom(\cD_{\eta,\alpha}^{(\min)})=\dom(\cD_{\aux}^{(\min)})$.
Its intersection with $\mathcal{H}_{\zeta}$ is the \emph{regular} subspace of $\dom(\cD_{\model})$. 

The \emph{singular} subspace is its $\norm{\cdot}_{\cT_{\eta,\alpha}}$-orthogonal complement.
Let $E_{\sing,k}$ be:
\begin{equation*}
	E_{\sing,k}:=\{h(s)\mathrm{exp}(i(k+\tfrac{1}{2})I_0(s))f_{\sing,k}(w),\ h\in H^1(\T_{\ell})\}
\end{equation*}
We recall that for $0\le s\le 2N\ell$, $I_0(s)$ denotes the integral $\int_0^s\tau(s')\d s'$ where $\tau$
is the relative torsion of $\gamma_0$ w.r.t. the Seifert surface $S\supset \gamma_0$ 
(and up to the congruent map, $\tau$ and $h$ are seen as functions of $\T_{2N\ell}$).

Observe that for any function $h\in L^2(\T_\ell)$ and $k\neq k'$, we have:
\[
\int_{\T_{2N\ell}}e^{i(k'-k)I_0(s)}h(s)\d s=2\sum_{j=0}^{N-1}\zeta^{j(k'-k)}\int_{s=0}^\ell e^{i(k'-k)I_0(s)}h(s)\d s=0.
\]
This implies the following lemma (whose full proof is left to the reader). 
\begin{lemma}\label{lem:decomp_sing_subspace}
The $E_{\sing,k}$'s define an $\norm{\cdot}_{\cT_{\eta,\alpha}}$-orthogonal decomposition of 
the singular subspace. 
For $0\le k\neq k'\le N-1$,  $E_{\sing,k}$ and $E_{\sing,k'}$
are also orthogonal with respect to the graph norm of $\cD_{\model}$, and given $e_k(s)f_{\sing,k}(w)\in E_{\sing,k}$, 
$e_{k'}(s)f_{\sing,k'}(w)\in E_{\sing,k'}$ the functions $e_k(s)$ and $e_{k'}(s)$ are orthogonal with respect to the norms of $L^2(\T_{2N\ell})$
and $H^1(\T_{2N\ell})$.
\end{lemma}

\subsection{Graph norm relation with the model operator}\label{sec:graph_norm_rel}
The purpose of this section is to establish \eqref{eq:graph_norm_formula} and Lemma~\ref{lem:graph_norm_estimate}.
We work in the gauge of Remark~\ref{rem:wlog}, with the cable coordinates $(s,w)=(s,\rho e^{i\widetilde{\theta}})$ \eqref{eq:new_coord} and in the section basis $(\xi_u,\xi_d)$.
We fist write the free Dirac operator in these coordinates.

\medskip

\paragraph{\textit{The free Dirac operator in cable coordinates}}
Let $M_\xi$ be the connexion form of $\sigma\cdot(-i\nabla)$ (in the flat metric of $B_\eps[\gamma_0]\times\C^2$) corresponding to the basis $(\xi_+(s),\xi_-(s))$:
\begin{align*}
 M_\xi&:=
  \begin{pmatrix}
    \cip{\xi_+}{\sigma\cdot(-i\nabla)\xi_+}&\cip{\xi_+}{\sigma\cdot(-i\nabla)\xi_-}\\
    \cip{\xi_-}{\sigma\cdot(-i\nabla)\xi_+}&\cip{\xi_-}{\sigma\cdot(-i\nabla)\xi_-}
  \end{pmatrix},\\
  &=-\frac{i}{h(s,\rho,\theta)}
  \begin{pmatrix}
    \cip{\xi_+}{\partial_s\xi_+}&\cip{\xi_+}{\partial_s\xi_-}\\
    -\cip{\xi_-}{\partial_s\xi_+}&-\cip{\xi_-}{\partial_s\xi_-}.
  \end{pmatrix}
\end{align*}
We have in cable coordinates $(s,w):=(s,\rho\exp(i\widetilde{\theta}))$:
\begin{equation}\label{eq.Dirac_cable_coord}
 \begin{pmatrix}\cip{\xi_u}{\sigma\cdot(-i\nabla)\psi}\\ \cip{\xi_d}{\sigma\cdot(-i\nabla)\psi} \end{pmatrix}=
		-i\begin{pmatrix}\frac{1}{h(s,\rho,\theta)}\partial_s & 2\partial_w \\ 2\partial_{\overline{w}} & -\frac{1}{h(s,\rho,\theta)}\partial_s\end{pmatrix}f
		      \quad +W_\xi(s) f, 
\end{equation}
where $h(s,\rho,\theta)$ denotes
\begin{equation*}
 h(s,\rho,\theta)=h(s,\rho,\wt{\theta}-I_0(s))=1-\rho(\kappa_g(s)\cos(\theta)+\kappa_n(s)\sin(\theta)),
\end{equation*}
and $W_\xi$ the connection form of $\sigma\cdot (-i\nabla)$ for $(\xi_u,\xi_d)$:
\begin{equation}\label{eq.conn_form}
	   W_\xi(s):= (\cos(\tfrac{iI_0(s)}{2})-\sin(\tfrac{iI_0(s)}{2})\sigma_3)M_\xi(s)(\cos(\tfrac{iI_0(s)}{2})+\sin(\tfrac{iI_0(s)}{2})\sigma_3)
	  +\dfrac{\tau(s)}{2}.
\end{equation}

\paragraph{\textit{The operator $\cD_{\bA_{\full}}^{(\max)}$ in cable coordinates}}
Using Lemma~\ref{lem:refinement_BS_behav} and the identity $e^{i\theta_\eta}=e^{-iNI_0(s)+\Theta_\eta(w)}$, 
for $\Omega^{-1}\psi\in\dom(\cD_{\bA_{\full}}^{(\max)})$ there holds:
\begin{multline*}
 \Omega^2\begin{pmatrix}\cip{\xi_u}{\cD_{\bA_{\full}}^{(\max)}\Omega^{-1}\psi}\\ \cip{\xi_d}{\cD_{\bA_{\full}}^{(\max)}\Omega^{-1}\psi} \end{pmatrix}=
		-i\begin{pmatrix}\frac{1}{h(s,\rho,\theta)}\partial_s & 2\partial_w \\ 2\partial_{\overline{w}} & -\frac{1}{h(s,\rho,\theta)}\partial_s\end{pmatrix}f
		\\+\alpha\sum_{k=0}^{N-1}\frac{\sigma_{\perp}\cdot(w-\eta \zeta^k)}{|w-\eta \zeta^k|^2}f
		      -\big(N\alpha\tau+\tfrac{2\pi}{\ell}\alpha_a\big)\frac{\sigma_3}{h(s,\rho,\theta)}f 
		       +W_\xi f,
\end{multline*}
where $f:=\begin{pmatrix}\cip{\xi_u}{\psi}_{\C^2} & \cip{\xi_d}{\psi}_{\C^2}  \end{pmatrix}^T$ and $\sigma_{\perp}\cdot(w-\eta \zeta^k)$ denotes:
\[
 \sigma_{\perp}\cdot(w-\eta \zeta^k):= \mathrm{Re}(w-\eta \zeta^k)\sigma_1+\mathrm{Im}(w-\eta \zeta^k)\sigma_2.
\]
We link to \cite{dirac_s3_paper1} by choosing the singular gauge $2\pi \alpha[\{\theta_\eta=0\}]$ for $2\pi \alpha[\gamma_{\eta}]$.
Writing $\wt{f}:= e^{i\alpha \theta_\eta}f$ and $\cD_{\wt{\bA}_{\full}}^{(\max)}:= e^{i\alpha \theta_\eta}\cD_{\bA_{\full}}^{(\max)}e^{-i\alpha \theta_\eta}$,
we obtain:
\begin{multline}\label{eq:form_Dirac_adapted_basis}
 \Omega^2\begin{pmatrix}\cip{\xi_u}{\cD_{\wt{\bA}_{\full}}^{(\max)}\Omega^{-1}e^{i\alpha\theta_\eta}\psi}\\ \cip{\xi_d}{\cD_{\bA_{\full}}^{(\max)}\Omega^{-1}e^{i\alpha\theta_\eta}\psi} \end{pmatrix}=
		-i\begin{pmatrix}\frac{1}{h(s,\rho,\theta)}\partial_s & 2\partial_w \\ 2\partial_{\overline{w}} & -\frac{1}{h(s,\rho,\theta)}\partial_s\end{pmatrix}\wt{f}\\
		      -\tfrac{2\pi}{\ell}\alpha_a\frac{\sigma_3}{h(s,\rho,\theta)}\wt{f}
		       +W_\xi \wt{f},
\end{multline}

The projection onto $\gamma_{\eta}$ is defined on $B_\eps[\gamma_0]$ minus a rectifiable surface of finite area,
that is almost everywhere. 
The projection is written $\gamma_{\eta}(s_\eta(\bx))$  with $0\le s_\eta(\bx)\le N\ell$ and the distance $\rho_\eta(\bx)$. We have the expansion:
\[
 h(s,\rho,\theta)^{-1}=h(s_\eta,\eta,-I_0(s_\eta))^{-1}+\mathcal{O}(\rho_\eta),
\]
where we recall that $h(s_\eta(\bx),\eta,-I_0(s_\eta(\bx)))=|\gamma_{\eta}(s_\eta(\bx))|$ corresponds to the value of $h$ at $\gamma_{\eta}(s_\eta(\bx))$.

\medskip

\paragraph{\textit{Definition of $B_\eps^{(\eta)}[\gamma_0]$}}
The subset of $B_\eps[\gamma_0]$ for which the projection exists is denoted by $B_\eps^{(\eta)}[\gamma_0]$: it is the open set on which
the tubular coordinates $(u,\rho_\eta)$ w.r.t. $\gamma_{\eta}$ are well-defined.
Observe that we have:  
\begin{equation}\label{eq:partial_u_and_s}
 Y_u(\bx)=\frac{1}{h(s_\eta(\bx),\eta,-I_0(s_\eta(\bx)))}\partial_s,
\end{equation}
where $Y_u$ denotes the (pushforward of) the vector field $\partial_u$.

\medskip

\paragraph{\textit{From $\psi$ to its coordinates in $(\xi_u,\xi_d)$}}
There also hold the two following facts.
\begin{enumerate}
 \item On $B_\eps[\gamma_0]$, we have $C^{-1}\le \Omega \le C$, hence for $\Omega^{-1}\psi\in\dom(\cD_{\bA_{\full}}^{(\max)})$
 with support in $B_\eps[\gamma_0]$, we have:
 \begin{equation*}
  \int \big|(\cD_{\bA_{\full}}\Omega^{-1}\psi)(\bx)\big|^2\Omega(\bx)^3\d\bx<\infty\ \iff\ \int \big|\sigma\cdot(-i\nabla+\bA_{\full})\psi\big|^2<+\infty.
 \end{equation*}
  \item If $\Omega^{-1}\psi\in\dom\,(\cD_{\bA_{\full}}^{(\max)})$, $\supp\psi\subset B_\eps[\gamma_0]$, then there holds
  \[
    \rho_\eta\Omega^{-1}\psi\in\dom(\cD_{\bA_{\full}}^{(\min)}).
  \]
  Indeed, take a smooth partition of unity $1=X(\bx)+(1-X(\bx))$ where $\supp\,X\subset B_{\eta/(4N)}[\gamma_{\eta}]$ and $X=1$ on $B_{\eta/(8N)}[\gamma_{\eta}]$.
  The distance function $\rho_\eta$ is continuous and Lipschitz.
  Hence $(1-X)\psi$ is $H^1(\R^3)^2$ with support in $B_\eps[\gamma_0]\setminus B_{\eta/(8N)}[\gamma_{\eta}]$. We refer to \cite{dirac_s3_paper1}*{Lemma~12} 
  for the fact that $\rho_\eta X\psi$ is in $\dom(\cD_{\bA_{\full}}^{(\min)})$ (or $\rho_\eta Xe^{i\alpha\theta_\eta}\psi\in \dom(\cD_{\wt{\bA}_{\full}}^{(\min)})$). 
\end{enumerate}
Using these two facts together with \eqref{eq:form_Dirac_adapted_basis}, we get that a spinor $\psi$ with $\supp\psi\subset B_\eps[\gamma_0]$
satisfies $\Omega^{-1}\psi\in\dom\,(\cD_{\bA_{\full}}^{(\max)})$ if and only if its corresponding function 
$\wt{f}=e^{i\alpha\theta_\eta}\begin{pmatrix}\cip{\xi_u}{\Omega^{-1}\psi} & \cip{\xi_d}{\Omega^{-1}\psi}\end{pmatrix}^T$ satisfies
\begin{equation}\label{eq:reformulate}
 \Big(\underbrace{-i\begin{pmatrix}\partial_u & 2\partial_w \\ 2\partial_{\overline{w}} & -\partial_u\end{pmatrix}}_{\cT_{\model}}\wt{f}\Big)\restriction\{\Theta_\eta(w)\neq 0\} 
 \in L^2(B_\eps^{(\eta)}[\gamma_0],\d u\d w_1\d w_2)^2.
\end{equation}

\medskip

\paragraph{\textit{Reinterpretation of \eqref{eq:reformulate}}}

For $\Omega^{-1}\psi\in\dom(\cD_{\bA_{\full}}^{(-)})$, we use Stokes' formula to decouple $\partial_u \wt{f}$ from the rest in \eqref{eq:reformulate}.
In that case $\partial_u \wt{f}$ is indeed square integrable. We prove the following more precise equality and then Lemma~\ref{lem:graph_norm_estimate}:
\begin{equation}\label{eq:graph_norm_formula}
 \int_{\Theta_\eta(w)\neq 0} |\cT_{\model}\wt{f}|^2\d u\d w_1\d w_2=\int_{\Theta_\eta(w)\neq 0} \big[|\partial_u \wt{f}|^2+|\sigma_\perp\cdot\nabla_w \wt{f}|^2\big]\d u\d w_1\d w_2.
\end{equation}

\begin{rem}
Consider $B_{\mathrm{cov}}$ and $C_{\mathrm{cov}}$, respectively the $2N$-folded and $N$-folded covering space of $B_\eps[\gamma_0]$.
We can lift $\gamma_{\eta}$ on $C_{\mathrm{cov}}$ in $N$ different ways, among which we consider the one $\widehat{\gamma}_\eta$ corresponding to $\mathrm{arg}(w)=0$.
The lift of $B_\eps^{(\eta)}[\gamma_0]$ along $\widehat{\gamma}_\eta$ defines an angular sector
 \[
    \{(u,w)\in \T_{N\ell}\times D_{\C}(0,\eps),\ -\pi\tfrac{M}{N}<\mathrm{arg}(w)< \pi\tfrac{M}{N}\},
 \]
whose lift on $B_{\mathrm{cov}}$ is written $\mathcal{C}_0$.
The integrals \eqref{eq:graph_norm_formula} can be understood as \emph{half} the integrals over this angular sector.
\end{rem}

Let us pick such a $\psi$ with support in $B_\eps[\gamma_0]$. By a density argument, it suffices to check \eqref{eq:graph_norm_formula}
on elements $\Omega^{-1}\psi$ satisfying $\partial_u^2f,\overline{\partial}_w\partial_wf\in L^2(\d u\d w\d \overline{w})^2$
(in \eqref{eq:descr_dom_D_full}-\eqref{eq:ansatz}, it suffices to take $f_{\gamma_0}(s)(s)\in H^2(\T_\ell)$ and a smooth regular part with support away from $\gamma_{\eta}$).

For $0<r<(\eta/2)^N$ let us consider the regions $R(\eps,r)\subset B_\eps^{(\eta)}[\gamma_0]$ defined by:
\[
 R(\eps,r)=\big\{\bx\in B_\eps^{(\eta)}[\gamma_0],\,q_\eta(\bx)\ge r\ \&\ \theta_\eta(\bx)=\Theta_\eta(w)\neq 0\mod(2\pi)\big\}.
\]
The boundary $\partial R(\eps,r)$ can be split into four parts: first we have the part $D_a:=\{\dist(\cdot,\gamma_0)=\eps\}$, and then 
\begin{multline*}
D_b:=\{|z(\bx)|=r\},\ D_c:=\Big\{\mathrm{arg}(w(\bx))=\Big((2k+1)\tfrac{\pi}{N}\Big)^{\pm}\Big\}\cap R(\eps,r),\\
\&\quad D_d:=\{\Theta_\eta(w(\bx))=0^{\pm}\}\cap R(\eps,r).
\end{multline*}
We use Stokes' formula on $R(\eps,r)$ for the difference of the L.H.S. and R.H.S. of \eqref{eq:graph_norm_formula}
and take the limit $r\to 0$. Due to the conditions on $\psi$, there is no boundary term on $D_a$. The boundary terms
on $D_d$ cancel each other due to the phase jump condition. The normals $\mathbf{n}_b,\mathbf{n}_c$ to $D_b$ and $D_c$ are orthogonal to $\bT\parallel \partial_u$, 
hence for $r$ small enough the boundary terms vanish.

Using \eqref{eq:partial_u_and_s} in \eqref{eq:graph_norm_formula}, we can replace $\partial_u$ by $\partial_s$ up to making a small error.
Seeing $f$ as a function on the angular sector on $B_{\mathrm{cov}}$, it defines a unique element $f_{\model}$ of $\mathcal{H}_{\zeta}$ \eqref{eq:mod_hilb},
and the obtained equality gives us the equivalence: $\Omega^{-1}\psi\in\dom(\cD_{\bA_{\full}}^{(-)})$ if and only if $f_{\model}\in\dom(\cD_{\model}^{(-)})$.

Controlling all the error terms, we obtain a similar result to \cite{dirac_s3_paper1}*{Lemma~27}.
\begin{lemma}\label{lem:graph_norm_estimate}
There exists $C_1>1$ and $C_2>0$ depending continuously on $\gamma_0,\alpha_a,\eps$ such that,
for $0<\eta\le \eta_0(\eps)$ and $\Omega^{-1}\psi\in\dom(\cD_{\bA_{\full}}^{(-)})$ with support in $B_\eps[\gamma_0]$, there holds:
\begin{equation*}
 -C_2\norm{\psi}_{L^2}^2+\frac{1}{C_1}\norm{\cD_{\bA_{\full}}^{(-)}\psi}_{L^2}^2\le \tfrac{1}{2N}\norm{f_{\model}}_{\cD_{\model}^{(-)}}^2\le C_1 \norm{\psi}_{\cD_{\bA_{\full}}^{(-)}}^2.
\end{equation*}
\end{lemma}
The factor $1/(2N)$ is due to the fact that $f_{\model}(x)$ is the element $L^2(B_{\mathrm{cov}})^2$ which covers $2N$ times $\wt{f}$.

\section{Proof of the auxiliary lemmas}\label{sec:proof_aux_lem}

We first prove the strong resolvent continuity 
of the homotopy \eqref{eq:loop_even} and then use it to prove Lemma~\ref{lem:compac}. 
We finish the section with the proof of Lemma~\ref{lem:continuity_along_collapse}.
Observe that the parameter space is metric,
henceforth it suffices to check the sequential continuity.
We write $\cD_\Omega$ for $\Omega^{-2}\sigma\cdot(-i\nabla^{\R^3})\Omega$, the free Dirac operator of $\S^3$ written in the stereographic chart $\R^3$.
We recall that $\Omega(\bx)=\tfrac{2}{1+|\bx|^2}$ is the conformal factor, and that the Banach spaces $L^2_{\Omega}$ and $H^1_{\Omega}$ are introduced in Section~\ref{sec:desc_dom}.

We use the localization functions at a small level $\delta>0$ (defining a partition of unity), 
Section~\ref{sec:loc}, and the phase jump functions $E_{\delta,0}$ and $E_{\delta,\underline{a}}$, see Section~\ref{sec:remov_phase}
and \eqref{eq:phase_jump_fun_surface}. 
Up to taking $n$ large enough, we can assume $\eta_n<\delta/2$ and the partition of unity
is still adapted to the link $\gamma_{\full}(\eta_n)$. 

We denote by $E_{\delta,0}^{(n)},E_{\delta,\underline{a}}^{(n)}$ the phase jump functions for
$\cD_n$ and by $E_{\delta,0}$ and $E_{\delta,\underline{a}}$ these of the limit operator $\cD$. As the auxiliary surface $S_{\aux}$
does not change they only differ through the auxiliary fluxes. In particular the functions $E_{\delta,0}^{(n)},E_{\delta,\underline{a}}^{(n)}$
converge to $E_{\delta,0},E_{\delta,\underline{a}}$ in the $L^\infty$-norm of their corresponding domains. We start with the following 
observation.

\begin{rem}[Convergence of the B.S. vector field]\label{rem:useful_and_trivial}
The B.S. gauge $\bA_{\gamma_{\eta_n}}(\bx)$ converges to 
$N\bA_{\gamma_0}(\bx)$ in $C_{\loc}^\infty(\R^3\setminus \gamma_0)^3$. In particular, for all $\eps>0$, we have:
\[
C_{\eps,\bA}:=\limsup_n\sup_{\bx\in \complement B_\eps[\gamma_0]}|\bA_{\gamma_{\eta_n}}(\bx)|<+\infty.
\]
\end{rem}

To shorten notation, we also write $\bA_{B.S.}^{(n)}$ for $2\pi t_n\bA_{\gamma_{\eta_n}}$ and $\bA_{\aux}^{(n)}$ for the auxiliary gauge
$\bA_{\full}(\eta_n,t_n)-\bA_{B.S.}^{(n)}$, with corresponding limits $\bA_{B.S.}=2Nt\pi \bA_{\gamma_0}$ and $\bA_{\aux}$.
At last, we write $Nt=\alpha+\lfloor Nt\rfloor$.

\subsection{Strong-resolvent Continuity}

Thanks to \cite{dirac_s3_paper1}, we know that the homotopy $(\cD_{\bA_{\full}})$ is continuous in the strong-resolvent sense in the range 
$(t,\eta)\in[0,\tfrac{M}{M+1}]\times (0,\eta_0)$, and $(t,\eta)\in [1,\tfrac{M}{M+1}]\times \{0\}$ (with $0\le \alpha_a<K$). Let us check the continuity for $t\in [0,1]$ and $\eta=0$,
and consider a sequence $(t_n,\alpha_a^{(n)},\eta_n)$ converging to $(t,\alpha_a,0^+)$. We denote by $\cD_n$ the corresponding Dirac operator, and $\cD$ the limit Dirac operator
$\cD_{\bA_{\full}(t,0,\alpha_a)}$.

We apply the same method as in \cite{dirac_s3_paper1} and check the second characterisation of Lemma~\ref{lem:char_sres_conv}.
Let $(\Omega^{-1} \psi,\cD(\Omega^{-1} \psi))$ be an element in the graph of $\cD$. We split $\psi$ with respect to the partition of unity \ref{def:part_unity}. 

\subsubsection{Away from the knots}
Consider $\chi_{\delta,\underline{a}}\psi$
with $\underline{a}\in\{2,3\}^{K+1}$. Because of the boundedness of $C_{\delta,\bA}$, we have
\begin{multline*}
\Omega^{-1}\psi_{\underline{a}}^{(n)}:=E_{\delta,\underline{a}}^{(n)}\overline{E}_{\delta,\underline{a}}\chi_{\delta,\underline{a}}\psi\in \dom(\cD_n),\\
\cD_n(\Omega \psi_{\underline{a}}^{(n)})=E_{\delta,\underline{a}}^{(n)}\big(\cD_{\Omega}+\Omega^{-1}\sigma\cdot\bA_{B.S.}^{(n)}\big)\overline{E}_{\delta,\underline{a}}\chi_{\delta,\underline{a}}\Omega^{-1} \psi.
\end{multline*}
By dominated convergence, $\Omega^{-1}\psi_{\underline{a}}^{(n)}$ and $\cD_n\Omega^{-1} \psi_{\underline{a}}^{(n)}$ 
converges to $\chi_{\delta,\underline{a}}\Omega^{-1}\psi$ and 
$\cD(\chi_{\delta,\underline{a}}\Omega^{-1} \psi)$ in $L^2_{\Omega}$.

\subsubsection{In the vicinity of an auxiliary knot $\gamma_k$, $1\le k\le K$}
Observe that for any function $\phi\in L^2_{\Omega}$, the sequence $(\Omega^{-1}\sigma\cdot\bA_{B.S.}^{(n)}\phi)$ converges to
$Nt\Omega^{-1}\sigma\cdot\bA_{\gamma_0}\phi$ in $L^2_{\Omega}$. 

In the region $\supp\,\chi_{\delta,k}$, the operator $\cD_n-\Omega^{-1}\sigma\cdot\bA_{B.S.}^{(n)}$
 coincides with the Dirac operator with the (singular) magnetic gauge $\bA_{\aux}(t_n)$ only. By strong-resolvent continuity
 of such operators \cite{dirac_s3_paper1}*{Theorem~23}, there exists a sequence $(\psi_{\delta,k}^{(n)})$ subordinated to
 $\cD_{\bA_{\aux}^{(n)}}$ such that
 \[
 \Big(\psi_{\delta,k}^{(n)},\cD_{\bA_{\aux}^{(n)}}\psi_{\delta,k}^{(n)}\Big)\underset{n\to+\infty}{\longrightarrow}
 	\Big(\chi_{\delta,k}\Omega^{-1}\psi,\cD_{\bA_{\aux}}[\chi_{\delta,k}\Omega^{-1}\psi]\Big)\ \mathrm{in}\ \big[L^2_{\Omega}\big]^2.
 \]
 Up to localizing, we can assume that $B_{\delta'}[\gamma_0]\cap \supp\,\psi_{\delta,k}^{(n)}=\emptyset$ for all $n$.
Thus $\psi_{\delta,k}^{(n)}\in\dom(\cD_n)$ and we have:
\[
(\psi_{\delta,k}^{(n)},\cD_n(\psi_{\delta,k}^{(n)}))\underset{n\to+\infty}{\longrightarrow}
 	\Big(\chi_{\delta,k}\Omega^{-1}\psi,\cD[\chi_{\delta,k}\Omega^{-1}\psi]\Big)\ \mathrm{in}\ \big[L^2_{\Omega}\big]^2.
\]

\subsubsection{In the vicinity of $\gamma_0$}
At last, let us deal with $\chi_{\delta,\gamma_0}\psi$. 
We use \cite{dirac_s3_paper1}*{Proof~of~Theorem~7} and Theorem~\ref{thm:behavior_BS}. The spinor $f=\begin{pmatrix}f_+& f_- \end{pmatrix}^T$ 
in the space $L^2(\T_{\ell}\times\C)^2$ defined by the formula: 
\[
e^{-i Nt(\theta+\Phi_{\gamma_0}+g_{\gamma_0})}\overline{E}_{\delta,0}\chi_{\delta,\gamma_0}\psi=f_+\xi_++f_-\xi_-,
\] 
and seen as a function of the tubular coordinates $(s,\rho e^{i\theta})$ is in the domain of the model operator $\cD_{\T,\alpha}^{(-)}$ 
(cut along $\theta=0$). This operator (in its Coulomb gauge) is described in
\cite{dirac_s3_paper1}*{Section~3.2.2}, where $s,\rho,\theta$ have to be considered as the tubular coordinates in the flat metric.
As in Section~\ref{sec:model_op} (see \cite{dirac_s3_paper3}*{Lemma~32}), an essential domain is given
by the subspace spanned by the minimal domain $\dom(\cD_{\T,\alpha}^{(\min)})$ and, for $\alpha>0$, the functions $h_{\sing,j}$ defined by
\begin{equation*}
	h_{\sing,j}(s,\rho,\theta)=\frac{1}{\sqrt{2\pi\ell}}
		e^{ijs}\begin{pmatrix} 0\\ e^{i\alpha\theta}\end{pmatrix},\ j\in T_{\ell}^*=\frac{2\pi}{\ell}\Z.
\end{equation*}
Outside the cut $\theta=0$, the operator $\cD_{\T_\ell,\alpha}^{(-)}$ acts like $-i\sigma_3\partial_s-i\sigma_\perp\cdot(-i\nabla_\perp)$,
the latter term denoting the free Dirac operator in $\C$. The minimal domain $\dom(\cD_{\T,\alpha}^{(\min)})$ is the graph norm closure of 
$e^{i\alpha\theta}C^1_0(\{\rho=0\}^c)^2$, where $\{\rho=0\}^c$ stands for $\complement_{\T_\ell\times\C}\{(s,\rho e^{i\theta}),\ \rho=0\}$.

Up to taking $\delta$ small enough, $\supp\,\chi_{2\delta,0}$ does not intersect the auxiliary knots, and we have
\[
\chi_{2\delta,0}(\rho)[f_0+f_{\sing}]=f_0+f_{\sing}.
\]
Substituting $K_\alpha(\rho)$ by $\chi_{2\delta,0}(\rho)\rho^{-\alpha}$ defines the spinor $\wt{f}_j$.
The difference $\tfrac{\Gamma(\alpha)}{2^{1-\alpha}}\wt{f}_j-h_{\sing,j}$ is in $\dom(\cD_{\T,\alpha}^{(\min)})$.
We now provide an approximation of the spinor $\Omega^{-1}\wt{\psi}_j$ defined by the relation
\[
e^{-i Nt(\theta+\Phi_{\gamma_0}+g_{\gamma_0})}\overline{E}_{\delta,0}\wt{\psi}_j=\sum_{\circ\in\{\pm\}}\wt{f}_{j,\circ}(s,\rho,\theta)\xi_{\circ}(s).
\]
We use the convergence of the tubular coordinates of $\gamma_{\eta}$ as $\eta\to 0^+$, section~\ref{sec:prop_adap_seif_fib}, 
and the form of the free Dirac operator in tubular coordinates 
(same form as \eqref{eq:form_Dirac_adapted_basis}, but with $w$ and $W_\xi$ replaced by $z:=\rho e^{i\theta}$ and $M_\xi$ respectively).
From the equality $\overline{z}^{-\alpha}=\overline{z}^{\lfloor Nt \rfloor} /\overline{z}^{Nt}$ and Lemma~\ref{lem:refinement_BS_behav}, we infer 
\[
\wt{\psi}_j^{(n)}(\bx)=\mathrm{exp}\big[it[\theta_{\eta_n}(\bx)+\Phi_{\eta_n}(\bx)+g_{\eta_n}(\bx)]+ijs(\bx)\big]\frac{\chi_{2\delta,0}(\rho(\bx))}{q_{\eta_n}(\bx)^t}\overline{z(\bx)}^{\lfloor Nt \rfloor}.
\]
Remember the convergence of the  tubular coordinates of $\gamma_{\eta}$ as $\eta\to 0^+$, section~\ref{sec:prop_adap_seif_fib},
and that of $\Phi_{\eta_n},g_{\eta_n}$ (Lemma~\ref{lem:refinement_BS_behav}): they imply that of $\Omega^{-1}\wt{\psi}_j^{(n)}$
to $\Omega^{-1}\wt{\psi}_j$ in $L^2_{\Omega}$. 
A computation shows that the free Dirac operator in tubular coordinates 
has the same form as \eqref{eq:form_Dirac_adapted_basis}, but with $w$ and $W_\xi$ replaced by $z:=\rho e^{i\theta}$ and $M_\xi$ respectively. 
Hence $\cD_n (\Omega^{-1}\wt{\psi}_j^{(n)})$ converges to $\cD(\Omega \wt{\psi}_j)$ in $L^2_{\Omega}$.

\medskip
There remains to deal with elements in the minimal domain $\dom(\cD_{\T,\alpha}^{(\min)})$: they correspond to elements
in $\dom(\cD^{(\min)})$. See the graph norm relation \cite{dirac_s3_paper3}*{Proposition~36} (and also \cite{dirac_s3_paper1}*{Lemma~27}): 
a density argument for $\cD_{\T_\ell,\alpha}^{(-)}$ also works for $\cD_n$ for localized element around $\gamma_0$ 
(even though the second paper deals with other coordinates, the proof is the same \emph{mutatis mutandis}, and the same kind of estimates hold here). 

Observe that for all $\phi\in \dom(\cD^{(\min)})$ and $\delta_1>0$,
we have $(1-\chi_{\delta_1,0})\phi\in \dom(\cD^{(\min)})\cap \dom(\cD_n^{(\min)})$ as long as $\eta_n<\delta_1$. Then the convergence of the B.S.
gauge (Remark~\ref{rem:useful_and_trivial}) ensures the convergence $\cD_n((1-\chi_{\delta_1,0})\phi)\to \cD((1-\chi_{\delta_1,0})\phi)$
in $L^2_{\Omega}$. The rest follows by density.

\subsection{Convergence of energy bounded sequence (Proof of Lemma~\ref{lem:compac})}

By the Banach-Alaoglu theorem, and up to the extraction of a subsequence, we can assume that
we have the weak convergence:
\[
(\psi_n,\cD_n\psi_n)\rightharpoonup (\psi,\psi')\ \mathrm{in}\ \big[L^2_{\Omega}\big]^2.
\]
It suffices to show that the limit $(\psi,\psi')$ is in the graph of the maximal operator $\cD^{(\max)}$, that is 
$\psi$ is in $H_{\loc}^1(\R^3\setminus(S_{\aux}\cup\gamma_0))^2$ with 
\[
[\Omega^{-2}\sigma\cdot(-i\nabla+\bA_{B.S.})\Omega\psi]\restriction{\R^3\setminus(S_{\aux}\cup\gamma_0)}\in L_{\Omega}^2.
\]
This follows from the strong-resolvent convergence of $\cD_n$. Indeed, for any $(\phi,\cD\phi)$ in the graph of $\cD$, there exists a sequence
$(\phi_n,\cD_n\phi_n)$ of elements in the graphs of the $\cD_n$'s converging to $(\phi,\cD\phi)$. We obtain:
\begin{equation}\label{eq:smart_argu}
	\cip{\cD_n\psi_n}{\phi_n}_{L_{\Omega}^2}=
		\cip{\psi_n}{\cD_n\phi_n}_{L_{\Omega}^2}\underset{n\to+\infty}{\longrightarrow}\cip{\psi}{\cD\phi}_{L_{\Omega}^2}.
\end{equation}
That is $\psi\in\dom(\cD^*)=\dom(\cD)$. The convergence for $\phi$ in $\dom(\cD^{(\min)})$ 
gives $\psi'=\cD\psi$. Let us turn to the proof.

\medskip

We localize the sequence $(\psi_n)$ with respect to a small level $\delta>0$. 

Consider $\chi_{\delta,\underline{a}}\psi_n$ with $\underline{a}\in\{2,3\}^{K+1}$.
The function $\overline{E}_{\delta,\underline{a}}^{(n)}\chi_{\delta,\underline{a}}\psi_n$
does not exhibit any phase jump. Hence it is $H^1_{\Omega}\simeq H^1(\S^3)^2$,
and the corresponding sequence is $H^1$-bounded. Up to the extraction of a subsequence, it converges in $L^2_{\Omega}$, and the $L^\infty$-convergence of $E_{\delta,\underline{a}}^{(n)}$ gives that of $\chi_{\delta,\underline{a}}\psi_n$, with the correct phase jump across the surfaces $S_k$'s. 

Consider $\cD_{\sing,n}:=\cD_n-\Omega^{-1}\sigma\cdot\bA_{B.S.}^{(n)}$,
which acts like $\cD_\Omega$ away from the link $\gamma_{\full}(\eta_n)$ and the surfaces $S_k$.
Since we also have convergence of the Biot and Savart gauge (Remark~\ref{rem:useful_and_trivial}),
the sequence $(\cD_n\chi_{\delta,\underline{a}}\psi_n)$ also converges to $\cD \chi_{\delta,\underline{a}}\psi$ weakly in $L^2_{\loc}(\R^3\setminus(S_{\aux}\cup\gamma_0))^2$. 
Then, repeating the argument \eqref{eq:smart_argu} for $\chi_{\delta,\underline{a}}\psi_n$
and $\phi\in\dom(\cD^{(\min)})$ gives $\chi_{\delta,\underline{a}}\psi\in\dom(\cD^{(\max)})$.

Recall that $\dom(\cD^{(\min)})$ is the graph norm closure of elements in $H^1(\R^3\setminus(S_{\aux}\cup\gamma_0),\Omega^3\d\bx)^2$ (defined w.r.t. the connection $\nabla^{(\Omega)}$, \eqref{eq:connexion_Omega}) 
with the correct phase jump across the $S_k$'s and whose supports do not intersect the link $\gamma_{\lim}$ (see \cite{dirac_s3_paper1}). 

Since $\delta$ was arbitrary, a diagonal extraction subordinated to 
$\delta_n=2^{-n}\delta$ followed by the argument \eqref{eq:smart_argu} applied
to $\psi_n$ and $\phi$ in the described above set gives $\psi\in\dom(\cD^{(\max)})$.

\subsection{Study of the bump continuity (Proof of Lemma~\ref{lem:continuity_along_collapse})}

Let us now investigate the bump-continuity of the family of Dirac operators. We emphasize that we can indifferently pick the B.S. gauge or a singular gauge for $2\pi t[\gamma_{\eta}]$ (the spectrum is gauge invariant).
The cases which are not covered by the bump-continuity results \cite{dirac_s3_paper3}*{Theorems~14~$\&$~15} are when $\eta\to 0^+$
for any value of the auxiliary flux $\alpha_a$.
We study the bump continuity of $\cD_{\bA_{\full}}$ at such a point $(t,0,\alpha_a)=:\bp$.

\subsubsection{Reduction to the study of a vanishing quasimode}
As in \cite{dirac_s3_paper2}*{Theorem~15}, we show that the failure of bump continuity at a level $\lambda$ is equivalent to the existence of
a vanishing quasimode at this level along a sequence $(t_n,\eta_n,\alpha_a^{(n)})\to \bp$. We refer to the proof of this theorem for full details of the argument.

Let $\Lambda$ be a bump function centered at $\lambda$, and assume that $\Lambda(\cD_{\bA_{\full}(\cdot)})$ is not continuous at $\bp$. Then there exists a sequence $(\bp_n)$ converging to $\bp$ with 
\[
\limsup_n \norm{\Lambda(\cD_{\bA_{\full}(\bp_n)})-\Lambda(\cD_{\bA_{\full}(\bp)})}_{\mathcal{B}}>0.
\]
Since the Dirac operators in the range of $\cD_{\bA_{\full}}$ all have discrete spectrum, then the strong resolvent continuity together and 
Lemma~\ref{lem:compac} imply the existence of a sequence of \emph{normalized} eigenfunctions 
$(\psi_n)$, satisfying
\[
\psi_n'\in\ker(\cD_{\bA_{\full}(\bp_n)}-\lambda_n),\  \lambda_n\in (\supp\,\Lambda)^{\circ},\ \&\ \psi_n'\rightharpoonup_n 0\ \mathrm{in}\ L^2_{\Omega}.
\]
Furthermore, the sequence $(\psi_n)$ collapses onto $\gamma_0$:
\[
\lim_{\eps\to 0}\liminf_n \int_{B_\eps[\gamma_0]} |\psi_n'|^2\Omega^3=1.
\]
Now if continuity fails for all bump functions centered at $\lambda$, then a diagonal argument along a sequence $(\Lambda_n)$ of bump functions with support decreasing to $\{\lambda\}$ provides us with a sequence $(\bp_n)$ converging to $\bp$ together with a sequence of normalized eigenfunctions $(\psi_n')$, $\psi_n'=:\Omega^{-1}\psi_n\in \dom(\cD_{\bA_{\full}(\bp_n)})$ satisfying:
	\[
	\left\{
		\begin{array}{ccc}
			\Omega^{-2}\cD_n\psi_n&=&\lambda_n\Omega^{-1}\psi_n,\\
			\lambda_n&\to_n&\lambda,\\
			\Omega^{-1}\psi_n&\underset{L^2_{\Omega}}{\rightharpoonup}&0, 
		\end{array}
	\right.
	\]
and
$
1=\int \Omega|\psi_n|^2=\lim_{\eps\to 0}\liminf_n \int_{B_\eps[\gamma_0]} |\psi_n'|^2\Omega.
$
Reciprocally, if such a sequence exists, the map $\cD_{\bA_{\full}(\cdot)}$ is not $\lambda$-bump continuous at $\bp$.

\subsubsection{Study of a collapsing sequence}
Let us study in detail a collapsing sequence $(\Omega^{-1}\psi_n)$ (described in the previous section). 
For short we write $\cD_n$ instead of $\cD_{\bA_{\full}(\bp_n)}$. We will denote $T_n$ the corresponding operator
in the flat metric of $\R^3$:
\[
T_n:= \Omega^2 \cD_n\Omega^{-1}:\clos_{\cG}\, \Omega \dom(\cD_n)\to L^2(\R^3)^2.
\]

Our aim is to show that $\lambda$ must be in the set described in the theorem, and reciprocally
that this fact implies the non bump continuity.

\medskip

The eigen-equation can be rewritten: $(T_n-\lambda_n\Omega)\psi_n=0$.
Since the function $\Omega^{-1}\psi_n$ collapses onto $\gamma_0$, 
for any level of localization $\delta>0$, see Section~\ref{sec:loc}, $\Omega^{-1}\chi_{\delta,\gamma_0}\psi_n$ defines a sequence
of $\lambda$-quasimode:
\[
\lim_n \norm{(\cD_n-\lambda)\Omega^{-1}\chi_{\delta,\gamma_0}\psi_n}_{L^2_{\Omega}}=0.
\]
Since $\Omega$ and $\Omega^{-1}$ are bounded around $\gamma_0$, 
we can work with the flat metric metric. In particular we also have
\[
\lim_n \int_{\R^3}|(T_n-\lambda_n\Omega) \chi_{\delta,\gamma_0}\psi_n|^2=0.
\]

Thanks to the localization, we can use Remark~\ref{rem:wlog}: up to the gauge transformation 
$e^{it_n(\Phi_{\eta_n}+g_{\eta_n})}\overline{E}_{\delta,0}$, we can assume that $\bA_{\full}$ takes the form:
$t_n\d\theta_{\eta_n}-2\pi\alpha_a^{(n)}\tfrac{\d s}{\ell}$. Up to another gauge transformation $e^{it\theta_{\eta_n}}$ 
(with say a cut along $\theta_{\eta_n}=\theta_0$, and $\theta_0\in[\R\setminus \pi\mathbb{Q}]/(2\pi\Z)$ fixed), 
it takes the form $-2\pi\alpha_a^{(n)}\tfrac{\d s}{\ell}$, and the function
exhibits the phase jump $e^{-2i\pi t_n}$ across $\{\theta_{\eta_n}=\theta_0\}$.

\medskip
Now we analyze $\chi_{\delta,\gamma_0}\psi_n=:\psi_{\delta}^{(n)}$ following the method of Section~\ref{sec:model_op} 
(and we refer the reader to this section for the definition of all the geometrical objects).

We form the associated spinor $f^{(n)}=\begin{pmatrix}f_+^{(n)} & f_-^{(n)} \end{pmatrix}^T$ made of its coordinates
in the basis $(\xi_u,\xi_d)$. We write $f^{(n)}$ in cable coordinates $(s,w)$ \eqref{eq:new_coord}, and its lift
on $L^2(B_{\mathrm{cov}}\times \C)^2$, written $F^{(n)}$, is seen as an element in the domain of the model case $\cD_{\model}^{(n)}$
associated with $\eta_n$ and $\alpha=t_n$.
Equipped with Lemma~\ref{lem:graph_norm_estimate}, we decompose $f^{(n)}$ into Fourier modes with values in $\dom(\cD_{\eta_n,t_n})$.
Recall also Formula~\eqref{eq:form_Dirac_adapted_basis}, which gives the form of the Dirac operator in cable coordinates. 

In particular, we can decompose $F^{(n)}$ into Fourier modes of $L^2(\T_{2N\ell})$ with values in $\dom(\cD_{\eta_n,t_n}))$:
\begin{equation}\label{eq:fourier_decomp}
F^{(n)}(s,w)=\frac{1}{\sqrt{2N\ell}}\sum_{j\in\T_{2N\ell}^*}e^{ijs}h_j^{(n)}(w),\ h_j^{(n)}\in \dom(\cD_{\eta_n,t_n}).
\end{equation}
We make several remarks.

\medskip
\noindent 1. Since we deal with localized elements around $\gamma_0$, convergence in $L^2(B_\eps[\gamma_0])^2$ 
implies convergence for the corresponding spinor in $L^2(B_\eps[\T_{2N\ell}\times\{0\}],\d s\d w)^2$ and vice versa.
We also recall that on this tubular neighborhood, we have $\partial_u=\tfrac{1}{h}\partial_s$, where $0<c_b\le h\le c_h<+\infty$.

\medskip
\noindent 2. Due to the collapse of $\psi_n$, we have $\lim_n \norm{\rho F^{(n)}}_{L^2(B_{\mathrm{cov}})^2}=0$.
Thanks to \eqref{eq:graph_norm_formula}-\eqref{eq:model_case_energy_equalities} 
and Lemma~\ref{lem:graph_norm_estimate}, we also have
\[
\lim_n \norm{\partial_s(\rho F^{(n)})}_{L^2(B_{\mathrm{cov}})^2}=0.
\]
\medskip
\noindent 3. Using Formula~\eqref{eq:form_Dirac_adapted_basis}, for $(u,w)\in B_\delta^{(\eta_n)}[\gamma_0]$, see Section \ref{sec:graph_norm_rel}, we obtain:
\begin{multline*}
\begin{pmatrix}
	\cip{\xi_u}{T_n \psi_n}\\
	\cip{\xi_d}{T_n \psi_n}
\end{pmatrix}(s,w)= \Big[\cD_{\model}^{(n)} F^{(n)}-\frac{2\pi}{\ell}\alpha_a^{(n)}\sigma_3F^{(n)}\Big](s_l,w_l)
\\
+W_\xi(s)F^{(n)}(s_l,w_l)+o_{L^2}(1),\ n\to+\infty.
\end{multline*}
Above, $(s_l,w_l)$ denotes the point on the lift 
of $B_\delta^{(\eta_n)}[\gamma_0]$ in $B_{\mathrm{cov}}$ corresponding to $(s,w)$ like in Section \ref{sec:graph_norm_rel}),
(after having prescribed the lift
of $\gamma_{\eta_n}(s=0)$ to $(0,\eta_n/(8N))\in \T_{2N\ell}\times \C$, see Section \ref{sec:graph_norm_rel}).
\begin{rem}
As a word of caution: here the points are given in \emph{cable coordinates}, not in tubular coordinates of $\gamma_{\eta_n}$.
Running along $\gamma_{\eta_n}$ \emph{twice} from $s=0$ to $s=2N\ell$, $s_l$ corresponds to the \emph{continuous} lift to $B_{\eps}(\T_{2N\ell}\times\{0\})$ 
of the projection onto $\gamma_0$. In particular $s_l(\bx)$ is always equal to $s(\bx)$ modulo $\ell$.
\end{rem}

Similarly, the quasimode relation becomes:
\begin{multline}\label{eq:eig_eq}
\Big[\cD_{\model}^{(n)} -\frac{2\pi}{\ell}\alpha_a^{(n)}\sigma_3+W_\xi(s)-\lambda\Omega(s)\Big]F^{(n)}(s_l,w_l)
=o_{L^2}(1),\\ 
W_\xi(s):=W_\xi(\gamma_0([s])),\ \Omega(s):=\Omega(\gamma_0([s])),\ [s]=s+\R/\ell\Z.
\end{multline}

Consider $j\in\T_{2N\ell}^*$. We split $h_j^{(n)}$ into its regular part $h_{j,0}^{(n)}$ and its singular part $h_{j,\sing}^{(n)}$,
and we further decompose the latter with respect to the basis $(f_{\sing,k}^{(n)})$ \eqref{eq:def_f_sing_2D}: $h_{j,\sing}^{(n)}=\sum_k c_{j,k}^{(n)}f_{\sing,k}^{(n)}$.
We call $F_0^{(n)}$ the regular part of $F^{(n)}$, $F_{\sing,k}^{(n)}$ its $k$-singular part, $e_{k}^{(n)}(s):=\sum_{j\in\T_{2N\ell}^*}c_{j,k}^{(n)}e^{ijs}$ and $F_{\sing}^{(n)}=\sum_k F_{\sing,k}^{(n)}$:
\[
F_0^{(n)}(s,w):=\sum_{j\in\T_{2N\ell}^*}e^{ijs}h_{j,0}^{(n)}(w),\ F_{\sing,k}^{(n)}(s,w):=
e_{k}^{(n)}(s)f_{\sing,k}^{(n)}(w).
\]

The boundedness of the $L^2$-norm ensures $\sum_{j,k}|c_{j,k}^{(n)}|^2\le 1$ for all $n$.
By a diagonal extraction argument, we can assume that all the sequences $(c_{j,k}^{(n)})_{n\in\ge 0}$ converge.
From Lemma~\ref{lem:dichotomy_sing_basis} and the boundedness of $F^{(n)}$ in graph norm, 
we get $\lim_n c_{j,k}^{(n)}=0$ for any $0\le k\le N-1$, $Nt-k<1$ (this follows from collapse of $F^{(n)}$ for $k=Nt$, and from graph norm boundedness in the other cases).  
For $k<\lfloor Nt\rfloor$, $f_{\sing,k}^{(n)}$
collapses to $0$.  Hence $c_{j,k}^{(n)}f_{\sing,k}^{(n)}$ either converges to $0$ or collapses.

By Lichnerowicz' formula (second line of \eqref{eq:model_case_energy_equalities}), and up to the extraction of a subsequence, we can assume that the regular part $F_0^{(n)}$ converges in $L^2(\T_{2N\ell}\times\C)^2$: its limit is necessarily $0$.

Reconsider \eqref{eq:eig_eq} in light of these convergences, we have:
\begin{multline}\label{eq:light}
\Big[\cD_{\model}^{(n)}-\frac{2\pi}{\ell}\alpha_a^{(n)}\sigma_3+W_\xi(s)-\lambda\Omega(s)\Big]\sum_{k:Nt-k\ge 1}F_{\sing,k}^{(n)}\\
+\cD_{\model}^{(n)}\big(\sum_{k:Nt-k<1}F_{\sing,k}^{(n)}+F_{0}^{(n)}\big)=o_{L^2}(1).
\end{multline}
As in the proof of \cite{dirac_s3_paper2}*{Theorem~15}, we expect the following behavior. 
The lower terms cannot cancel the upper ones, hence converges also to $0$. This leads to an effective eigen-equation for the $F_{\sing,k}^{(n)}$
and forces $\lambda$ to be on the spectrum of an effective operator $\cT_{k}$ for some $k$, $Nt-k\ge 1$. 

From Lemma~\ref{lem:decomp_sing_subspace} and \eqref{eq:model_case_energy_equalities}, we know that $(e_{k}^{(n)}(s))_{0\le k\le N-1}$ is an orthogonal family in $H^1(\T_{2N\ell})$ and $L^2(\T_{2N\ell})$, with uniformly bounded energy. 
We decompose \eqref{eq:light} with respect to $L^2(\T_{2N\ell})\otimes L^2(\C)^2$.
The (bounded) multiplication operator $-\frac{2\pi}{\ell}\alpha_a^{(n)}\sigma_3+W_\xi(s)-\lambda\Omega(s)$ satisfies:
\[
-\frac{2\pi}{\ell}\alpha_a^{(n)}\sigma_3+W_\xi(s+\ell)-\lambda\Omega(s+\ell)=-\frac{2\pi}{\ell}\alpha_a^{(n)}\sigma_3+W_\xi(s)-\lambda\Omega(s),
\]
hence the term $\Big[\cD_{\model}^{(n)}-\frac{2\pi}{\ell}\alpha_a^{(n)}\sigma_3+W_\xi(s)-\lambda\Omega(s)\Big]F_{\sing,k}^{(n)}$, and especially its 
spin down can only be compensated by a term $\cD_{\model}^{(n)} F_{0}^{(n)}$ (since the other terms are orthogonal to it in $L^2(\T_{2N\ell}\times\C)^2$). If we show that this is impossible, then from 
the spin down equation we will obtain the quasi-mode relation:
\begin{equation}\label{eq:eigeq_spin_down}
	     \Big(i\partial_s +(N\alpha-k)\tau(s)+\tfrac{2\pi}{\ell}\alpha_a+i\cip{\xi_-}{\partial_s\xi_-}(\gamma_0(s))-\lambda \Omega(\gamma_0(s))\Big)e_{k}^{(n)}(s)=o_{L^2}(1).
\end{equation}
Since we have $\liminf_n\norm{e_{k}^{(n)}}_{L^2}>0$ for at least one $k$, this implies that $\lambda$ is in the spectrum of the effective operator
$\cT_{\gamma_0}^{(-)}+(N\alpha-k)\tau_{\bS,\S^3}$ on $\gamma_0$ (where $\tau_{\bS,\S^3}(\gamma_0(s))=\frac{\tau(\gamma_0(s))}{\Omega(\gamma_0(s))}$ is the relative torsion of $\bS$ in the metric of $\S^3$), see \cite{dirac_s3_paper2}*{Section~3.3}.
That is, $\lambda$ is in the set \eqref{eq:crit_eig}.

\paragraph{\emph{No possibility of compensation}}
Consider \eqref{eq:eigeq_spin_down}, rewritten $a_{k,-}^{(n)}(s)=o_{L^2}(1)$. 
We show by a contradiction argument that $\lim_n\norm{a_{k,-}^{(n)}(s)}_{L^2}=0$, or in other words that $a_{k,-}^{(n)}(s)f_{\sing,k}^{(n)}(w)$
cannot be compensated by the term $\cD_{\model}^{(n)}F_0^{(n)}$. 
Assume $\limsup_n\norm{a_{k,-}^{(n)}(s)}_{L^2}>0$, $\lim_n\norm{a_{k,-}^{(n)}(s)f_{\sing,k}^{(n)}(w)-\cD_{\model}^{(n)}F_0^{(n)}}_{L^2}=0$.

Recall that $F_0^{(n)}$ satisfies the second equality of \eqref{eq:model_case_energy_equalities}. By collapse of $f_{\sing,k}^{(n)}$, the spin-down component of $F_0^{(n)}$ cannot help compensating $a_{k,-}^{(n)}(s)f_{\sing,k}^{(n)}(w)$. Decomposing $\cD_{\model}^{(n)}F_0^{(n)}$
with respect to $L^2(\T_{2N\ell})\otimes L^2(\C)^2$, we see that the putative compensating term can only be of the form
$\wt{a}_{k}^{(n)}(s)h_{k,0}^{(n)}(w)$, with $h_{k,0}^{(n)}(w)\in\dom(\cD_{\eta_n,t_n}^{(\min)})$ and $\lim_n\norm{\wt{a}_{k}^{(n)}-a_{k,-}^{(n)}}_{L^2}=0$.
Since $\limsup_n\norm{a_{k,-}^{(n)}(s)}_{L^2}>0$,
we have $\lim_n\norm{h_{k,0}^{(n)}}_{L^2(\C)^2}=0$ (the regular part converges to $0$). 
But Lemma~\ref{lem:dichotomy_sing_basis} (case $k<E(N\alpha)$) gives:
\begin{multline*}
\cip{\cD_{\model}^{(n)}\wt{a}_{k}^{(n)}h_{k,0}^{(n)}}{a_{k,-}^{(n)}f_{\sing,k}^{(n)}}_{L^2}=-i\cip{\tfrac{\d}{\d s}\wt{a}_{k}^{(n)}}{a_{k,-}^{(n)}}_{L^2(\T_{2N\ell})}\cip{h_{k,0}^{(n)}}{f_{\sing,k}^{(n)}}_{L^2(\C)^2}\\
	+\cip{\wt{a}_{k}^{(n)}}{a_{k,-}^{(n)}}_{L^2(\T_{2N\ell})}\cip{h_{k,0}^{(n)}}{\cD_{\eta_n,t_n}f_{\sing,k}^{(n)}}_{L^2(\C)^2}\underset{n\to+\infty}{\longrightarrow} 0.
\end{multline*}

\subsubsection{Existence of collapsing sequences}
We turn to the reciprocal statement: given $\lambda$ in the set \eqref{eq:crit_eig}, we construct a collapsing sequence along $\eta_n\to 0$
and $t_n\to t$ (if $t=1$, $t_n\to 1^-$) and $\alpha_a^{(n)}\equiv \alpha_a$. We call $e_{k,\lambda}(s)$ the corresponding eigenfunction of the effective operator. We use \cite{ErdSol01}*{Section~1} and \cite{dirac_s3_paper2}*{Appendix}: we have $i\cip{\xi_-}{\partial_s\xi_-}(\gamma_0(s))=\omega(s)-\tfrac{\tau}{2}$ with $\int_0^\ell\omega(s)\d s\equiv \pi\mod 2\pi$. So $e_{k,\lambda}(s)$ takes the form
\[
e_{k,\lambda}(s)=\frac{1}{\sqrt{2N\pi \ell}}\mathrm{exp}\Big[iNt I_0(s)+\tfrac{2\pi}{\ell} s-i\lambda \int_0^s (\Omega-i\cip{\xi_-}{\partial_s\xi_-})(\gamma_0(s'))\d s'\Big].
\]
We emphasize that $e^{-i(k+\tfrac{1}{2})I_0(s)}e_{k,\lambda}(s)$ is $\ell$-periodic and defines an element in $C^1(\T_\ell)$.
The ansatz $F_{\ans,\lambda}^{(n)}$ of a $\lambda$-quasimode is given up to normalisation by
\[
e_{k,\lambda}(s)\frac{\chi_{\delta,\gamma_0}(w)}{(\overline{w}^N-\eta_n^N)^{t_n}}\Big[1-\tfrac{1}{2}we^{-iI_0(s)}\cip{\xi_+}{\partial_s \xi_-}(\gamma_0(s))\Big]\begin{pmatrix} 0 \\ 1\end{pmatrix}.
\]
\begin{rem}
We emphasize that the ansatz is given in the same local gauge as the one for $F^{(n)}$: with a cut along $\{\theta_{\eta_n}=\theta_0\}$
but without phase jumps due to the auxiliary knots. It is made of the components of the collapsing sequence in the $(\xi_u,\xi_d)$-basis.
The reader might argue that the second term also has a non-vanishing singular part in general. This is true, but its regular part is much bigger in graph norm of $\cD_{\model}^{(n)}$.
\end{rem}

\begin{appendix}

\section{Proof of the first part of Theorem~\ref{thm:behavior_BS}}\label{sec:first_part_proof_thm_behavior}

  \paragraph{\emph{Decomposition of the integral}}
    We fix $\eps\in (0,\tfrac{\ell}{2})$.
    
    For any $\bx$ in the vicinity of $\gamma$, 
    we decompose the domain of integration $\T_{\ell}=\R/(\ell\Z)$ into two: a window $J^{(\eps)}:=(s_0-\eps,s_0+\eps)$
    of fixed size around $s_0$, and the remainder. We are interested in the limit of $\bA(\bx)$ as $\rho\to 0$,
    in particular at fixed $s_0$. By dominated convergence, 
    it is clear that the integral over $\T_{\ell}\setminus J^{(\eps)}$ (at fixed $s_0$) is differentiable in $\rho$, giving:
    \begin{equation}\label{eq:part_remainder}
     \int_{|s-s_0|\ge \eps}\dot{\gamma}(s)\times\frac{\bx-\gamma(s)}{|\bx-\gamma(s)|^3}=\int_{|s-s_0|\ge \eps}\dot{\gamma}(s)\times\frac{\gamma(s_0)-\gamma(s)}{|\gamma(s_0)-\gamma(s)|^3}
      +\underset{\rho\to 0}{\mathcal{O}}(\rho).
    \end{equation}
    The difficulty is of course to evaluate the expansion of the integral over $J^{(\eps)}$.
    
    \medskip
    
    \paragraph{\emph{Notations for the proof}}
    By a change of basepoint we can assume $s_0=0$.
    To shorten notations, the arguments of the functions are written as subscripts, 
    for instance $\gamma(0)$ and $\gamma_{0}$ denotes the same point.
    We will also use the following notations:
    \begin{equation*}
     \delta=\sqrt{s^2+\rho^2},\ A=\frac{\eps}{\rho},\ \langle s\rangle=\sqrt{1+s^2}.
    \end{equation*}
    
    We use Taylor-Lagrange formula with integral remainder and write:
    \begin{equation*}
    \begin{array}{rcl}
      \gamma_s-\gamma_0&=& \sum_{n=1}^N\frac{s^n}{n!}\gamma_0^{(n)}+g_{n+1}(s),\\
      \dot{\gamma}_s-\dot{\gamma}_0&=&\sum_{n=1}^N\frac{s^n}{n!}\gamma_0^{(n+1)}+\ell_{n+1}(s),
    \end{array}
    \end{equation*}
    and later on we will drop the dependence in $s$ in the remainding terms $g_{n+1}$,$\ell_{n+1}$ and $h_6$.
    The big O notation is understood to refer to the limit $\rho\to 0^+$.
    We also write:
     \begin{equation*}
     |\gamma_s-\gamma_0|^2=s^2-\tfrac{s^4}{12}|\ddot{\gamma}_0|^2+h_5(s),\ h_5=\underset{s\to 0}{\mathcal{O}}(s^5).
    \end{equation*}
    We have used: $\cip{\dot{\gamma}_0}{\ddot{\gamma}_0}=0$ and $|\ddot{\gamma}_0|^2=-\cip{\dot{\gamma}_0}{\gamma_0^{(3)}}$.
    We will also have to consider the expansion of $(1-X)^{-3/2}$, and for $k\ge 0$, we write
    \begin{equation}\label{eq:exp_k}
    	(1-X)^{-3/2}=1+\frac{3}{2}X-\frac{5}{8}X^2+\cdots+(-1)^{k}\frac{(2k+1)!}{4^{k}\,k! }X^{k}+F_{k+1}(X).
    \end{equation}
    
    \medskip
    
    \paragraph{\emph{Elementary equalities}}
    For an integer $n$ and $s\neq 0$, we have:
    \begin{equation*}
     \dfrac{1}{\delta^{n}}-\dfrac{1}{s^{n}}=-\dfrac{\rho^2 \sum_{k=0}^{n-1} s^{2k}\delta^{2(n-k)}}{s^{n}\delta^{n}(s^{n}+\delta^{n})},
    \end{equation*}
    so for $m_1\ge n$ and $m_2\ge 0$ we have:
    \begin{equation*}
    \begin{array}{rcl}
     I_\eps(n,m_1,m_2)&:=&\displaystyle\int_0^{\eps}s^{m_1}\rho^{m_2}\Big[\dfrac{1}{\delta^{n}}-\dfrac{1}{s^{n}}\Big]\d s\\
     	&=& -\displaystyle\int_0^{\eps}\frac{s^{m_1}\rho^{m_2}}{\delta^n s^n}\frac{\rho^2 \sum_{k=0}^{n-1} s^{2k}\delta^{2(n-k)}}{\delta^n+s^n}\d s
	\end{array}
    \end{equation*}
    Introducing $L_\eps(n,m_1,m_2):=\int_0^{\eps}\frac{s^{m_1}\rho^{m_2}}{\delta^n}\d s$, we obtain:
    \begin{equation}\label{eq:estim}
	\begin{array}{lcr}
	 L_\eps(n,m_1,m_2)=\mathcal{O}(\rho)&\mathrm{if}& m_2\ge 1\ \&\ m_1+m_2\ge n+1,\\
    	 I_\eps(n,m_1,m_2)=\mathcal{O}(\rho)&\mathrm{if}& m_1\ge n\ \&\ m_1+m_2\ge n+1.
    	 \end{array}
    \end{equation}

    We have other relevant integrals to compute. For $A>0$ (which will be $\eps/\rho$ in the proof) we have:
  \begin{equation*}
	\left\{
		\begin{array}{rcl}
			\dint_0^A\frac{\d s}{\langle s\rangle}&=&\mathrm{arcsh}(A),\\
								&=&\log(A)+\log(2)+\underset{A\to+\infty}{\mathcal{O}}(A^{-2}),\\
  			\dint_0^{A}\frac{\d s}{\langle s\rangle^3}&=&1+\underset{A\to+\infty}{\mathcal{O}}(A^{-2}).
		\end{array}
	\right.
  \end{equation*}

  In particular the following holds. For $n\ge 1$, we define 
  \[
  J_n(A):=\dint_0^A\frac{s^{2n}}{\langle s\rangle^{2n+1}}\d s,
  \] 
  an integration by parts yields
    $J_n(A)=-\frac{A^{2n-1}}{2(n-1)\langle A\rangle^{2(n-1)}}+J_{n-1}(A)$. Hence at fixed $n\ge 1$, we obtain:
  \begin{equation}\label{eq:compute_integ_3}
	\dint_0^A\frac{s^{2n}}{\langle s\rangle^{2n+1}}\d s=\log(A)+\log(2)-\frac{1}{2}\sum_{j=1}^{n-1}\frac{1}{j} +\underset{A\to+\infty}{\mathcal{O}}(A^{-2}).
  \end{equation}

 \medskip
    
    \paragraph{\emph{Expansion of the numerator}}
      A computation gives:
      \begin{multline}\label{eq:numerator}
       \dot{\gamma}_s\times(\bx-\gamma_s)=
       \rho \dot{\gamma}_0\times V+\tfrac{s^2}{2}\dot{\gamma}_0\times\ddot{\gamma}_0+\rho s\ddot{\gamma}_0\times V
	    +\tfrac{s^3}{3}\dot{\gamma}_0\times\gamma_0^{(3)}+\tfrac{\rho s^2}{2}\gamma_0^{(3)}\times V\\
	   +\tfrac{s^4}{12}\ddot{\gamma}_0\times\gamma_0^{(3)}+\ell_3\times(\bx-\gamma_s)-\dot{\gamma}_s\times g_4.
      \end{multline}
      
	 \medskip
	 
    \paragraph{\emph{Expansion of $|\bx-\gamma_s|^2$}}
      As $\cip{v}{\dot{\gamma}_0}=0$, we also have:
      \[
	\tfrac{|\bx-\gamma_s|^2}{\delta^2}=1-\tfrac{s^2\rho}{\delta^2}\cip{v}{\ddot{\gamma}_0}
	-\tfrac{s^4}{12\delta^2}|\ddot{\gamma}_0|^2-\tfrac{s^3\rho}{3\delta^2}\cip{v}{\gamma_0^{(3)}}
	+\tfrac{h_5}{\delta^2}-\tfrac{2\rho}{\delta^2}\cip{v}{g_4}.
      \]
	We introduce the notations:
	\begin{equation*}
		1-X(s,\bx):=\tfrac{|\bx-\gamma_s|^2}{\delta^2}.
	\end{equation*}
       \medskip
      
      \paragraph{\emph{Expansion of the integral}}
      We compute the expansion of the integral over $J^{(\eps)}$. We deal with the elements of the numerator \eqref{eq:numerator} one after the other starting
      with the ones of highest degrees in $s$ and $\rho$. For each of them, we use \eqref{eq:estim} to determine at which order we need to stop 
      the expansion of $|\bx-\gamma_s|^{-3}$ in powers of $X(s,\bx)$.
            
	 \medskip
	
	\subparagraph{\emph{The terms of order $s^4+\rho s^3$}}
	Using \eqref{eq:estim} for $n=3$ and $|\bx-\gamma_s|^{-3}=\mathcal{O}(\delta^{-3})$, we consider $k=0$ in \eqref{eq:exp_k} and obtain:
	\begin{multline*}
		\int_{-\eps}^{\eps}\frac{\d s}{|\bx-\gamma_s|^3}
			\big(\tfrac{s^4}{12}\ddot{\gamma}_0\times\gamma_0^{(3)}+\ell_3(s)\times(\bx-\gamma_s)-\dot{\gamma}_s\times g_4(s)\big)\\
			=\int_{-\eps}^{\eps}\frac{\d s}{|\gamma_0-\gamma_s|^3}
				\big(\tfrac{s^4}{12}\ddot{\gamma}_0\times\gamma_0^{(3)}+\ell_3(s)\times(\gamma_0-\gamma_s)-\dot{\gamma}_s\times g_4(s)\big)+\mathcal{O}(\rho).
	\end{multline*}
	
	 \medskip
	 
	 \subparagraph{\emph{The term $\tfrac{\rho s^2}{2}\gamma_0^{(3)}\times V$}}
	 We use \eqref{eq:estim} for $(n,m_1,m_2)$ equal to $(5,6,1)$ and $(5,4,2)$, which corresponds to the lowest orders of the term $\tfrac{\rho s^2}{\delta^3}X(s,\bx)$: 
	 so we only have to expand $(1-X)^{-3/2}$ at order $k=0$.
	 At $k=0$, we compute the integral corresponding to $(n,m_1,m_2)=(3,2,1)$ and obtain \eqref{eq:compute_integ_3}:
	 \begin{equation*}
	    \int_{-\eps}^{\eps}\frac{\rho s^2}{2|\bx-\gamma_s|^3}\gamma_0^{(3)}\times V\d s=-\rho\log(\rho)\gamma_0^{(3)}\times V+\mathcal{O}(\rho).
	 \end{equation*}

	\subparagraph{\emph{The term $\tfrac{s^3}{3}\dot{\gamma}_0\times\gamma_0^{(3)}$}}
	Because of the odd exponent of $s$, we expect some cancellation. Recall \eqref{eq:exp_k}. 
	By \eqref{eq:estim}, $(n,m_1,m_2)=(5,7,0),(5,5,1)$, for $k\ge 1$ we have:
	\[
	 \tfrac{1}{3}\dot{\gamma}_0\times\gamma_0^{(3)}\int_{-\eps}^{\eps}\tfrac{s^3}{\delta^3}F_1(X(s,\bx))\d s
	  =\tfrac{1}{3}\dot{\gamma}_0\times\gamma_0^{(3)}\int_{-\eps}^{\eps}\mathrm{sign}(s)F_1(X(s,\gamma_0))\d s+\mathcal{O}(\rho).
	\]
	For $k=0$, we have cancellation as $\int_{-\eps}^{\eps} \tfrac{s^3}{\delta^3}\d s=0$.

	\subparagraph{\emph{The term $\rho s\ddot{\gamma}_0\times V$}}
	Once again, the odd degree in $s$ gives rise to cancellation. By \eqref{eq:estim} with
	$(n,m_1,m_2)$ equal to $(7,5,3)$, $(7,9,1)$ (for $k=2$) and $(5,5,1)$, $(5,4,2)$ (for $k=1$),
	we only need to compute the integrals corresponding to $(3,1,1)$ (case $k=0$) and $(5,3,2)$ (case $k=1$).
	For these cases there holds exact cancellation:
	\[
	 \ddot{\gamma}_0\times V\int_{-\eps}^{\eps}\Big[\frac{\rho s}{\delta^3}+\frac{3\cip{\ddot{\gamma}_0}{V}}{2}\frac{\rho^2 s^3}{\delta^5}\Big]\d s=0,
	\]
	and thus we have:
	\[
	 \ddot{\gamma}_0\times V\int_{-\eps}^{\eps}\frac{\rho s}{|\bx-\gamma_s|^3}\d s=\mathcal{O}(\rho).
	\]

	\subparagraph{\emph{The term $\tfrac{s^2}{2}\dot{\gamma}_0\times\ddot{\gamma}_0$}}
	For $k=2$, we consider \eqref{eq:estim} with $(n,m_1,m_2)$ equal to $(7,6,2)$ and $(7,10,0)$. This shows:
	\[
	 \dot{\gamma}_0\times\ddot{\gamma}_0\int_{-\eps}^{\eps}\frac{s^2}{2\delta^3}F_2(X(s,\bx))\d s
	  =\dot{\gamma}_0\times\ddot{\gamma}_0\int_{-\eps}^{\eps}\frac{\d s}{2|s|}F_2(X(s,\gamma_0))+\mathcal{O}(\rho).
	\]
	For $k=1$, we consider \eqref{eq:estim} with $(n,m_1,m_2)$ equal to $(5,6,0)$ and $(5,5,1)$. This shows that we only have 
	to compute the cases $(5,4,1)$ (for $k=1$) and $(3,2,0)$ (for $k=0$).
	So we get:
	\begin{multline*}
	 \dot{\gamma}_0\times\ddot{\gamma}_0\int_{-\eps}^{\eps}\frac{s^2}{2\delta^3}(1-X(s,\bx))^{-3/2}\d s= 
	 \dot{\gamma}_0\times\ddot{\gamma}_0\int_{-\eps}^{\eps}\big((1-X(s,\gamma_0))^{-3/2}-1\big)\tfrac{\d s}{2|s|}\\
	  +\dot{\gamma}_0\times\ddot{\gamma}_0\dint_0^{\eps}\Big[\frac{s^2}{\delta^3}+\frac{3}{2}\cip{v}{\ddot{\gamma}_0}\frac{\rho s^4}{\delta^5}\Big]\d s+\mathcal{O}(\rho).
	 \end{multline*}
	Using \eqref{eq:compute_integ_3}, we obtain:
	\begin{multline*}
	 \dot{\gamma}_0\times\ddot{\gamma}_0\int_{-\eps}^{\eps}\frac{s^2}{2\delta^3}(1-X(s,\bx))^{-3/2}\d s= 
	 \dot{\gamma}_0\times\ddot{\gamma}_0\int_{-\eps}^{\eps}\big((1-X(s,\gamma_0))^{-3/2}-1\big)\tfrac{\d s}{2|s|}\\
	  +\dot{\gamma}_0\times\ddot{\gamma}_0\big(-\log(\rho)+\log(2\eps)-\tfrac{1}{2}-\tfrac{3}{2}\cip{v}{\ddot{\gamma}_0}\rho\log(\rho)\big)+\mathcal{O}(\rho). 
	\end{multline*}

	\subparagraph{\emph{The term $\rho \dot{\gamma}_0\times V$}}
	For $k=2$, we use \eqref{eq:estim} with $(n,m_1,m_2)=(7,8,1)$ and for $(7,4,3)$ we observe that:
	\[
	 \int_{-\eps}^{\eps}\frac{s^4\rho^3}{\delta^7}=2\rho \int_0^A\frac{s^4}{\langle s\rangle^7}\d s=\mathcal{O}(\rho).
	\]
	This shows that we have:
	\[
	 \rho \dot{\gamma}_0\times V\int_{-\eps}^{\eps}\frac{\d s}{\delta^3}F_2(X(s,\bx))=\mathcal{O}(\rho).
	\]
	For $k=1$. We use \eqref{eq:estim} with $(n,m_1,m_2)$ equal to $(5,5,1)$ and $(5,4,2)$. There remain the terms $(5,3,2)$, $(5,4,1)$ and $(5,2,2)$.
	The term $(5,3,2)$ vanishes by symmetry (the degree in $s$ is odd). For the term $(5,4,1)$, we have:
	\begin{align*}
	 \frac{3}{2}\rho \dot{\gamma}_0\times V\int_{-\eps}^{\eps}\frac{s^4}{12\delta^5}|\ddot{\gamma}_0|^2
		  &=\rho\frac{|\ddot{\gamma}_0|^2}{4}\dot{\gamma}_0\times V\int_0^A\frac{s^4}{\langle s\rangle^5}\d s\\
		  &=\frac{|\ddot{\gamma}_0|^2}{4}\dot{\gamma}_0\times V(-\rho\log(\rho))+\mathcal{O}(\rho).
	\end{align*}
	For the term $(5,3,2)$, we have:
	\begin{align*}
	  \frac{3}{2}\rho \dot{\gamma}_0\times V\int_{-\eps}^{\eps}\frac{s^2\rho}{\delta^5}\cip{V}{\ddot{\gamma}_0}\d s
		&=3\cip{V}{\ddot{\gamma}_0}\dot{\gamma}_0\times V\int_0^A\frac{s^2}{\langle s\rangle^5}\d s\\
		&=\cip{V}{\ddot{\gamma}_0}\dot{\gamma}_0\times V+\mathcal{O}(\rho^2).
	\end{align*}
	At last we deal with $k=0$. We have:
	\begin{equation*}
	 \rho \dot{\gamma}_0\times V\int_{-\eps}^{\eps}\frac{\d s}{\delta^3}=\frac{2}{\rho}\dot{\gamma}_0\times V\int_0^A\frac{\d s}{\langle s\rangle^3}\d s
		  =\frac{2}{\rho}\dot{\gamma}_0\times V+\mathcal{O}(\rho).
	\end{equation*}
	Thus we have:
	\begin{equation*}
	 \rho \dot{\gamma}_0\times V\int_{-\eps}^{\eps}\frac{\d s}{|\bx-\gamma_s|^3}=\frac{2}{\rho}\dot{\gamma}_0\times V+\cip{V}{\ddot{\gamma}_0}\dot{\gamma}_0\times V
		  -\rho\log(\rho)\frac{|\ddot{\gamma}_0|^2}{4}\dot{\gamma}_0\times V+\mathcal{O}(\rho).
	\end{equation*}
	\subparagraph{\emph{Conclusion for the integral}}
	  Writing
	  \begin{multline*}
	   A_{\eps}(0)-(\log(2)-\tfrac{1}{2})\dot{\gamma}_0\times\ddot{\gamma}_0
	   :=\mathrm{Pf}\,\int_{-\eps}^{\eps}\dot{\gamma}_s\times \frac{\gamma_0-\gamma_s}{|\gamma_0-\gamma_s|^3}\d s
	   =\log(\eps)\dot{\gamma}_0\times\ddot{\gamma}_0\\
	   +\dot{\gamma}_0\times\ddot{\gamma}_0\int_{-\eps}^{\eps}\big(\tfrac{1}{|\gamma_0-\gamma_s|^3}-\tfrac{1}{|s|^3}\big)\tfrac{s^2\d s}{2}
	    +\tfrac{1}{3}\dot{\gamma}_0\times\gamma_0^{(3)}\int_{-\eps}^{\eps}\tfrac{s^3}{|s|^3}F_1(X(s,\bx))\d s\\
	    +\int_{-\eps}^{\eps}\frac{\d s}{|\gamma_0-\gamma_s|^3}
				\big(\tfrac{s^4}{12}\ddot{\gamma}_0\times\gamma_0^{(3)}+\ell_3(s)\times(\gamma_0-\gamma_s)-\dot{\gamma}_s\times g_4(s)\big),
	  \end{multline*}
	  we obtain the expansion:
	  \begin{multline}\label{eq:part_J_eps}
	    \int_{-\eps}^{\eps}\dot{\gamma}_s\times\frac{\bx-\gamma_s}{|\bx-\gamma_s|^3}=\frac{2}{\rho}\dot{\gamma}_0\times V-\log(\rho)\dot{\gamma}_0\times\ddot{\gamma}_0
			      +\cip{V}{\ddot{\gamma}_0}\dot{\gamma}_0\times V+A_{\eps}(0)\\
			      -\rho\log(\rho)\big[\gamma_0^{(3)}\times V+\tfrac{1}{4}|\ddot{\gamma}_0|^2\dot{\gamma}_0\times V
				-\tfrac{3}{2}\cip{V}{\ddot{\gamma}_0}\dot{\gamma}_0\times\ddot{\gamma}_0\big]+\mathcal{O}(\rho).
	  \end{multline}

	From \eqref{eq:part_remainder} and \eqref{eq:part_J_eps}, we obtain the second equality of \eqref{eq:formule_behavior}.

\end{appendix}


\end{document}